%% file: nub15.tex
\documentclass[11pt]{article}
\usepackage{graphicx}
\usepackage{gensymb}
\usepackage{amsmath}
\usepackage{amssymb}
\usepackage{color}
\usepackage{tabularx}
\usepackage{comment}
\usepackage{url}
\usepackage{bm}

\usepackage{lipsum}
\usepackage{ragged2e}

\newcommand{\C}{\mathsf{C}}
\newcommand{\E}{\mathsf{E}}
\newcommand{\LC}{\mathsf{L}}
\newcommand{\cast}{($\bm{\ast}$)}

\newcommand{\qed}{\hfill\rule{2mm}{2mm}}                % QED (1)

\newtheorem{theorem}{Theorem}[section]
\newtheorem{corollary}[theorem]{Corollary}

\newtheorem{observation}[theorem]{Observation}
\newenvironment{proof}{{\noindent \em Proof:~}}{\hfill{\hfill$\Box$}}

\input{figbox.tex}
\def\figboxmargin{0.2in}
\def\figboxmargin{0.2in}
\begin{document}

\begin{titlepage}

\def\thepage{}

% Author macros::begin %%%%%%%%%%%%%%%%%%%%%%%%%%%%%%%%%%%%%%%%%%%%%%%%
\title{
   Improved Upper Bounds on the Growth Constants of Polyominoes and
   Polycubes\footnote{
      Work on this paper by the two authors has been supported
      in part by Grant~575/15 from the Israel Science Foundation (ISF).
   }
}

\date{}

\author{
   Gill Barequet\thanks{
      Center for Graphics and Geometric Computing,
      Computer Science Department,
      The Technion---Israel Institute of Technology,
      Haifa~3200003, Israel.
      E-mail: \texttt{{barequet}@cs.technion.ac.il}
   } \and
	Mira Shalah\thanks{
		Computer Science Department,
		Stanford University, CA.
		E-mail:
		\texttt{mira@cs.stanford.edu}
	}
}

\maketitle

% ----------------------------------------------------------------------------
% Abstract
% ========

\begin{abstract}
	A $d$-dimensional polycube is a facet-connected set of cells (cubes) on
	the $d$-dimensional cubical lattice~$\mathbb{Z}^d$.
	%2-dimensional polycubes are known as polyominoes.
	%The size of a polycube is the number of cubes it contains, and~
	Let~$A_d(n)$ denote the number of $d$-dimensional polycubes (distinct
	up to translations) with~$n$ cubes,
	%The growth constant~
	and~$\lambda_d$ denote the limit of the ratio~$A_d(n{+}1)/A_d(n)$
	as~$n \to \infty$.
	%Although it is known that~$\lambda_d$, the growth constant of~$A_d(n)$,
	% the sequence enumerating polycubes,
	%exists in any fixed dimension~$d$, its
	The exact value of~$\lambda_d$ is still unknown rigorously for any
	dimension~$d \geq 2$; the asymptotics of~$\lambda_d$, as~$d \to \infty$,
	also remained elusive as of today.
	%and setting rigorous lower and upper bounds on~$\lambda_d$
	%is an extremely challenging combinatorial problem.
	%While significant progress has been made
	In this paper, we revisit and extend the approach presented by Klarner
	and Rivest in~1973 to bound $A_2(n)$ from above.
	Our contributions are:
	\begin{itemize}
      \item Using available computing power, we prove
            that~$\lambda_2 \leq 4.5252$.  This is the first improvement of the
            upper bound on~$\lambda_2$ in almost half a century;
      \item We prove that~$\lambda_d \leq (2d-2)e+o(1)$ for any value
            of~$d \geq 2$, using a novel construction of a rational generating
            function which dominates that of the sequence~$\left(A_d(n)\right)$;
      \item For $d=3$, this provides a subtantial improvement of the upper bound
            on~$\lambda_3$ from~12.2071 to~9.8073;%~10.016;
      \item However, we implement an iterative process in three dimensions,
            which improves further the upper bound on~$\lambda_3$
            to~$9.3835$;
		%, while the value predicted for~$\lambda_3$
		%is about~8.3439.
	\end{itemize}
   %
   %along the years with improving the lower bound on~$\lambda_2$,
   %very little progress was made for the upper bounds.
   %The known rigorous upper bounds are quite far from the estimated values
   %of~$\lambda_d$, and the only known method for setting nontrivial upper
   %bounds was so far limited to two dimensions.
   %--- anything about $\lambda_d$? ---
\end{abstract}

\end{titlepage}

%-----comment-------%
%MATERIAL FOR JOURNAL
\begin{comment}
The popular computer game Tetris features polyominoes of size 4.
Polyominoes and polycubes play a fundamental role in the analysis of percolation processes and the collapse transition branched polymers undergo under heat.
polyominoes and polycubes are ...
Some bound are rigorously known,
and better bounds are provided without proof in the physics literature.
There is a rich amount of work on giving a lower bound for lambda.
On the other hand, only few attempts for improving the upper bound in two dimensions exist, and a straightforward generalization (see .. and sec..) applies to any dimension~$d$, yielding ..
In this work, we give a nontrivial generalization of .. to 3 and higher dimensions.
This improves the upper bound to any~$d$.
\end{comment}
% ----------------- end comment ----------------

% ----------------------------------------------------------------------------
% Introduction
% ============

\section{Introduction}

\figbox[l]{
	\includegraphics[scale=0.4]{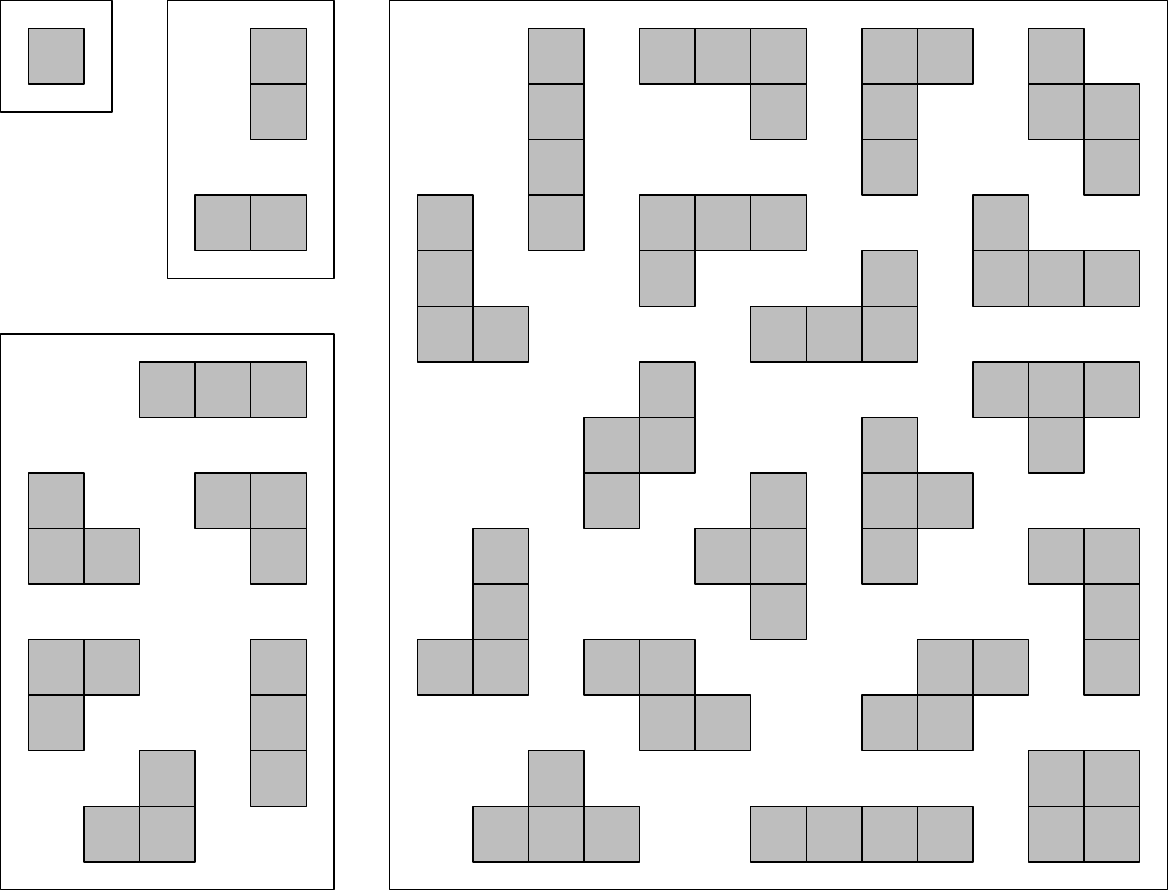}
}{figure}{fig:polyominoes}{Polyominoes of sizes $1 \leq n \leq 4$}
\emph{Polyominoes} are edge-connected sets of squares on the square lattice.
The size of a polyomino is the number of squares it contains.
All polyominoes of size up to~4 are shown in Figure~\ref{fig:polyominoes}.
%%% \begin{figure}
%%% 	\centering
%%% 	\includegraphics[scale=0.6]{fig-1}
%%% 	\caption{Polyominoes of sizes $1 \leq n \leq 4$}
%%% 	\label{fig:polyominoes}
%%% \end{figure}
%-----comment-------%
%MATERIAL FOR JOURNAL
\begin{comment}
Polyominoes are one of the most popular subjects in mathematical
recreations, and counting polyominoes by size is a fascinating
combinatorial problem.
\end{comment}
% ----------------- end comment ----------------
Likewise, \emph{polycubes} are facet-connected sets of
$d$-dimensional unit cubes, where connectivity is through
$(d{-}1)$-dimensional faces.
% ------------- begin comment --------------
% JOURNAL
\begin{comment}
To understand this term, consider, for example,
the dual graph of the hypercubic lattice.
The dual entities of a cube and an adjacency relation of two cubes are a
vertex and an edge (bond) connecting two vertices, respectively.
\end{comment}
% ------------- end comment --------------
%An animal is a connected subgraph of the dual graph of the lattice.
%When considering lattice animals simply as connected subsets of cells on
% a lattice, there are infinitely-many lattice animals of any size.
Two \emph{fixed} polycubes are considered identical
if one can be obtained by a translation of the other.
In this paper, we consider only fixed polycubes.
%and so we omit this adjective in the sequel.
Polyominoes and polycubes are specific types of \emph{lattice animals}, the
term used in the statistical physics literature to refer to
connected sets of cells on \textit{any} lattice.
%Consider the dual graph of the square or hypercubic lattice.
%The dual entities of a square or cube and an adjacency relation of two cubes are a cell and an edge (bond) connecting two cells, respectively.
%An ``animal'' is a connected subgraph (cluster) of cells.

% Polycubes have been used as mathematical models to study a broad range of important applications in chemistry, physics, and biology. % studied by mathematicians and natural scientists for many years and are the subject of an extensive literature.
%Polyominoes .
The fundamental combinatorial problem concerning polycubes is
``How many polycubes with~$n$ cubes are there?''
This problem originated in parallel in the theory of
\emph{percolation}~\cite{GSR76,SG76}, in the analysis of chemical
graphs~\cite{DH83,FGSW94,LI79,MSWMSFG90},
and in the graph-theoretic treatment of cell-growth problems~\cite{Kl67}
more than half a century ago.
% ------------- begin comment --------------
\begin{comment}
Temperley~\cite{Te56} investigated the mechanics of macro-molecules,
and Broadbent and Hammersley~\cite{BH57} studied percolation processes.
At about the same time, Harary~\cite{Ha60} composed a list of unsolved
problems in the enumeration of graphs,
and Eden~\cite{Ed61} analyzed cell growth processes.
Thus, aside from being a pure mathematical problem, enumeration of lattice animals has a broad range of applications to many other fundamental problems.
\end{comment}
% original ref here was {Ha67}
%cryptography~\cite{SE05}.
% ------------- end comment --------------
% As a result, the problem
%and has been studied by mathematicians and natural scientists.
However, despite much research in those areas,
most of what is known relies primarily on heuristics and empirical studies,
%by both mathematicians and natural scientists, %over the last five decades,
and very little is known rigorously, %about this problem,
even for the low-dimensional lattices in~$2 \leq d \leq 4$ dimensions.
%since then in the statistical physics and mathematics communities.
%However,
% and despite serious efforts over the last five decades in both the statistical physics and mathematics communities,
% it remains far from solved.
% Even
The simplest instance of counting polyominoes
%not much is known,
%and even much less is known for higher dimensions.
%Counting polyominoes
is considered
one of the fundamental open problems in combinatorial geometry~\cite{opp}.

Let~$A_d(n)$
(sequence~A001168 in
the On-line Encyclopedia of Integer Sequences~\cite{oeis})
denote the number of polycubes of size~$n$.
%Thus, for example, $A_2(1){=}1$, $A_2(2){=}2$, $A_2(3){=}6$, and
%$A_2(4){=}19$.
Since no analytic formula for~$A_d(n)$ is known for any dimension $d{>}1$,
many researchers have focused on efficient algorithms
for \emph{counting} polycubes by size, primarily the square lattice.
These methods are based on either explicitly enumerating all polycubes
(e.g., by an efficient back-tracking algorithm~\cite{ML92,Re81}), or on
implicit enumeration (e.g., by a transfer-matrix algorithm~\cite{Co95,Je03}).
The sequence~$A_2(n)$ has been determined so far up to $n=56$~\cite{Je03}.
Enumerating polycubes in higher dimensions is an even more elusive problem.
% ------------- begin comment --------------
\begin{comment}
Lunnon~\cite{Lu72} manually counted 3-dimensional polycubes (up to size~6)
by considering symmetry groups.
In a subsequent work, Lunnon~\cite{Lu75} computed the number of small-sized
polycubes in up to~6 dimensions.
\end{comment}
% ------------- end comment --------------
Most notably,
Aleksandrowicz and Barequet~\cite{AB09a,AB09b} extended polycube counting by efficiently generalizing Redelmeier's algorithm~\cite{Re81} to higher
dimensions.
The statistical-physics literature provides extensive enumeration data of polycubes~\cite{GSR76,Ma90,MSWMSFG90,Ga80}, the most comprehensive being
by Luther and Mertens~\cite{LM11}, in particular, listing~$A_3(n)$ up
to~$n{=}19$.

One key fact that holds in all dimensions was discovered in~1967 by
Klarner~\cite{Kl67}, showing that the limit
$\lambda_2 := \lim_{n \to \infty} \sqrt[n]{A_2(n)}$ exists.
This is a straightforward consequence of the fact that the
sequence~$(\log A_2(n))$ is supper-additive,
i.e.,~$A_2(n) A_2(m) \leq A_2(n+m)$.
% -------------- begin comment ------------
\begin{comment}
The proof uses a standard concatenation argument,
in which pairs of polyominoes of size~$n$ are uniquely concatenated to
produce a subset of all polyominoes of size~$2n$.
\end{comment}
% -------------- end comment ------------
Since then, $\lambda_2$ has been called ``Klarner's constant.''
Only in~1999, Madras~\cite{Ma99} proved
%, using a novel pattern-frequency argument,
%a stronger statement about
the existence of the
asymptotic growth rate, namely, $\lim_{n \to \infty} A_2(n+1) / A_2(n)$,
which clearly equals~$\lambda_2$.
Klerner's and Madras's results hold, in fact, in any dimension.

A great deal of attention has been given to estimating the values
of~$\lambda_d$, especially for~$d{=}2,3$.
Their exact values are not known and
have remained elusive for many years.
% Elementary considerations show that $\lambda_d \in (d, 2^{2d})$
Based on interpolation methods, applied to the known values of the
sequences~$(A_2(n))$ and~$(A_3(n))$,
it is estimated (without a rigorous proof),
that~$\lambda_2 \approx 4.06$~\cite{Je03}
and~$\lambda_3 \approx 8.34$~\cite{Gu09}.
%Setting good bounds on these constants is a hard and challenging problem in
%enumerative combinatorics.
%In particular,
There have been several attempts to bound~$\lambda_2$ from
below, with significant progress over the
years~\cite{BMRR06,BRS16,Ed61,Kl65,Kl67,RW81,Re62},
but almost nothing is known for higher dimensions.
For~$d=2$, it has been proven that $\lambda_2 \geq 4.0025$~\cite{BRS16}.
For~$d>2$, the only known way to set a lower bound on~$\lambda_d$ is by
using the fact~\cite{Kl67} that
$\lambda_d = \lim_{n\to \infty}\sqrt[n]{A_d(n)} =
   \sup_{n \ge 1} \sqrt[n]{d A_d(n)}$.
In particular, for $d{=}3$, the
value~$A_3(19) {=} 651,459,315,795,897$ yields the lower
bound~$\lambda_3 \ge \sqrt[19]{3 A_3(19)} \approx 6.3795$,
which is quite far from the best estimate of~$\lambda_3$ mentioned above.

On the other hand, only one procedure (Eden~\cite{Ed61}) is known
for bounding~$\lambda_d$
from above.\footnote{Another method for bounding~$\lambda_2$ from above
   was presented elsewhere~\cite{BB15}, but G.~Rote (personal
   communication) discovered an error in the computations thereof.
   Fixing this error raised the obtained upper bound above the known
   bound~\cite{KR73}.
}
\begin{comment}
Eden~\cite{Ed61} was the first to give an upper bound on~$\lambda_2$
by showing that~$A_2(n){\leq}\binom{3n}{n-1}$, implying that
$
   \lambda_2 \leq \lim_{n \to \infty} \sqrt[n]{\binom{3n}{n-1}} = 6.75.
$
Barequet et al.~\cite{BBR10} generalized Eden's method to show that
\[
   \lambda_d \leq \lim_{n \to \infty} \sqrt[n]{\binom{(2d-1)n}{n-1}}.
\]
Since $\lim_{n \to \infty} \sqrt[n]{\binom{(2d-1)n}{n-1}} \leq \frac{(2d-1)^{2d-1}}{(2d-2)^{2d-2}}$, ...
\end{comment}
This procedure (explained in detail in the next section)
shows that~$\lambda_2 \leq 6.7500$, $\lambda_3 \leq 12.2071$ and
that~$\lambda_d \leq (2d-1)e$.
It was shown in~\cite{BBR10}
that~$\lambda_d \sim 2ed-o(d)$ as~$d$ tends to infinity,
%based on a pattern of so-called \emph{diagonal formulae} (see reference~\cite{BS17}),
and conjectured (based on an unproven assumption)
% commonly used in the literature of statistical physics),
that~$\lambda_d$ is asymptotically equal
to~$(2d{-}3)e{+}O(1/d)$.

%In~$d>2$ dimensions,
%no procedure has been yet known for improving upon this upper bound.
As we detail in the next section, Klarner and Rivest~\cite{KR73} enhanced
Eden's method by using a more sophisticated system of ``twigs,'' proving
that~$\lambda_2 \leq 4.6496$.
In this paper, we extend this enhancement to higher dimensions, and show that
it results in the two-variable rational generating function
\[
   g^{(d)}(x,y) = \sum_{n,m=0}^{\infty} l_d(n,m)x^ny^m =
      \sum_{n=0}^{\infty} x y^n \left((1+x)^{2(d-1)} + x^2\right)^n =
      \frac{x}{1-y\left((1+x)^{2(d-1)}+x^2\right)},
\]
whose diagonal function~$\sum_{n=0}^{\infty} l_d(n,n) x^n y^n$
generates a sequence which dominates the sequence $(A_d(n))$.
Using elementary calculus, Klarner and Rivest~\cite{KR73} proved that
for~$d=2$,  $l_2(n,n) \leq 4.8285^n$.
Similarly, we prove that for~$d=3$, we have that~$l_3(n,n) \leq 9.8073^n$,
giving the first nontrivial upper bound on~$\lambda_3$.
Finally, we prove that~$\lambda_d \leq (2d-2)e+1/(2d-2)$,
by proving that $l_d(n,n) \leq \left((2d-2)e+1/(2d-2)\right)^n$.
To the best of our knowledge, this is the first generalization of Klarner and Rivest's method to higher dimensions, and
the first algorithmic approach for improving the upper bound
on~$\lambda_d$ for any value of~$d > 2$.
An important result of this enhancement for dimensions~$d \geq 3$ is an
improved upper bound on the space complexity required to encode polycubes,
namely, a polycube of size~$n$ can be encoded with~$O(2d-2)n)$ bits.

Lastly, we revisit the computer-assisted approach that Klarner and Rivest~\cite{KR73} used to further improve the upper bound on~$\lambda_2$ to $4.6495$.
We are not aware of any published attempt to reproduce their result.
%let alone carry out this computation beyond~$i=10$.
With the computing resources currently available to us, we
%computed~$\sigma_{i}$ up to~$i = 21$, and
improve the upper bound
on~$\lambda_2$ to~4.5252.
We also extend the approach to $d=3$, and prove
%TODO 2
that~$\lambda_3 \leq 9.3835$.

% -------------- begin comment ------------
\begin{comment}
This is a remarkable improvement of the difference between the known upper
bound on~$\lambda_3$ to the estimated bound (8.344) from~3.863 % 12.207-8.344
to~0.534. % 8.878-8.344

[[[REWRITE FOR THE JOURNAL VERSION]]]
The paper is organized as follows.
In the second part of the chapter[?] we describe our effort to improve the upper bound on~$\lambda_2$,
using Klarner and Rivest's approach~\cite{KR73}.
To achieve these aims we exploited to the maximum possible computer resources which were available to us.
\end{comment}
% -------------- end comment ------------

% ----------------------------------------------------------------------------
% Previous Works
% ==============

\section{Previous Works by Eden, and Klarner and Rivest}
\label{sec:prev}
\figbox[l]{
	{{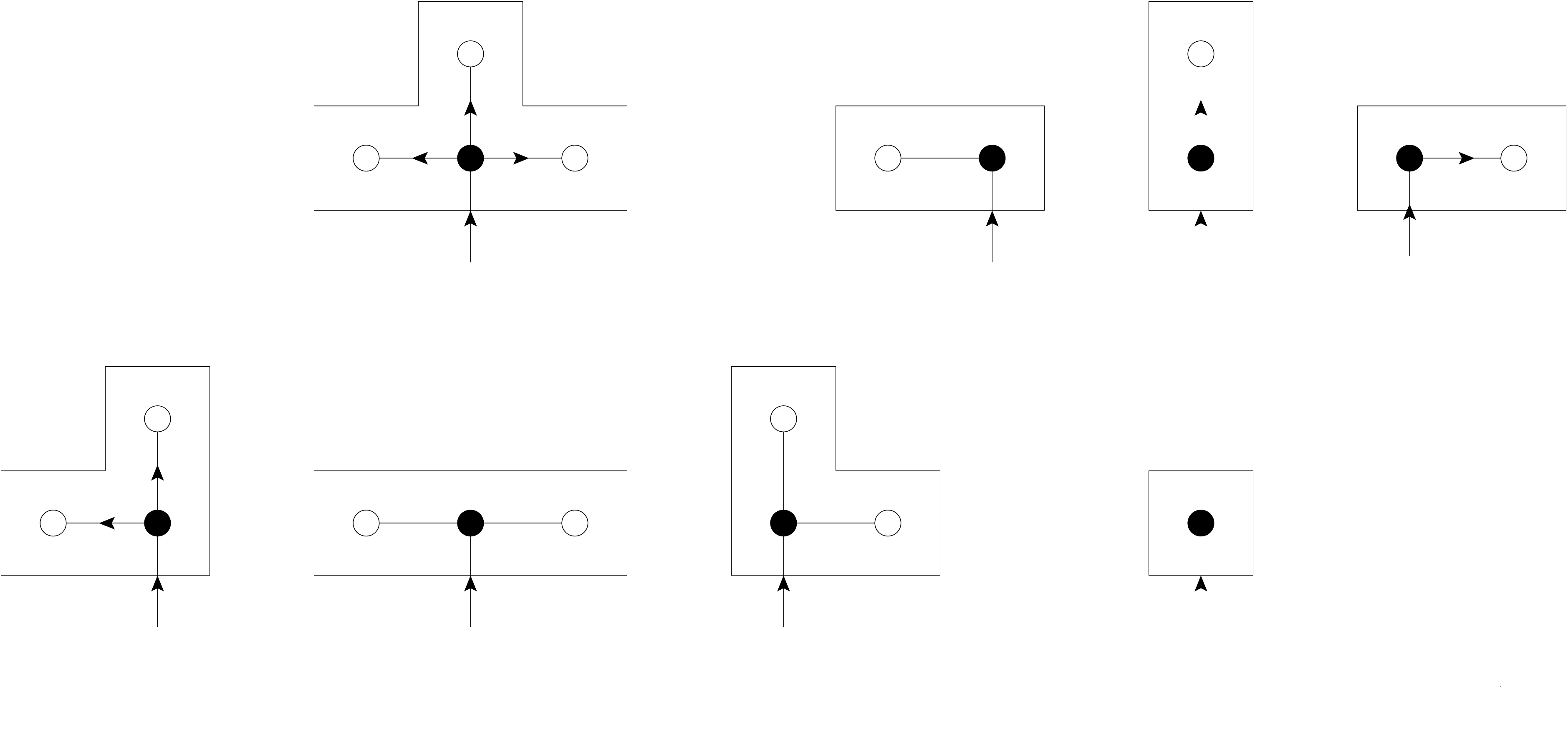}}
}{figure}{fig:eden-twigs}{Eden's set of twigs~\protect{\cite[Figure~3]{KR73}}}
For two $d$-dimensional cubes with centers~$c_1=(x_1,\dots,x_d)$
and~$c_2=(y_1,\dots,y_d)$, we say that~$c_1$ is
\emph{lexicographically smaller} than~$c_2$ if~$x_i < y_i$ for the first
value of~$i$ where they differ.
Let~$P$ be an $n$-cell polycube in~$d$ dimensions.
$P$ can be uniquely encoded with a binary string~$W_P$ of
length $(2d{-}1)n{-}1$~\cite{BBR10,Ed61}, as follows.
Let~$G$ be the cell-adjacency graph of~$P$.  (The vertices of~$G$ are the
centers of the cubes in~$P$, and two vertices are connected by an edge if
their corresponding cubes are adjacent.)
Perform a breadth-first search on~$G$,
starting at cell~1 (the lexicographically smallest cell of~$P$).
%where the lexicographic order is defined by the coordinates of the centers of the
%cubes).
In the course of this procedure, every cell~$c \in P$ is reached through an
incoming edge~$e$ since~$P$ is connected.
(An imaginary edge incoming into cell~1 is fixed so as
to supposedly originate from a cell that cannot belong to~$P$.)
Clearly, the cell~$c$ is connected by edges of~$G$ to at most~$2d{-}1$ additional neighboring cells.
The procedure now traverses all these outgoing edges according to a
fixed order determined by their orientations relative to~$e$.
(For example, for polyominoes, the outgoing edges are traversed
according to their clockwise order relative to~$e$.)
Then, if such an edge leads to a cell of~$P$ which has not been labeled yet,
this cell is assigned the next unused number, and we
update~$W_P := W_P \cdot 1$, where ``$\cdot$'' is the concatenation
operator.
Otherwise, if the cell does not belong to~$P$, or it is already assigned a
number, we set~$W_P := W_P \cdot 0$.
Since each cell can be assigned a number only once,
this procedure maps polycubes in a one-to-one manner into binary sequences
with $n{-}1$ ones and $(2d{-}2)n$ zeros.
Hence, using Stirling's formula,
\begin{equation}
   A_d(n) \leq \binom{(2d-1)(n-1)}{n-1} \leq
      \left(\frac{(2d-1)^{2d-1}}{(2d-2)^{2d-2}}\right)^n.
   \label{upper_bound}
\end{equation}

% -----------------------------
\begin{comment}
This construction was presented by Eden~\cite{Ed61} for polyominoes, and generalized by Barequet et al.~\cite{BBR10} to $d$-dimensional polycubes.
In fact, this idea can be applied to any lattice in which every cell has
up to some constant number of neighbors.
\end{comment}
% -----------------------------
For polyominoes, this procedure is equivalent to assigning an element
of~$\E=\{e_1,\dots,e_8\}$ (Figure~\ref{fig:eden-twigs}) to each square
of~$P$ (in the same order).
%%% \begin{figure}
%%%	    \centering
%%%	    \scalebox{0.24}{\input{twigs.pdf_t}}
%%%    \caption{Eden's set of twigs~\cite[Figure~3]{KR73}}
%%%    \label{fig:eden-twigs}
%%% \end{figure}
%Eden's method~\cite{Ed61} is based on the observation
%polyomino~$P$ can be associated with a unique spanning tree,~$T_P$,
% embedded in the dual graph of the square lattice. $T_P$ can then be interpreted as a sequence of ``twigs''.
%a collection of small subtrees.
%The mapping from~$P$ to a sequence of twigs is done by assigning an element
%of~$\E$ to each square of~$P$, such that all squares are included and cycles
%are excluded.
%See Figure~\ref{fig:spanning-tree-L}(a) for an example. %of a spanning tree
%created by this process for a polyomino of size~10.
%It can be shown that different polyominoes are associated with
%different sequences of elements of~$\E$.
%The number of trees spanning polyominoes (and, hence, the number of
%polyominoes) is thus bounded from above by the number of such sequences (some of which cannot be realized by polyominoes).
% ---------------------------
\begin{comment}
This process can be viewed as a mapping~$f(P){=}W_P$,
which maps a polyomino~$P$ of size~$n$
to a binary string~$W_P$ of length $3n{-}1$~\cite{Ed61}.
\end{comment}
% ---------------------------

%Then, we take the $n$th root of both sides, let~$n$ tend to infinity, and
%obtain an upper bound on~$\lambda_d$.
In three dimensions, one obtains that~$\lambda_3 \leq 5^5/4^4 \leq 12.2071$.
In general, since
$\frac{(2d-1)^{2d-1}}{(2d-2)^{2d-2}} =
    (2d-1)\left(1+\frac{1}{2d-2}\right)^{2d-2}<(2d{-}1)e$,
it follows that~$\lambda_d \leq (2d{-}1)e$.
(In fact, a more thorough analysis of the last relation shows
that~$\lambda_d \leq (2d{-}1.5)e$.)  %[[[[[SHOW HOW!!]]]]]
(The latter value is also known as the ``Bethe approximation''
of~$\lambda_d$; see Gaunt et al.~\cite[p.~1904, Eq.~3.9]{GSR76}, and
Gaunt and Peard~\cite[p.~7523, Eq.~4.9]{GP00}.)

%\subsection{The set of twigs $\LC$}
\figbox[l]{
	\begin{tabular}{cccc}
      \scalebox{0.35}{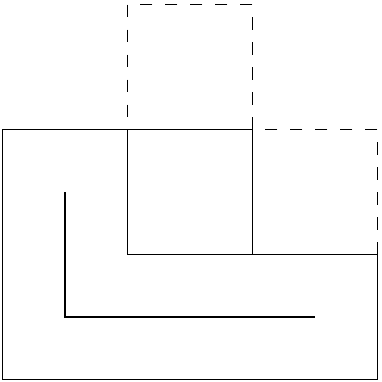} &
         \scalebox{0.35}{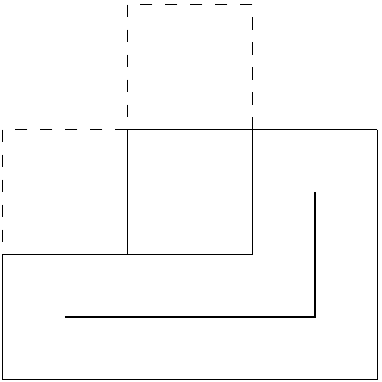} &
         \scalebox{0.35}{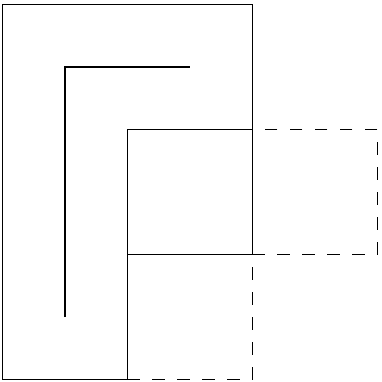} &
         \scalebox{0.35}{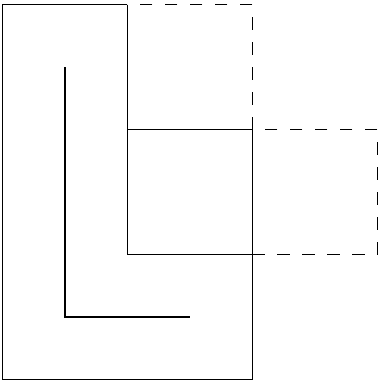} \\
      (a) & (b) & (c) & (d) \medskip \\
      \scalebox{0.35}{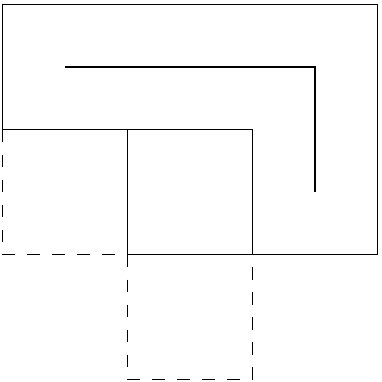} &
         \scalebox{0.35}{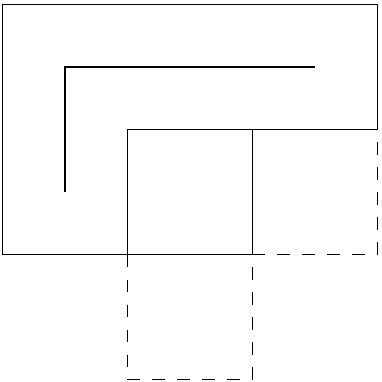} &
         \scalebox{0.35}{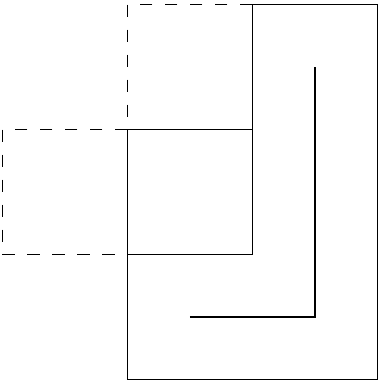} &
         \scalebox{0.35}{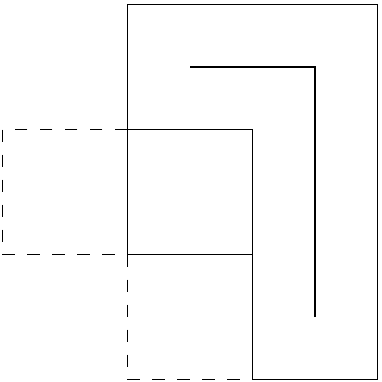} \\
      (e) & (f) & (g) & (h)
	\end{tabular}
}{figure}{fig:Lcontexts}{$\LC$-contexts
                         of~\protect{\cite[Figure~6]{KR73}}}
Now refer to Figure~\ref{fig:Lcontexts}.
%%% \begin{figure}
%%%    \centering
%%%    \begin{tabular}{cccc}
%%%       \scalebox{0.5}{\input{LC1.pdf_t}} &
%%%          \scalebox{0.5}{\input{LC2.pdf_t}} &
%%%          \scalebox{0.5}{\input{LC3.pdf_t}} &
%%%          \scalebox{0.5}{\input{LC4.pdf_t}} \\
%%%       (a) & (b) & (c) & (d) \medskip \\
%%%       \scalebox{0.5}{\input{LC5.pdf_t}} &
%%%          \scalebox{0.5}{\input{LC6.pdf_t}} &
%%%          \scalebox{0.5}{\input{LC7.pdf_t}} &
%%%          \scalebox{0.5}{\input{LC1.pdf_t}} \\
%%%       (e) & (f) & (g) & (h)
%%% 	\end{tabular}
%%% 	\caption{The eight $\LC$-contexts of cell~$u$~\cite[Figure~6]{KR73}}
%%% 	\label{fig:Lcontexts}
%%% \end{figure}
Around any square~$u$ on the square lattice, there are eight $\LC$-shaped
4-sets of squares, called the ``$\LC$-contexts'' of~$u$~\cite{KR73}.
Let us use the term \emph{status} of a cell to refer to whether or not the
cell belongs to a given polyomino.
In their beautiful paper~\cite{KR73},
Klarner and Rivest designed a set of \emph{twigs}~$\LC$
(see Figure~\ref{fig:Ltwigs}),
%%% \begin{figure}
%%%    \centering
%%%    \scalebox{0.6}{\input{Ltwigs.pdf_t}}
%%%    \caption{The set of twigs $\LC$~\cite[Figure~7]{KR73}}
%%%    \label{fig:Ltwigs}
%%% \end{figure}
which is more compact than~$\E$.
They showed that similarly to $e_1,\dots,e_8$, %(in Figure~\ref{fig:eden-twigs}),
$L_1,\dots,L_5$ %(in Figure~\ref{fig:Ltwigs})
can serve as building
blocks for polyominoes:
Every $n$-cell polyomino~$P$ corresponds to a unique $n$-term sequence of
elements of~$\LC$,
however, not every such sequence represents a polyomino, immediately
implying that
$\lambda_2 < |L|=5$, which is already a substantial improvement over
%Eden's bound of
6.75.
\figbox[l]{
   \scalebox{0.35}{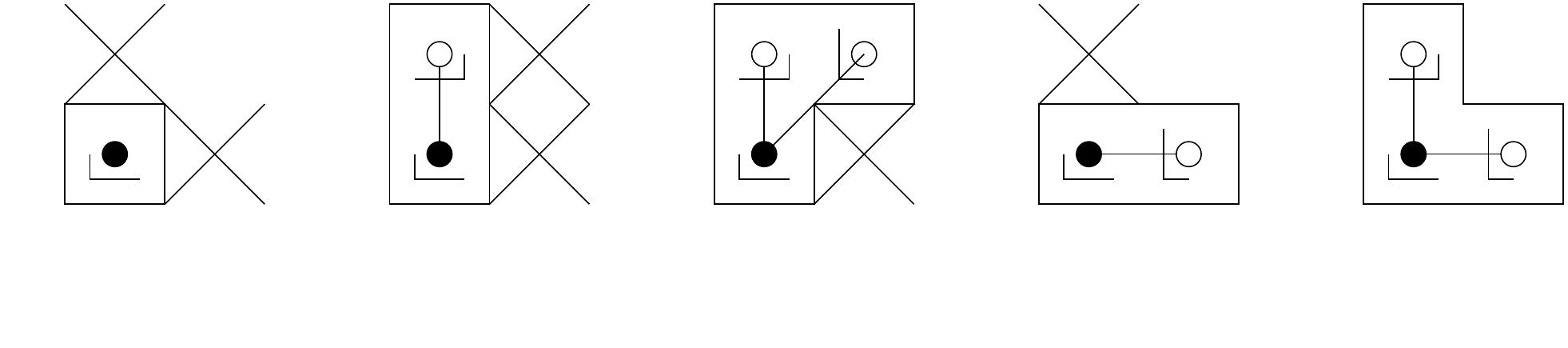}
}{figure}{fig:Ltwigs}{Twig set $\LC$~\protect{\cite[Figure~7]{KR73}}}
The key idea behind the design of~$\LC$ is that one can perform a
breadth-first search on the cell-adjacency graph of a polyomino, such that
at each step, each cell can be assigned one of the eight $\LC$-contexts
shown in Figure~\ref{fig:Lcontexts}, such that
the statuses of \emph{all} cells in this $\LC$-context are already encoded.
Therefore, while~$\E$ %Eden's encoding
encodes all~$2^3 = 8$ possible status configurations of \emph{three} neighbors
at every step of the traversal,
%Klarner and Rivest's twigs
$\LC$ encodes the statuses of just \emph{two} neighbors at each step.
% Let us detail the set~$\LC$.
%(and thus $|\LC| < |\E|$).
The sequence~$\LC_p$, encoding a polyomino~$P$, can be constructed
algorithmically as follows.
Assume, w.l.o.g., that in the lexicographic order defined over the cells
of~$P$, the cells are ordered first according to their $y$-coordinate, and
secondly according to their $x$-coordinate.
Maintain a queue (initially empty) of \emph{white} (or \emph{open})
cells.
White cells are cells that have not yet been visited by the algorithm.
Start from the lexicographically-smallest cell of~$P$, and put it in the queue.
The $\LC$-context of this cell is the one shown in Figure~\ref{fig:Lcontexts}(a), since, by definition, the cells in this $\LC$-context of the cell
do not to belong to~$P$.
The addition of twigs to the configuration~$T$ constructed so far
proceeds as follows until the queue is empty.
%no living cells are left.
% addition of twigs
%to the configuration~$T$ constructed so far
%no living cells are left.
Let~$u$ be the oldest cell in the queue.
Remove~$u$ from the queue.
Let~$a,b$ denote the
cells connected to~$u$ (as shown in Figure~\ref{fig:Lcontexts}), $c$~($\neq u$) denote the cell connected to~$a,b$,
and~$\ell$ denote the last label assigned to a cell of~$P$
(initially~$\ell=0$).
Refer to Figure~\ref{fig:Ltwigs}.
The twig~$L$ assigned to~$u$ is~$L_1$ if~$a,b,c \notin P$; % (this is indicated with an~$\mathsf{X}$);
%are \emph{forbidden} and they will never become cells of the subpolyomino of
%which the twig is a root.);
$L_2$ if $b,c \notin P$ and~$a \in P$;
$L_3$ if $b \notin P$
and~$a,c \in P$;
$L_4$ if~$a \notin P$ and $b \in P$; or $L_5$
if~$a,b \in P$.
%TODO
A new configuration $T*L$ is then constructed as follows:
\begin{enumerate}
	\itemsep0em
	\item The root cell of~$L$ (shown in black in Figure~\ref{fig:Ltwigs}) is
	placed over~$u$, such that the
	orientation ($\LC$-context) of~$L$ and~$u$ coincide (this may
	require a reflection and/or rotation of~$L$).
	\item The white cell, where~$L$ has been added, is turned black (\emph{dead}).	
	%The addition is legal if no other cells of~$L$ overlap with any
	%part of~$T$ and no cell of~$L$ occupies a forbidden cell.
	\item The cells of~$L$ marked with an~$\mathsf{X}$ (called \emph{forbidden})
	become forbidden in $T$, namely, they will never become cells of $T$.
	\item The white cells of~$L$ are added (in their indicated order) to the
	queue.
	Note that each white cell has an assigned $\LC$-context (indicated by the
    shape~$\LC$), and that the statuses of all cells in this~$\LC$-context are
    known (and thus, already encoded).
	%of living cells of~$T$.
	\item The white cell(s) of~$L$ are assigned the label~$\ell{+}1$
	(or~$\ell{+}1$ and~$\ell{+}2$, in order).
	%\item The oldest cell is removed from the list of living cells.
\end{enumerate}

%, and different polyominoes are assigned different sequences.
%The cells of a twig are divided into two types,
%namely ``dead'' (closed) and ``living'' (open),
%colored black and white, respectively.
% Every twig contains at least one dead cell,
% but not necessarily any living cells.
%Each element of~$\LC$ is composed of a root cell (colored black) with a specified $\LC$-context, which is also marked with asterisks,
%and a (possibly empty) set of open cells (colored white),
%which is ordered linearly (see~$L_3$ and~$L_5$).
Note that when $a\in P$ and $b\notin P$,
%for the construction to work,
it is necessary to encode %information (empty or occupied)
whether or not~$c$ is in~$P$ (twigs~$L_3$ and~$L_2$, resp.), so that
when~$a$ is inserted to the queue,
the statuses of all cells in its
indicated~$\LC$-context are encoded.
%Hence, we need two twigs to encode the configurations:~$a$ is occupied $b$ is not.
Note also that the linear order of the white cells in~$L_3$ and~$L_5$ is necessary % in order
to ensure the uniqueness of the construction.
%, as will be explained below.
%Consider the $\LC$-context of the root cell of all the twigs
%$L_1, \dots, L_5$, the single open cells of~$L_2$ and~$L_4$, and the first
%open cells of~$L_3$ and~$L_5$.  They all have the property that the status
%of the cells in their indicated $\LC$-context is known.
For the second white cell in either~$L_3$ or~$L_5$,
the statuses of all cells in its indicated $\LC$-context are known only
after assigning a twig to the first open cell.
% Moreover, their mapping from a polyomino to a spanning tree ensures that this holds throughout the construction for all the cells of the polyomino, not just the smallest cell.
% TODO 3
%\label{subsec:map}
An example of this process is shown in
Fig.~\ref{fig:spanning-tree-L}(b):
Let~$u_i$ be the cell visited at step~$i$,
and~$a_i,b_i$ be the cells connected to~$u_i$
(Figure~\ref{fig:Lcontexts}).
The cell~$u_1$
is assigned the label~1 and the twig~$L_5$ since~$a_1,b_1 \in P$.
In order to assign a twig~$L_i \in \LC$ to~$u_2$,
$L_i$ is rotated by~$90\degree$ and reflected around the~$x$ axis.
Note that a cell can be discovered and labeled only once.
Therefore, $u_5$ is assigned the twig~$L_1$ (not~$L_4$) since the cell on
its right was already discovered at the fourth iteration
and assigned the label~8.
\figbox[r]{
   %\begin{tabular}{ccc}
    %  \scalebox{0.4}{
    %  \parbox{.45\linewidth}{
      %   \vspace{0.2in}
       %      \includegraphics[scale=0.5]{spanning_tree.pdf}
       %     }
      %}       & &
      \scalebox{0.7}{
      \parbox{.6\linewidth}{
          \includegraphics[scale=0.5]{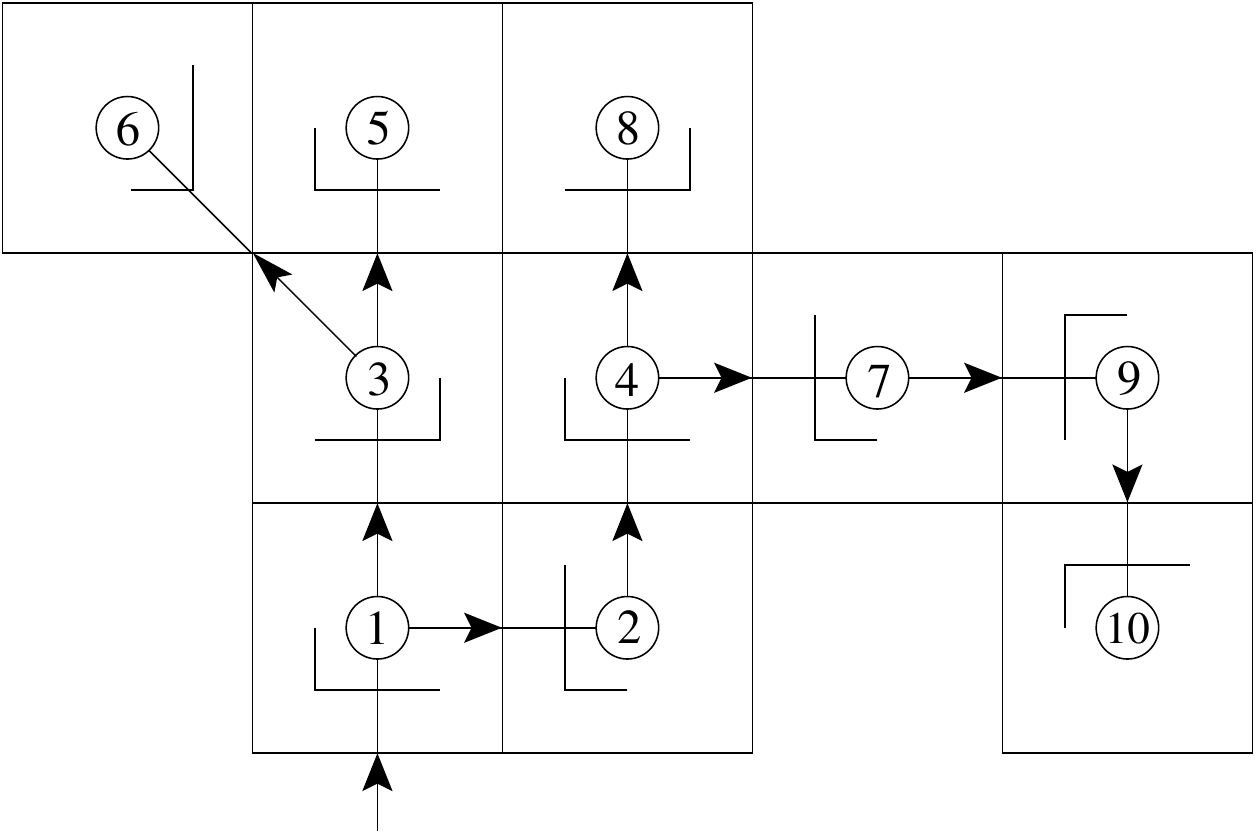}}
     }     %\\
      %(a) & & \hspace{-29mm} (b)
  % \end{tabular}
}{figure}{fig:spanning-tree-L}
{
      %\scriptsize
      %(a)~\protect{\cite[Figure~4]{KR73}}~A spanning tree generated by Eden's method;
      %(b)~
      \protect{\cite[Figure~8]{KR73}} $L_P=(L_5,L_4,L_3,L_5,L_1,L_1,L_2,L_1,L_4,L_1)$.
}

\begin{observation}
	\vspace{-8mm}
	\label{ob:diff}
	Two different polyominoes are encoded with different sequences of twigs.
\end{observation}

%% This use of knowledge about the ``past'' (the $\LC$-context),
%% to constrain the number of possibilities for the ``future'' is the
%% ingenuity behind Klarner and Rivest's idea.
With this construction, Klarner and Rivest reduced the number of building blocks of
polyominoes from~8 to~5.
 It is natural to ask if a more compact set of building blocks exists.
 We can now answer this question on the negative, since the existence of
 such set of size~4
%% We now know that no smaller set of building blocks exists, since the
%% existence of such a set
would imply that~$\lambda_2 \leq 4$, while we
already know that~$\lambda_2 \geq 4.0025$~\cite{BRS16}.

\subsection{Mathematical Formulation}

The upper bound~$\lambda_2 \leq 5$ can be improved by a more delicate analysis, which assigns different ``weights'' to  different elements of~$\LC$, as follows.
Each twig~$L \in \LC$ is assigned a weight $w(L) = x^a y^b$,
where~$a$ denotes the number of cells in~$L$ minus~1,
and~$b$ denotes the number of black cells in~$L$.
(Thus,~$w(L_1) = y$, $w(L_2) = xy$, $w(L_3) = x^2y$, $w(L_4) = xy$,
and~$w(L_5) = x^2y$.)
The \emph{weight} of a sequence
$S=(\ell_1,\dots,\ell_k) \in \LC^k$ is defined as
$W(S) = x \cdot w(\ell_1) \cdot \ldots \cdot w(\ell_k)$,
and the weight of the empty sequence is defined to be~$x$.

Let~$P$ be a polyomino of size~$n$, and let~$L_p=\{\ell_1,\dots,\ell_n\}\in \LC^n$ denote the sequence of twigs encoding~$P$.
For each $\ell_i \in \LC$, we clearly have that $w(\ell_i)=x^{a_i}y$, such that
$a_i\in \{0,1,2\}$ equals the number of open cells in $\ell_i$.
Thus, $w(L_p)=x\cdot x^{\sum_{i=1}^na_i}y^n$.
Moreover, $\sum_{i=1}^n a_i = n{-}1$ because each cell of~$P$
(other than the smallest cell) becomes open only once,
and is thus accounted for by some~$a_j$ in the sum.
The smallest cell is accounted for by the term~$x$ in~$w(L_p)$.
\begin{corollary}
	\label{cor:xnyn}
	$w(L_p)=x^n y^n$.
\end{corollary}

Now, let~$\LC^k$ denote the set of all sequences of~$k \geq 0$ elements
of~$\LC$.
The sum of weights of all finite sequences of elements of~$\LC$ is
\begin{equation}
   \sum_{k=0}^{\infty}{\sum_{S \in \LC^k} W(S)} =
      \sum_{k=0}^{\infty}x\left( \sum_{\ell\in \LC}w(\ell)\right)^k =
      x\left(1-\sum_{\ell \in \LC}w(\ell)\right)^{-1}.
   \label{eq:sum-of-weights}
\end{equation}

Clearly,
%for Eden's set of twigs $\E$,
%$\sum_{e \in \LC} w(e)=y(1+x)^3$,
$\sum_{\ell \in \LC} w(\ell) = y(2x^2+2x+1)$,
thus, the generating function given by~(\ref{eq:sum-of-weights}) is
\begin{equation}
   \sum_{m,n=0}^{\infty} l(m,n) x^m y^n =
      \frac{x}{1-y(2x^2+2x+1)} =
      \sum_{n=0}^{\infty} x y^n (2x^2+2x+1)^n,
   \label{eq:gen-func}
\end{equation}
where~$l(m,n)$ is the coefficient of the term~$x^m y^n$.
% ---------------------- maybe for the journal version --------------
% Eden described an injection of the set of polyominoes of size $n$ into the set of finite sequences of $\E$ having weight $x^n y^n$.
% At vertex $i$ ($i=1,\dots,n$) of $T_P$,
% we always find exactly one of the twigs shown in Figure~\ref{fig:eden-twigs}.
%Let % us denote this twig by $e_i$,
% and define
%$T_P=(\ell_1, \dots, \ell_n)$.
% The sequence of twigs in this example is $e=(e_4,e_4,e_8,e_3,e_6,e_8,e_8,e_6,e_7,e_8)$,
% and its weight is $w(e)=xw(e_1)\dots w(e_{10}) = x\cdot x^2y \cdot x^2y \cdot y \cdot x^2y \cdot xy \cdot y \cdot y \cdot xy \cdot xy \cdot y=%\bm
% {x^{10}y^{10}}$.
% ------------------------------------------------------------------

%removed 'by klarner and rivest'
\begin{comment}
As observed~\cite{KR73},
two different polyominoes are associated with different sequences of twigs,
and every polyomino of size~$n$ is associated with a sequence with
weight~$x^n y^n$.
Thus,
\end{comment}
%as is well known,

\begin{comment}
The coefficient~$l(n,n)$ of~$x^n y^n$ in
Equation~(\ref{eq:gen-func}) is an upper bound on~$A_2(n)$,
and if~$r$ is the radius of convergence of the ``diagonal'' power
series~$\sum_{n=0}^{\infty} l(n,n)z^n$, then~$\lambda \leq 1/r$.
Using a simple calculation, Klarner and Rivest showed that~$\lambda_2 \leq 4.8285$.
\end{comment}
By Observation~\ref{ob:diff} and Corollary~\ref{cor:xnyn}, there is an injection of the set of polyominoes of size $n$ into the set of finite sequences of elements of $\LC$ having weight $x^n y^n$. Thus,
the coefficient~$l(n,n)$ of~$x^n y^n$ in
Equation~(\ref{eq:gen-func}) is an upper bound on~$A_2(n)$.
In Section~\ref{sec:d2}, we follow the proof of Klarner and Rivest~\cite{KR73}
for showing that~$l(n,n) \leq 4.8285^n$.
In the next two sections, we extend the concept of ``$\LC$-context'' to higher
dimensions.
%%, and show how this leads to an improved upper bound on $\lambda_d$.
%In~$d$ dimensions (fixed~$d$), we obtain a set of~$2^{2(d-1)}+1$ twigs,
%which provides an upper bound on~$\lambda_d$ better than the
%bound obtained by the generalization of Eden's method.

\section{Higher Dimensions}

\subsection{Three Dimensions}
\label{sec:3d}

%\begin{theorem}
%   $\lambda_3 \leq 8.8782$. \qed
%\end{theorem}
\figbox{
    \scalebox{0.2}{{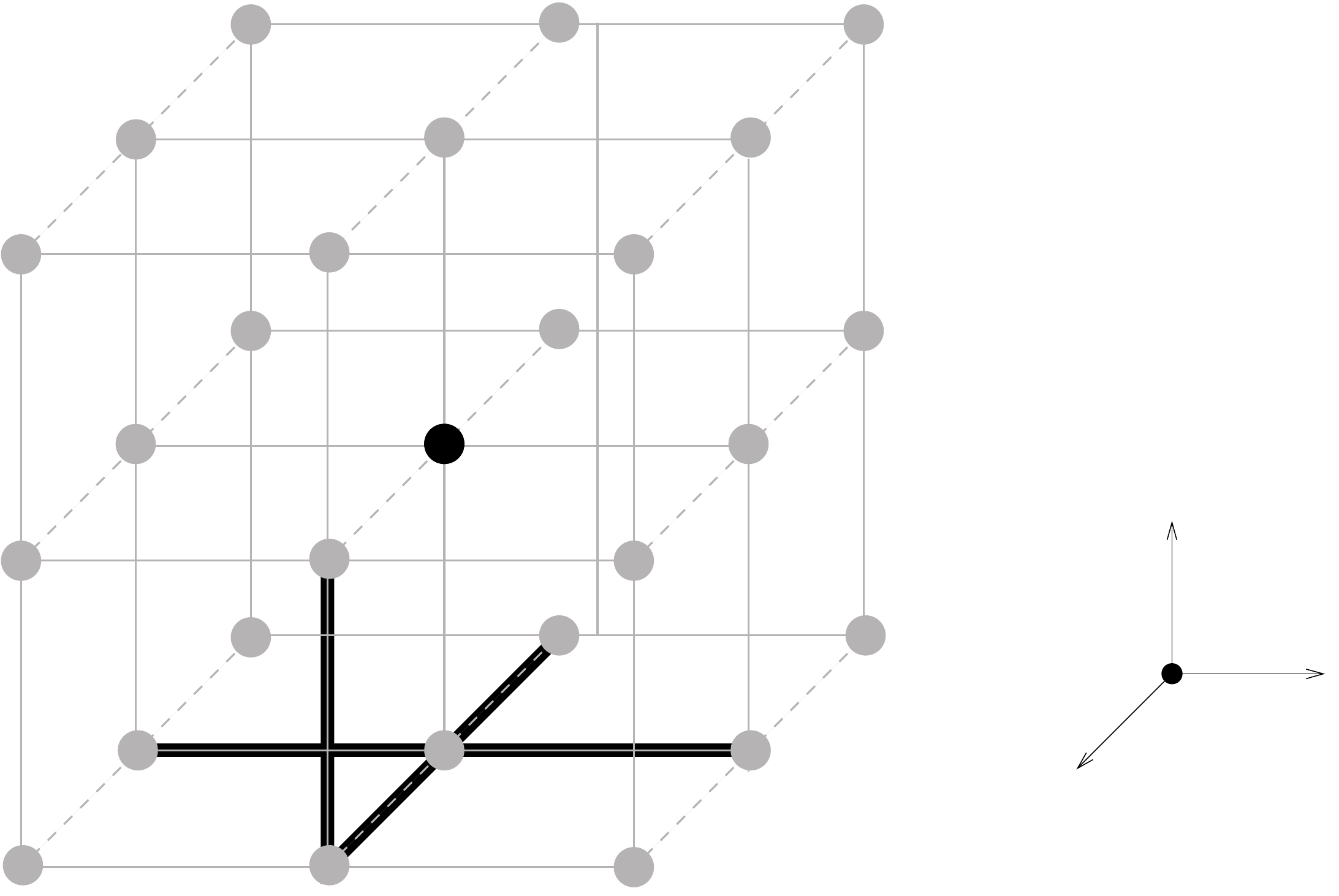}}
}{figure}{fig:Lcontext-3D}{The $+\LC$ context
            (bold black lines) of a cell~$o$ on the
            3D cubical lattice}
We generalize the twigs idea to three dimensions as follows.
Refer to Figure~\ref{fig:Lcontext-3D}.
%%% \begin{figure}
%%%    \centering
%%%    \scalebox{0.17}{\input{Lcontext.pdf_t}}
%%%    \caption{The $+\LC$ context
%%%             (shown in black) of a cell~$o$ on the
%%%             3-dimensional cubical lattice}
%%%    \label{fig:Lcontext-3D}
%%% \end{figure}
\figbox{
	\begin{tabular}{cccc}
		\includegraphics[scale=0.1]{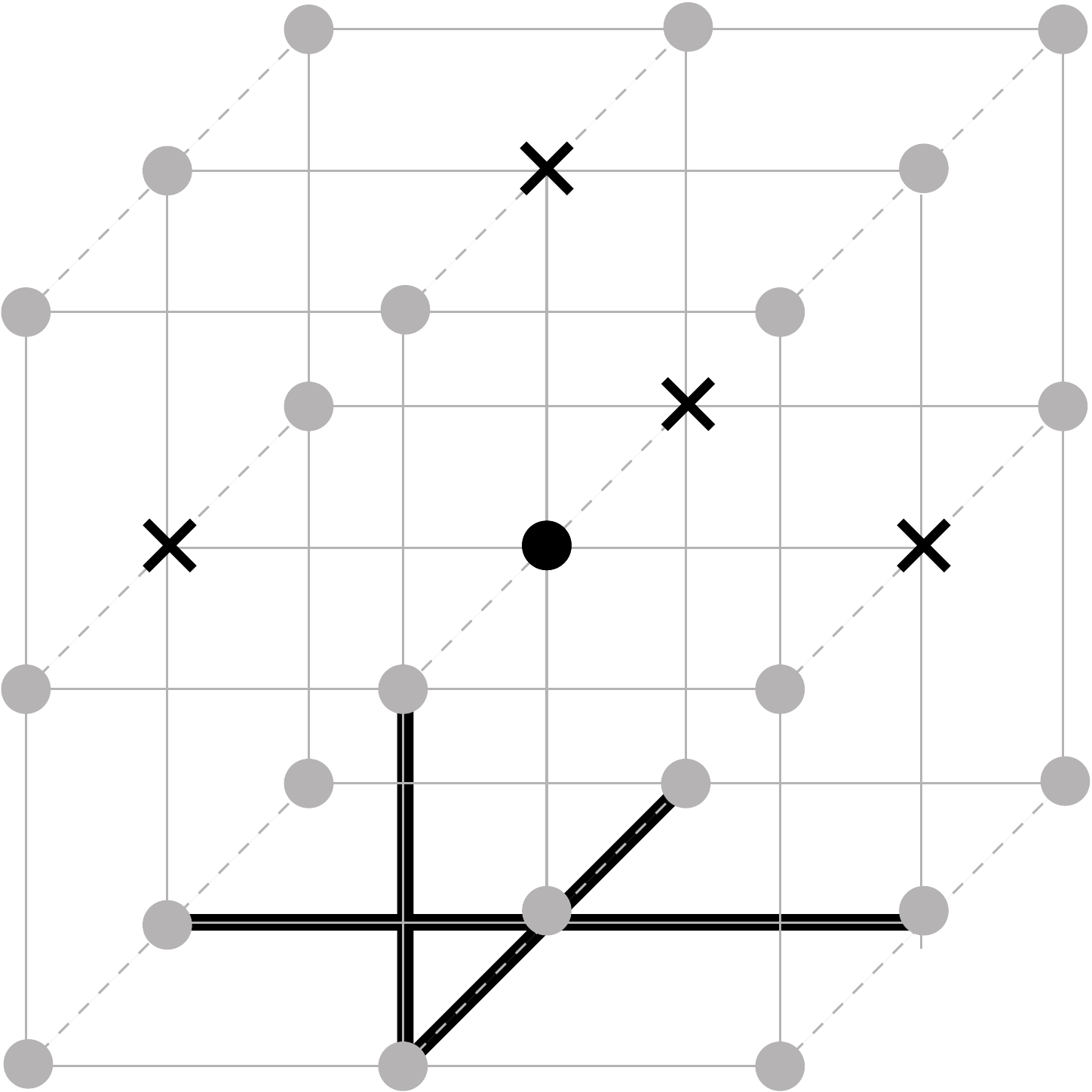} &
		\includegraphics[scale=0.1]{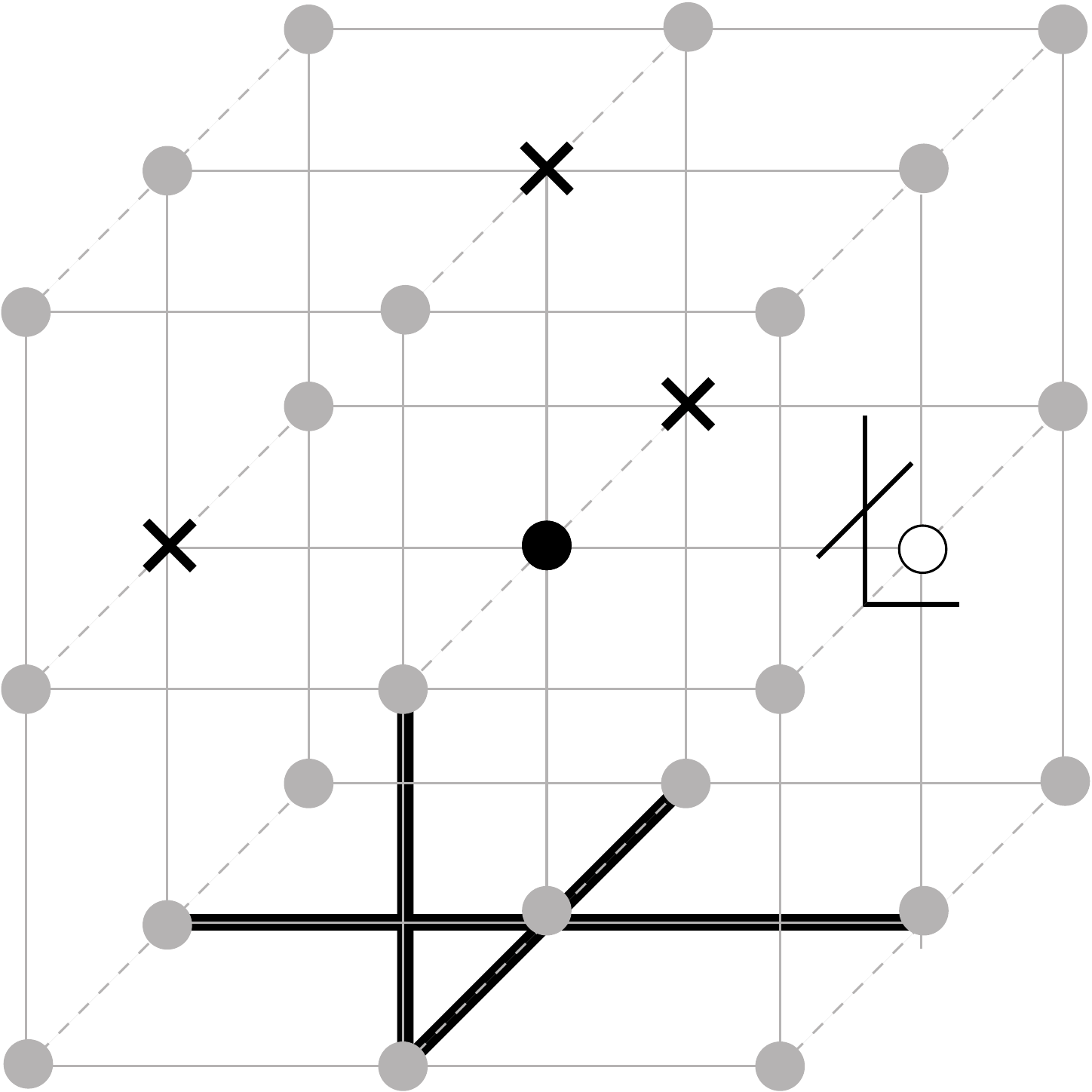} &
		\includegraphics[scale=0.1]{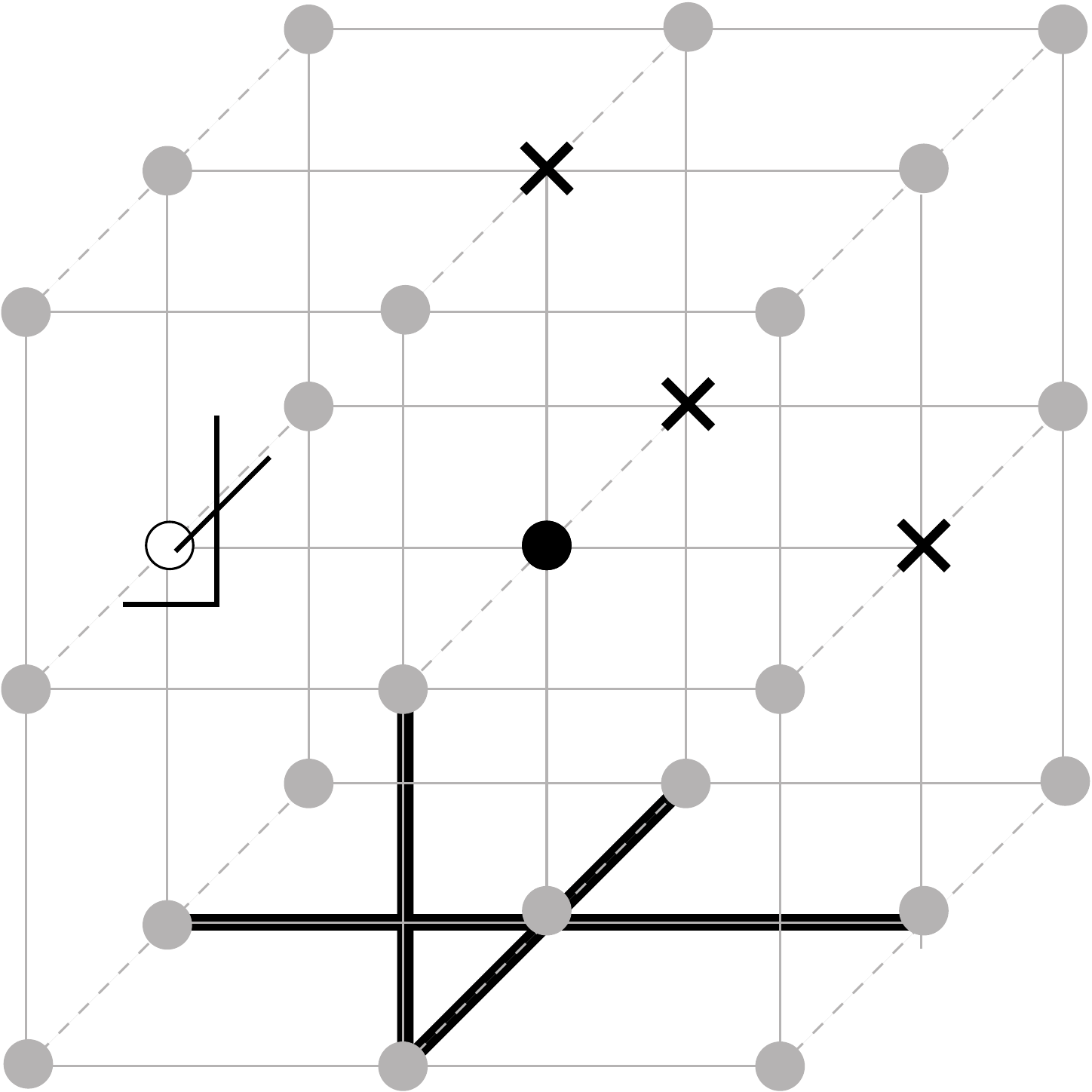} &
		\includegraphics[scale=0.1]{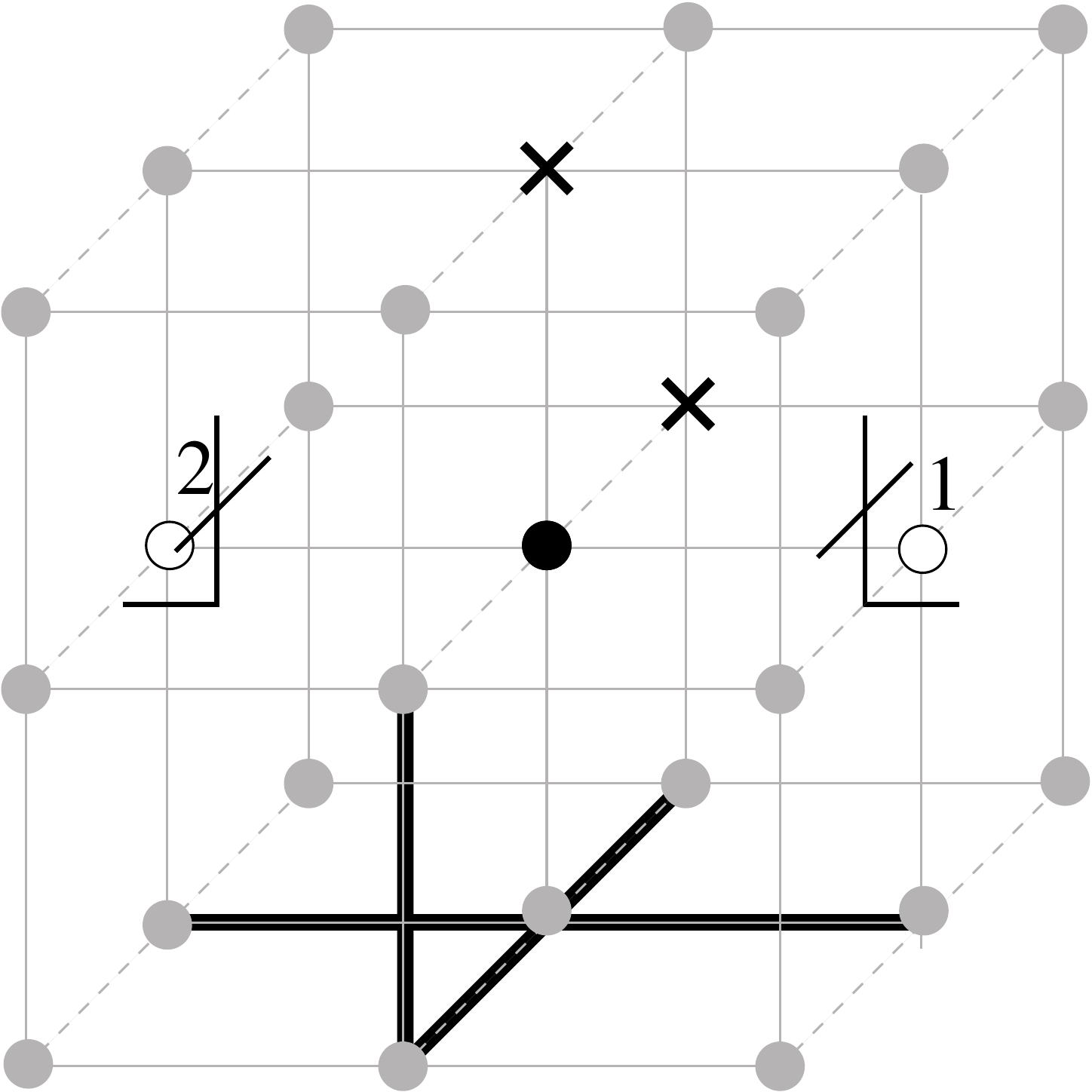} \\
		$T_1$ & $T_2$ & $T_3$ & $T_4$  \medskip \\
		\includegraphics[scale=0.1]{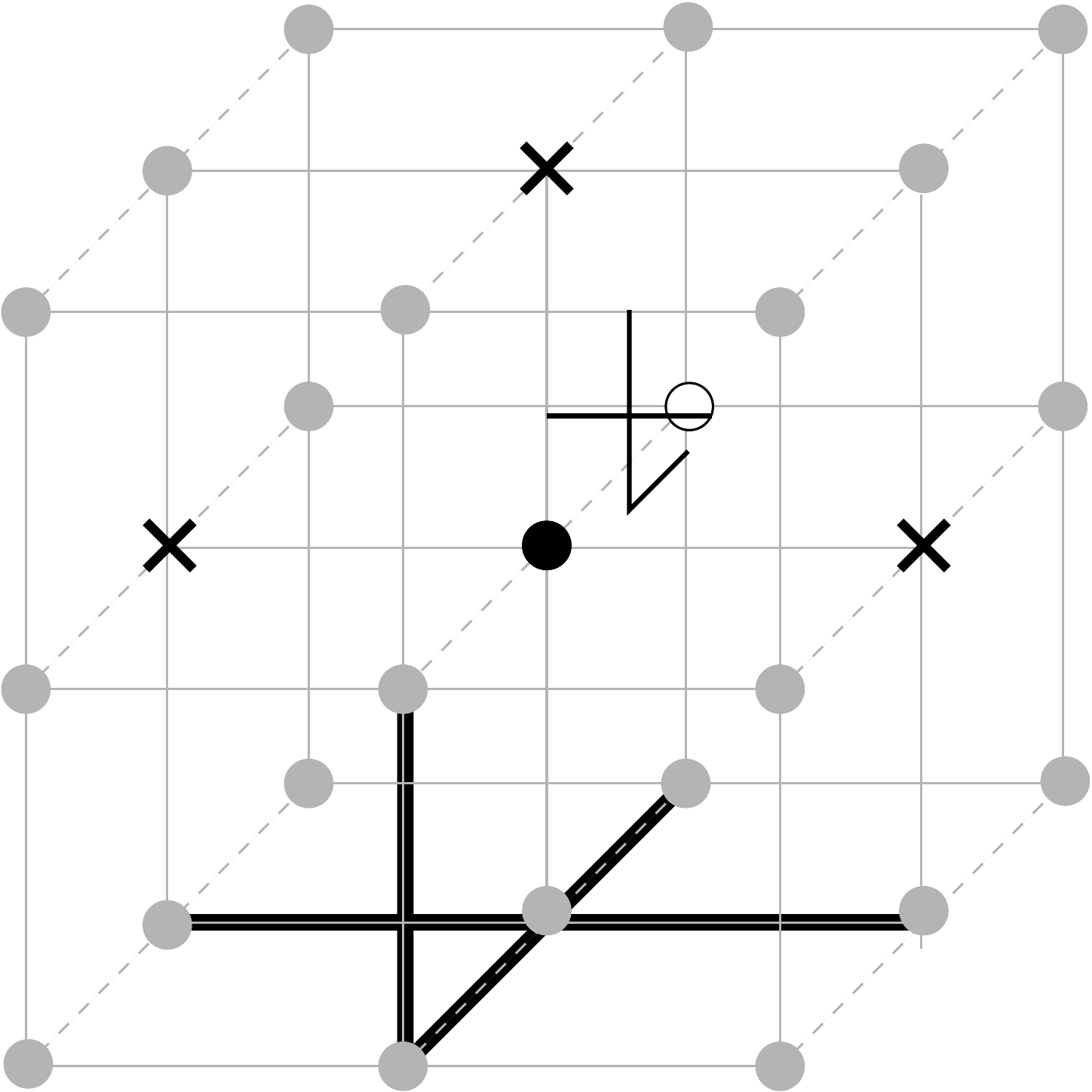} &
		\includegraphics[scale=0.1]{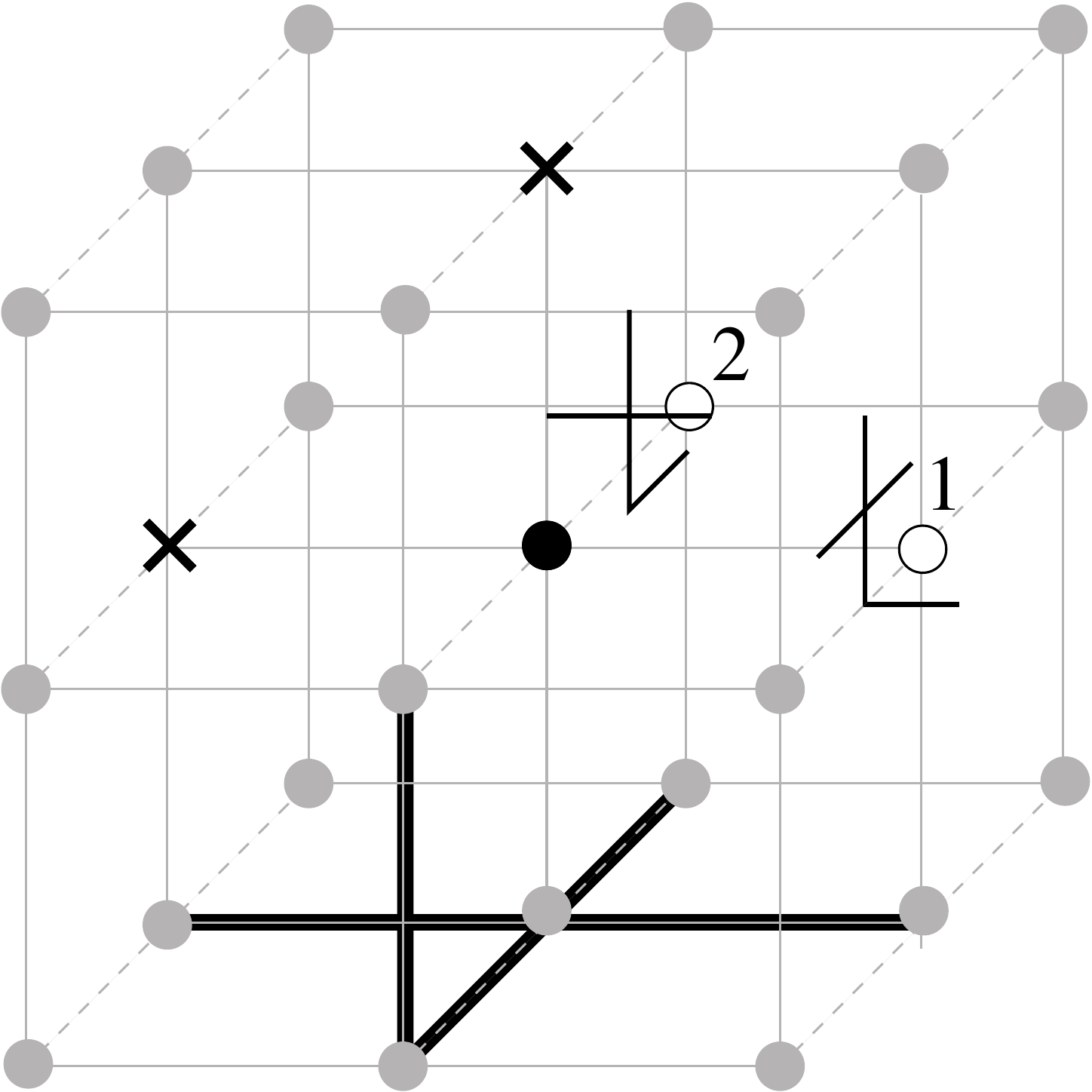} &
		\includegraphics[scale=0.1]{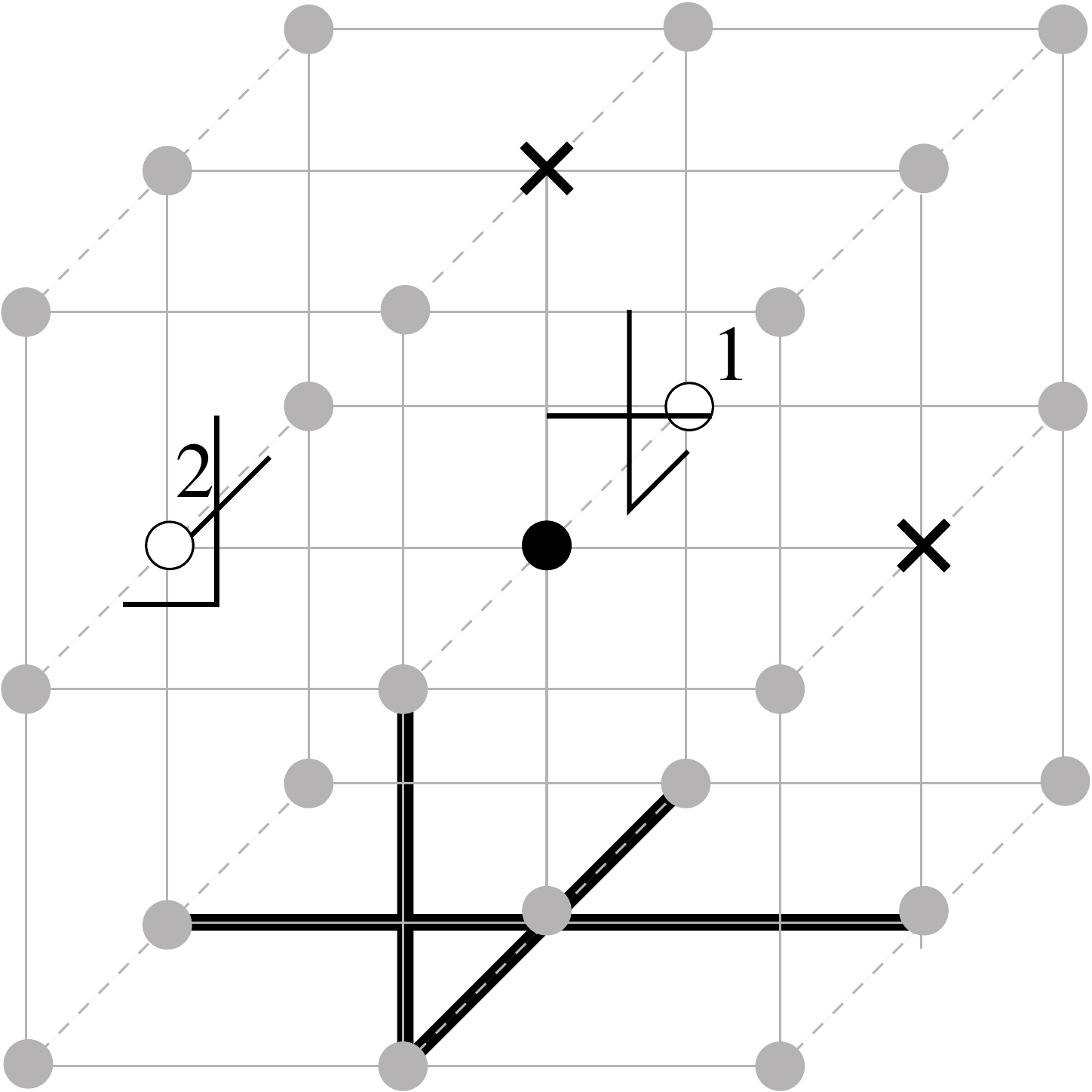} &
		\includegraphics[scale=0.1]{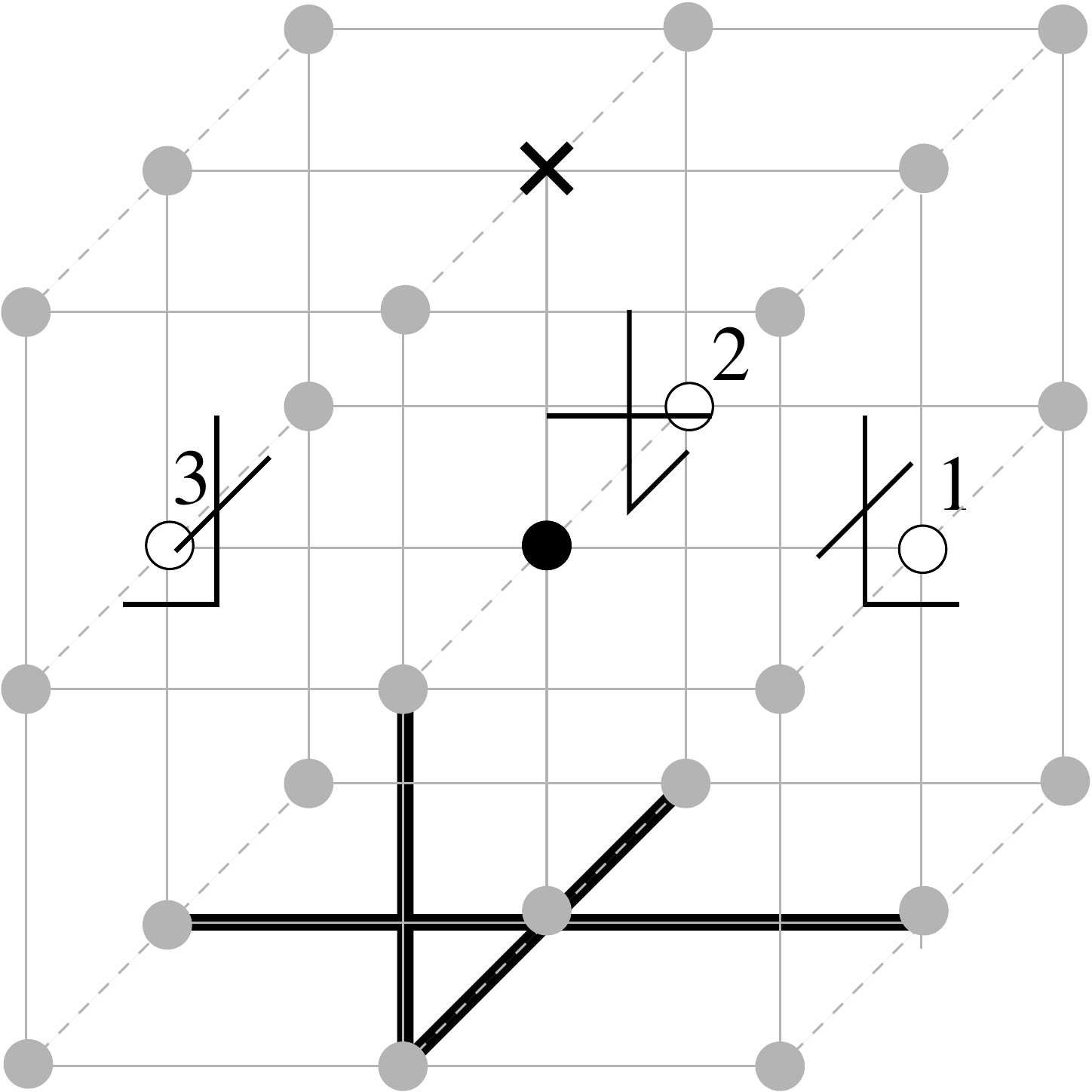} \\
		$T_5$ & $T_6$ & $T_7$ & $T_8$ \medskip \\
		\includegraphics[scale=0.11]{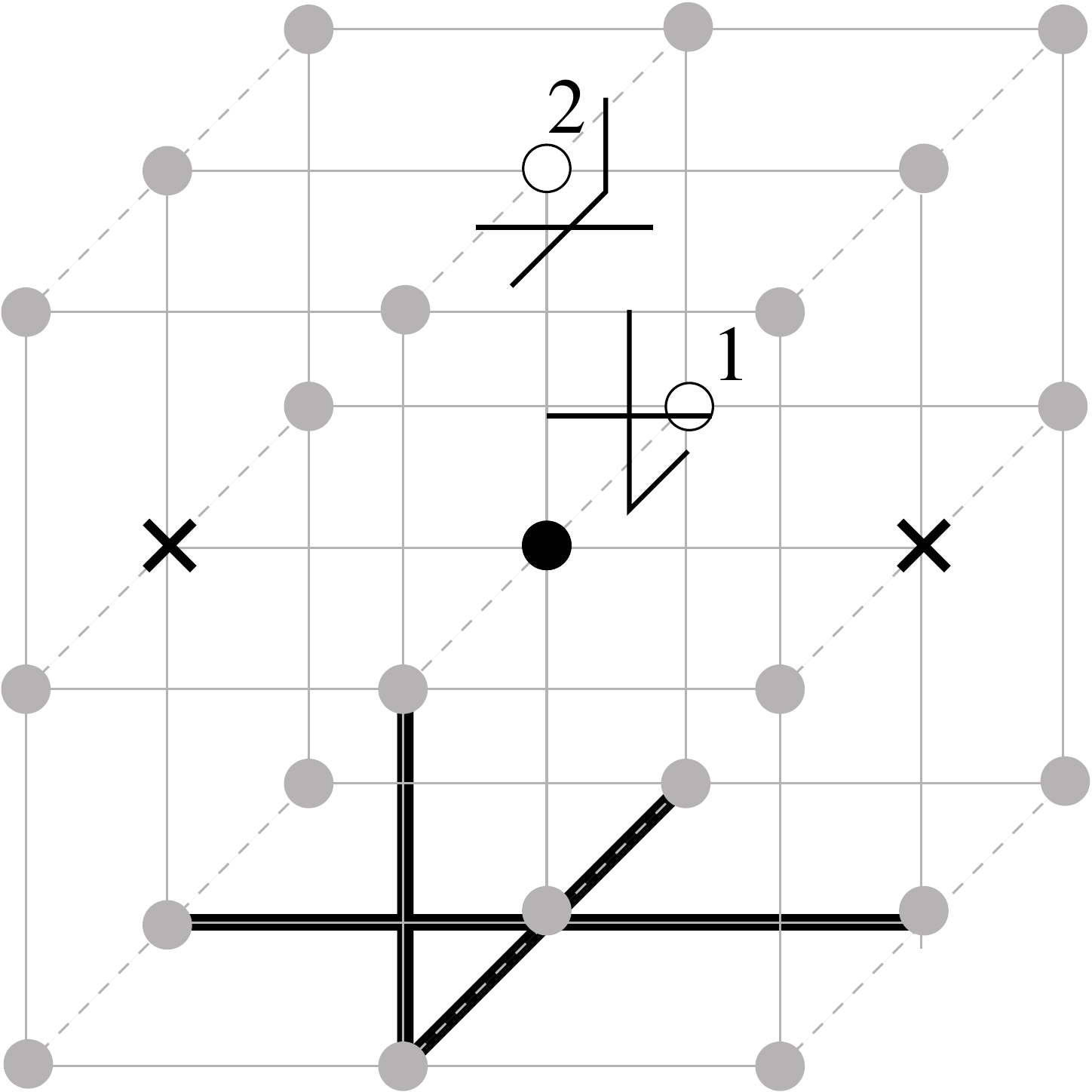} &
		\includegraphics[scale=0.11]{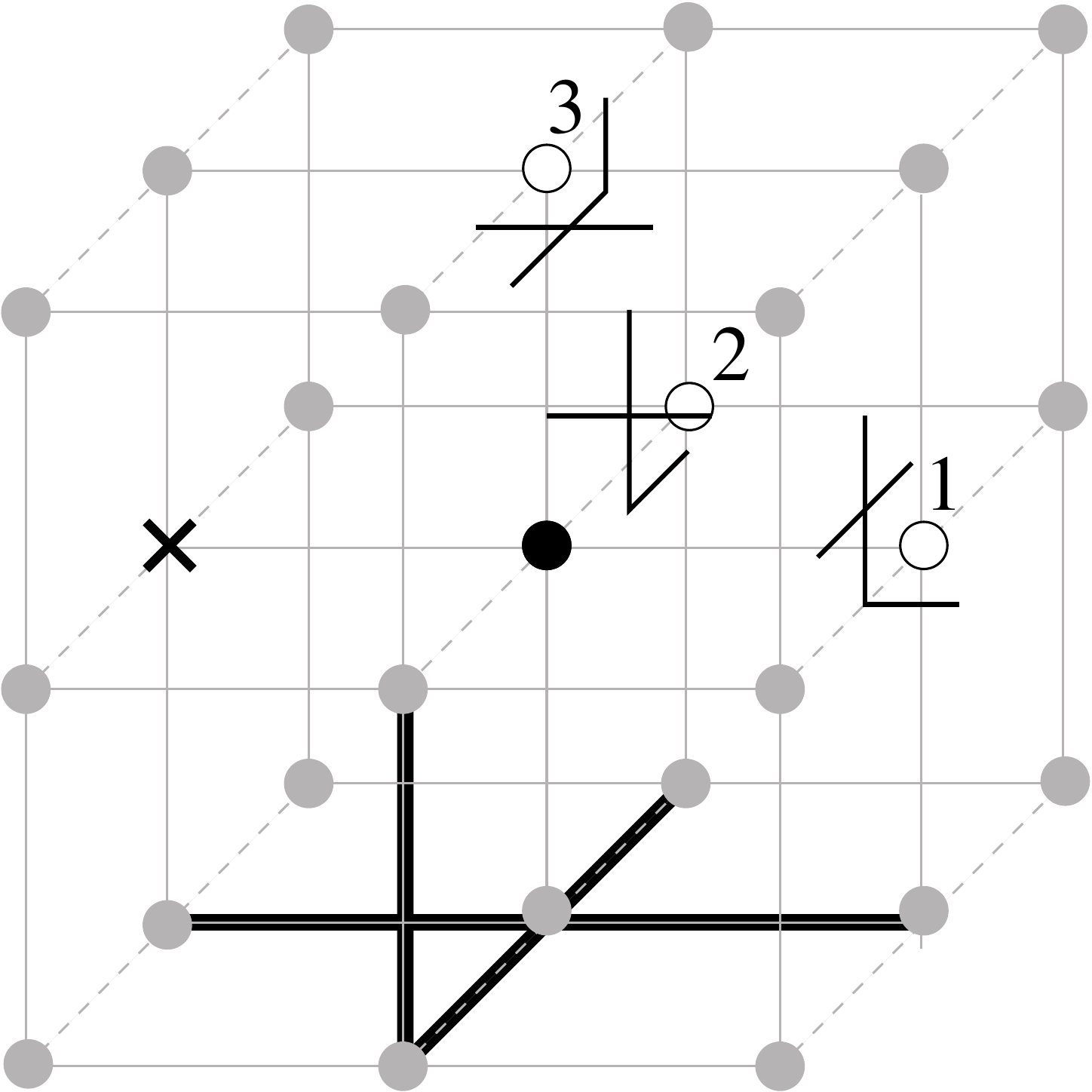} &
		\includegraphics[scale=0.11]{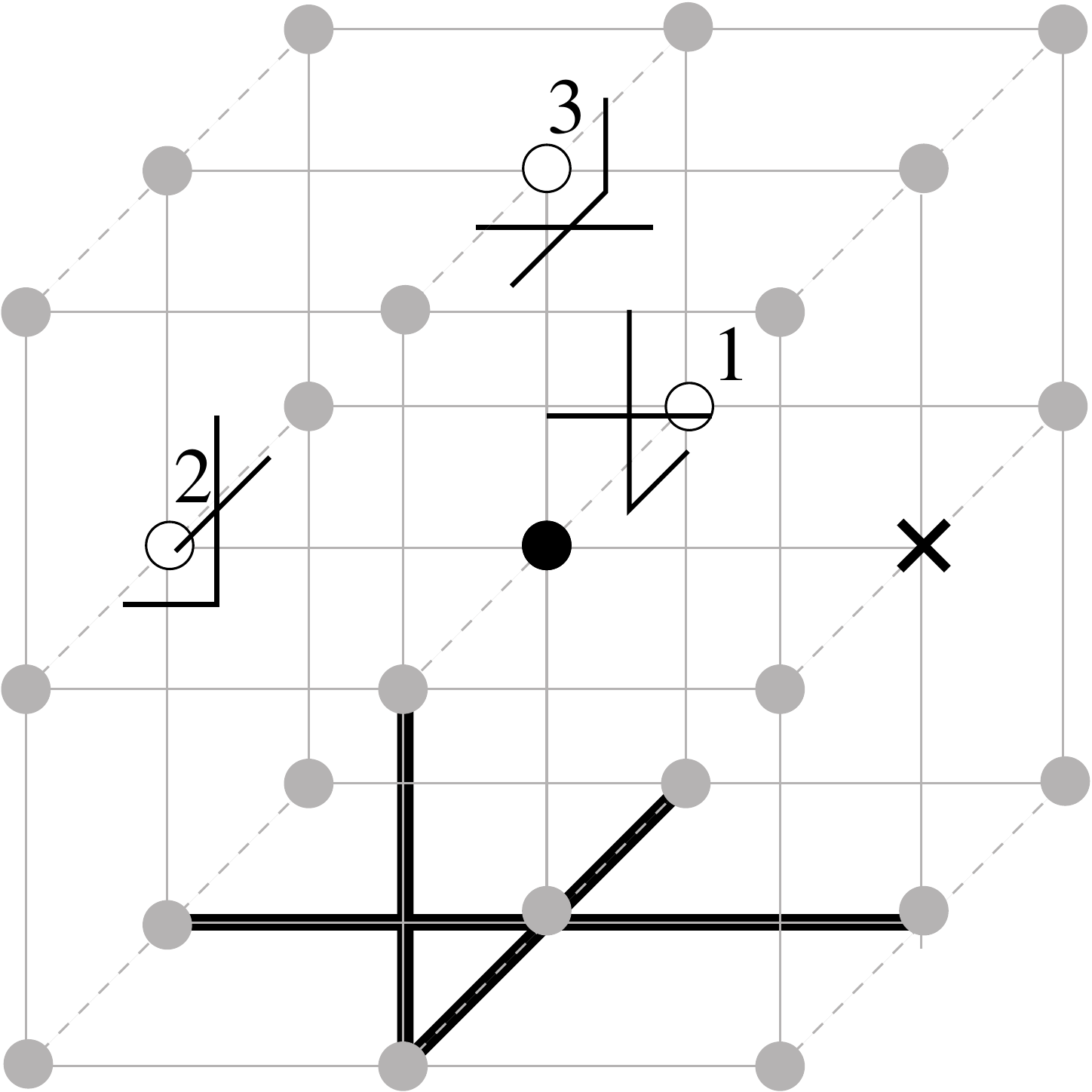} &
		\includegraphics[scale=0.11]{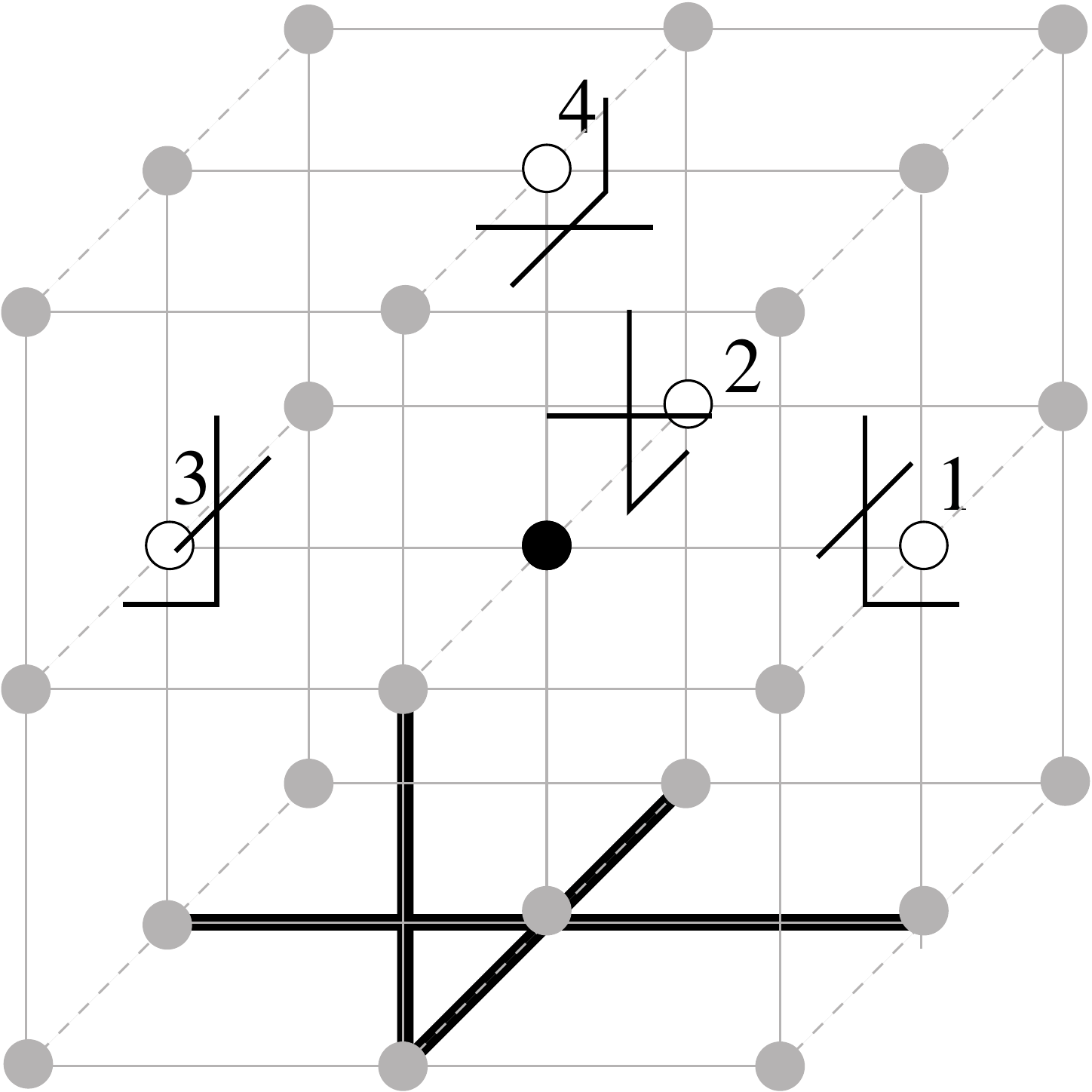}  \\
		$T_9$ & $T_{10}$ & $T_{11}$ & $T_{12}$  \\
		\includegraphics[scale=0.1]{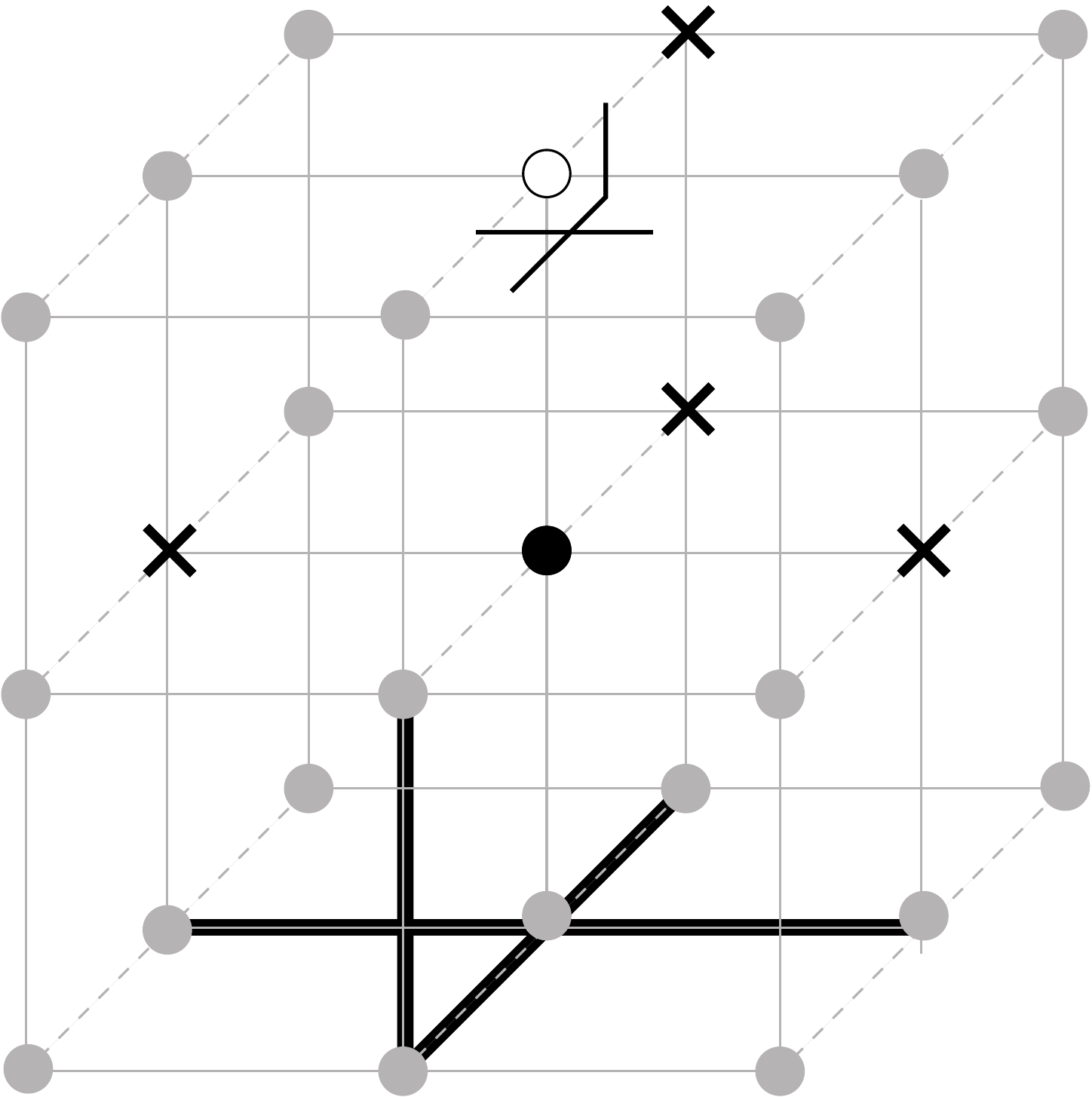} &
		\includegraphics[scale=0.1]{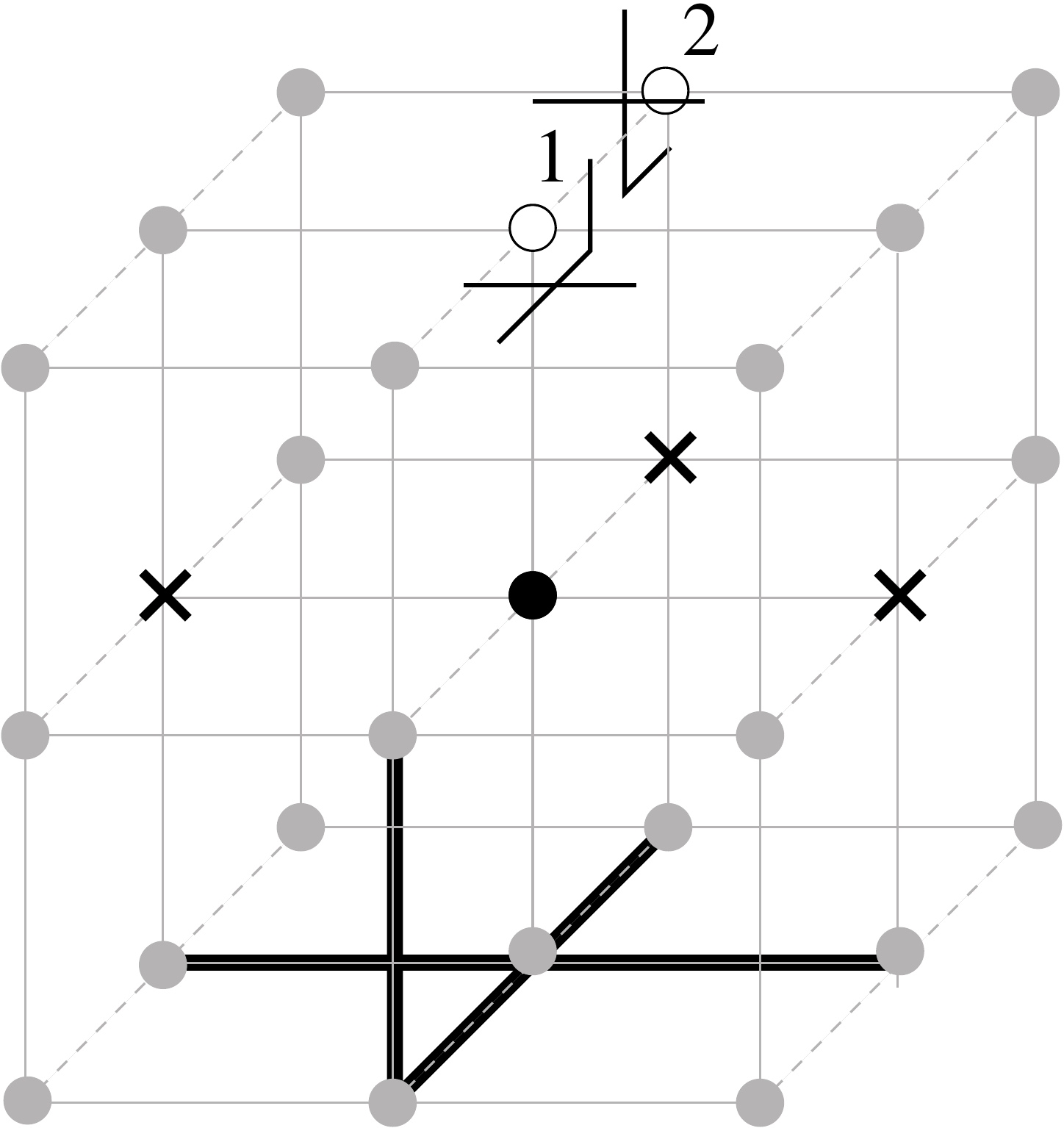} &
		\includegraphics[scale=0.1]{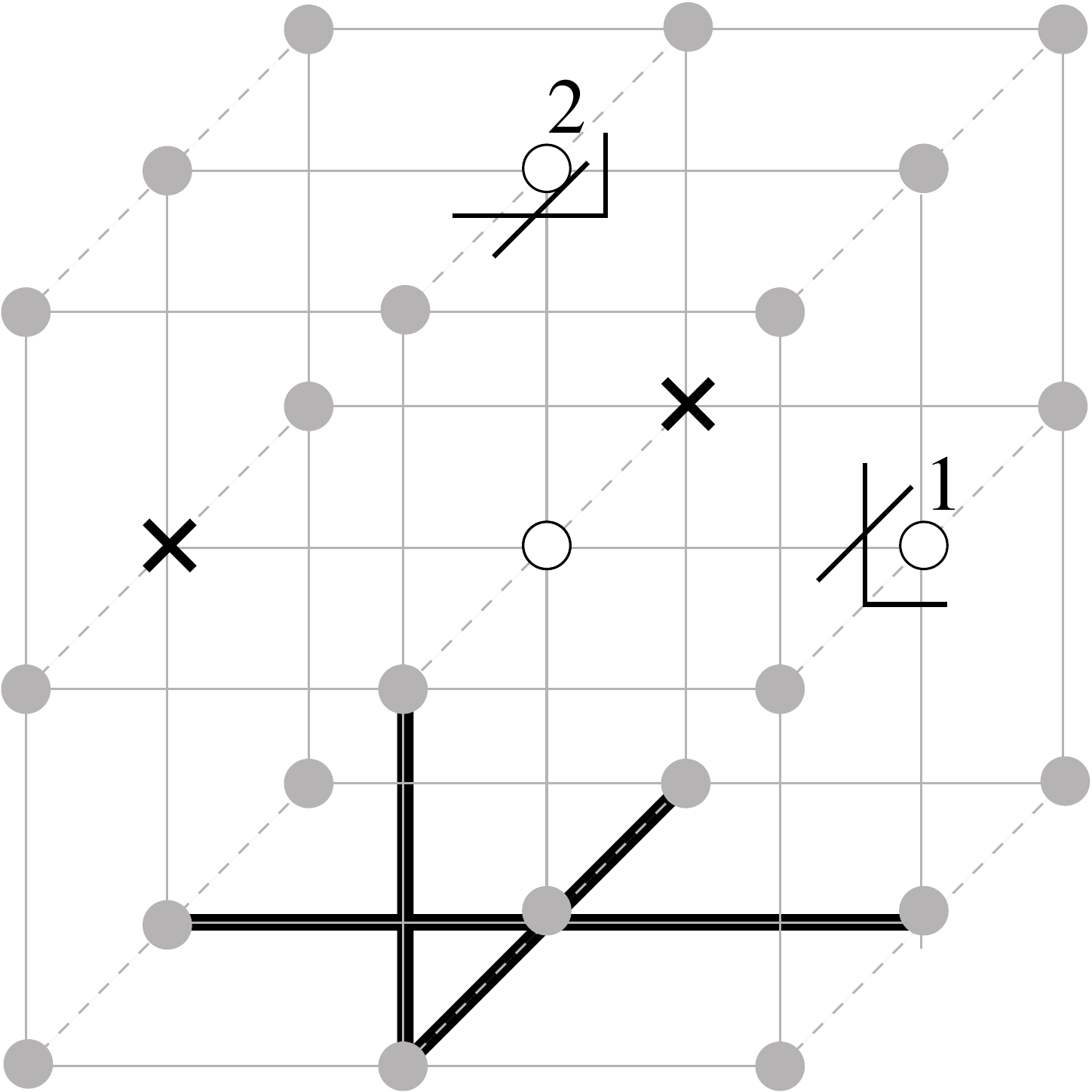} &
		\includegraphics[scale=0.1]{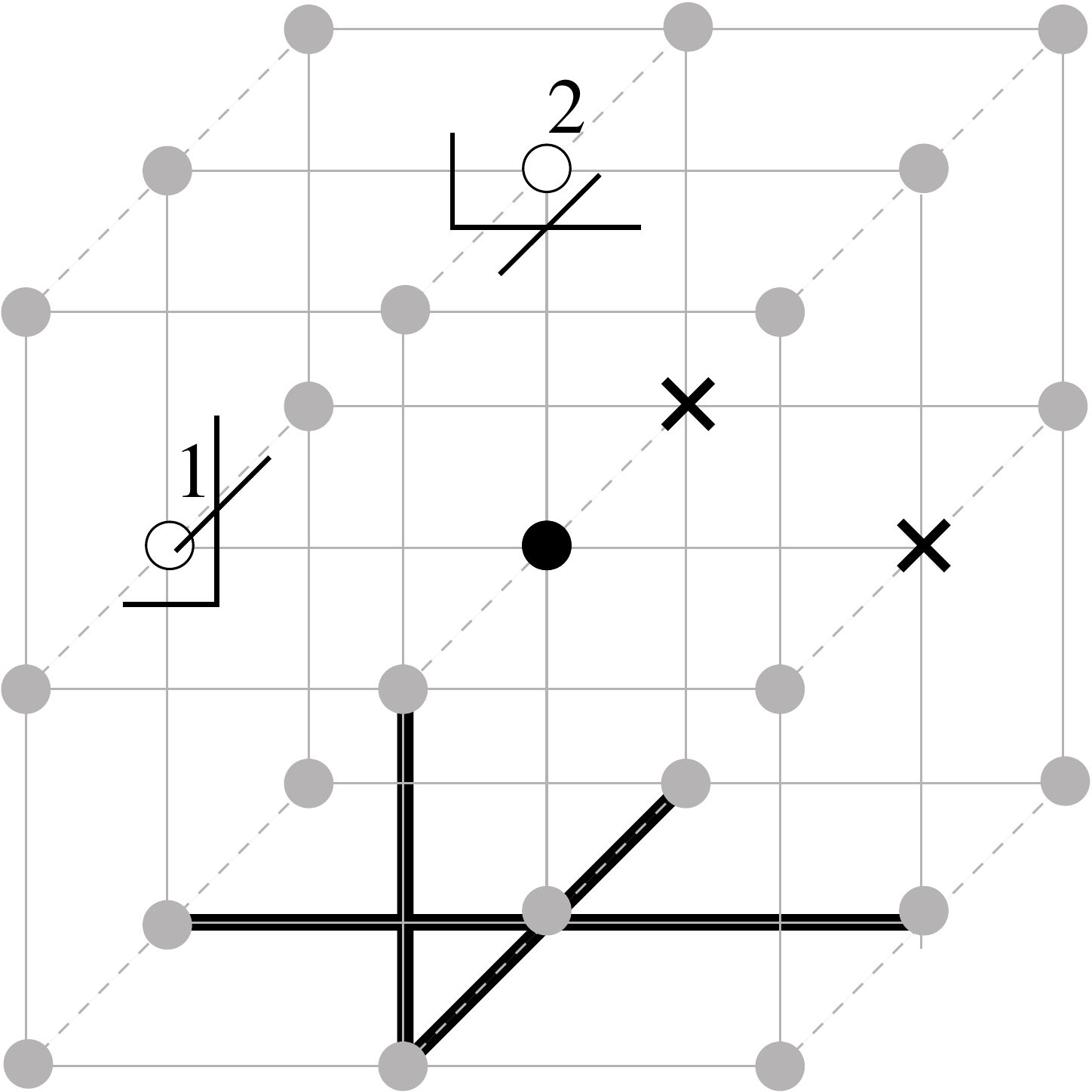} \\
		$T_{13}$ & $T_{14}$ & $T_{15}$ & $T_{16}$ \\
		\includegraphics[scale=0.1]{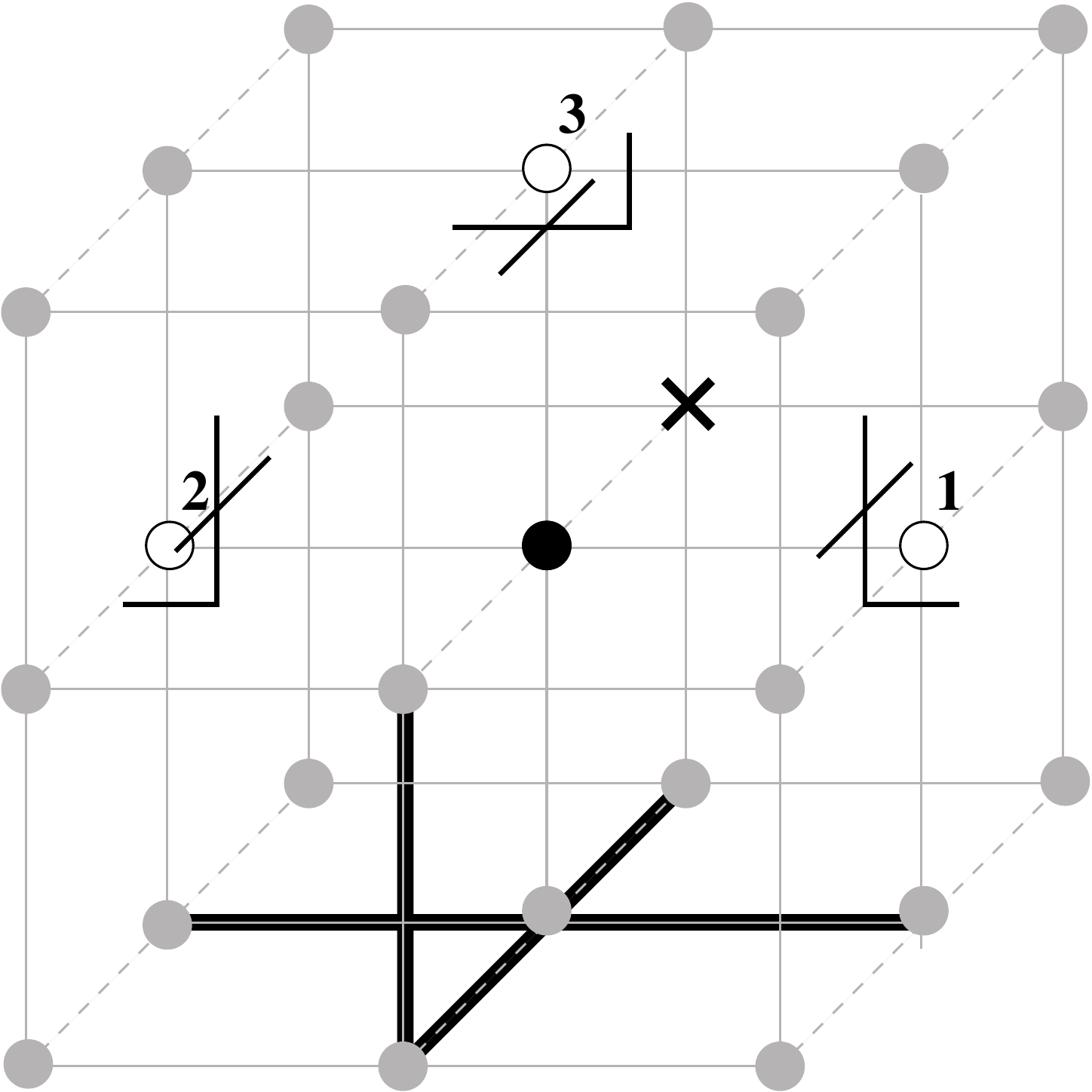}& & & \\
		$T_{17}$ & & &\\
	\end{tabular}
}{figure}{fig:3d-twigs}{3-dimensional twigs}
Let~$o=(0,0,0)$ be the lexicographically-smallest cell of the polycube.
Thus, by definition, all cubes that lie in the planes~$x_1=-1$ and~$x_2=1$ do
not belong to the polycube.
We define the ``$+\LC$-context'' of~$o$ (this name indicates a composition
of a ``Plus'' and an ``Ell'') to be the six cells around~$o$ shown in asterisks
in Figure~\ref{fig:Lcontext-3D}.
Observe that the set of 2-dimensional twigs~$\LC$ (Figure~\ref{fig:Ltwigs})
captures all possible occupancy configurations of the neighbors of~$o$ that lie in the~$x_1 x_2$ plane.
For the remaining neighbors of~$o$ (namely, cells~$(0,0,-1)$
and~$(0,0,1)$), there are~$2^2=4$ possible encodings of whether or not they
belong to the polycube.
This yields the set~$\LC^{(3)}$ of seventeen 3-dimensional twigs, shown in
Figure~\ref{fig:3d-twigs}.
%%% \begin{figure}
%%%    \begin{tabular}{ccccc}
%%%       \includegraphics[scale=0.17]{twigs/t6.pdf} &
%%%          \includegraphics[scale=0.17]{twigs/t7.pdf} &
%%%          \includegraphics[scale=0.17]{twigs/t9.pdf} &
%%%          \includegraphics[scale=0.17]{twigs/t13.pdf} & \\
%%% 		$T_1$ & $T_2$ & $T_3$ & $T_4$ & \medskip \\
%%%       \includegraphics[scale=0.17]{twigs/t8.pdf} &
%%%          \includegraphics[scale=0.17]{twigs/t12.pdf} &
%%%          \includegraphics[scale=0.17]{twigs/t14.pdf} &
%%%          \includegraphics[scale=0.17]{twigs/t2.pdf} & \\
%%%       $T_5$ & $T_6$ & $T_7$ & $T_8$ & \medskip \\
%%%       \includegraphics[scale=0.17]{twigs/t15.pdf} &
%%%          \includegraphics[scale=0.17]{twigs/t4.pdf} &
%%%          \includegraphics[scale=0.17]{twigs/t3.pdf} &
%%%          \includegraphics[scale=0.17]{twigs/t1.pdf} & \\
%%% 		$T_9$ & $T_{10}$ & $T_{11}$ & $T_{12}$ & \\
%%%       \includegraphics[scale=0.17]{twigs/t10.pdf} &
%%%          \includegraphics[scale=0.17]{twigs/t11.pdf} &
%%%          \includegraphics[scale=0.17]{twigs/t17.pdf} &
%%%          \includegraphics[scale=0.17]{twigs/t16.pdf} &
%%%          \includegraphics[scale=0.17]{twigs/t5.pdf} \\
%%%          $T_{13}$ & $T_{14}$ & $T_{15}$ & $T_{16}$ & $T_{17}$
%%%    \end{tabular}
%%%    \caption{3-dimensional twigs}
%%%    \label{fig:3d-twigs}
%%% \end{figure}
Similarly to two dimensions, the cells of a twig are either black or white, %, and
%forbidden cells are marked with an $\mathsf{X}$.
and the $+\LC$ context and linear order of the open cells is indicated.
Similarly to
$L_1,\dots,L_5$ (in Figure~\ref{fig:eden-twigs}), it is easy to check that
$T_1,\dots,T_{17}$ (in
Figure~\ref{fig:Ltwigs}) serve as complete building blocks for polycubes since they cover all possible situations (a formal proof is given in the next section).
Every $n$-cell polycube~$P$ corresponds to a unique $n$-term sequence of
elements of~$\LC^{(3)}$,
and different polycubes are assigned different sequences.
The sequence corresponding to a polycube can be constructed
algorithmically by a breadth-first search as in two dimensions.
Every twig is assigned a weight in the same manner, and we get that
\[
   \sum_{\ell\in \LC^{(3)}} w(\ell) = y(1+4x+7x^2+4x^3+x^4) =
      y\left((x+1)^4+x^2\right).
\]
Thus, the generating function given by Eq.~(\ref{eq:sum-of-weights}) is
\begin{equation}
   \frac{x}{1-y(1+4x+7x^2+4x^3+x^4)} =
      \sum_{n=0}^{\infty} x y^n (1+4x+7x^2+4x^3+x^4)^n.
   \label{eq:3d-gen}
\end{equation}
See Section~\ref{sec:agf} for the full analysis of this generating function.

%A simple analysis of generating function (given in the next section), proves that $\lambda_3 \le 9.8073$.
%A standard calculation shows that
%The radius of convergence of the two-variable
%generating function in Equation~(\ref{eq:3d-gen}) is about~0.1019649090.
%Hence, $\lambda_3 \leq 1/0.1019649090 \leq 9.8073$.

% --------------- begin comment ---------------------
\begin{comment}
DO WE NEED THE FOLLOWING TEXT?]]]
The idea is that many sequences of elements of L are illegal, namely, they contain intersections of cells or overlaps between forbidden cells and occupied cells. These sequences do not represent valid polyominoes. When building larger sets of twigs, on the one hand, we ensure that the new set of twigs is complete, and on the other hand, the number of illegal sequences (sequences that do not represent polyominoes) gets smaller, as by construction of the new set, we disqualify invalid sequences which were included in the generating function corresponding to the smaller set of twigs.
\end{comment}
% --------------- end comment ---------------------

% ----------------------------------------------------------------------------
% Higher Dimensions
% =================

\subsection{\bm{$d>3$}}
\label{sec:high-dim}

%[[[[[MAKE SURE THAT THE NUMBER OF CELLS AROUND $u$ ($2d$ AND NOT $2d-2$) DOES NOT AFFECT THE NUMBER OF TWIGS.]]]]]

Our construction in three dimensions can be applied inductively for~$d>3$.
%This is the first algorithmic approach for improving the upper bound
%on~$\lambda_d$.
% We believe that implementing a computer program to carry this approach
% is feasible for small values of $d$, say, up to~5 or~6.
%Let us, for simplicity, refer to $d$-dimensional polycubes as
%polycubes throughout this section.
The base of the induction is~$d{=}2$, where we fix a square in
the~$x_1 x_2$ plane (as in Figure~\ref{fig:Lcontext-3D}) together with its
$\LC$-context.
Going to~$d{=}3$, the square gains two new neighbors in the third
dimension~$x_3$.  In general, when we go from~$d{-}1$ to~$d$ dimensions,
a cube gains two neighbors in the new dimension~$x_d$.
Let~$o=(0,0,\dots,0)$ be a $d$-dimensional cube ($d>2$).
We define the $+_d\LC$-context of~$o$ in a recursive way.
The base of the definition is~$+_2\LC := \LC$ and~$+\LC := +_3\LC$, and
the recursion is~$+_d\LC := +_{d-1}\LC \cup \{c_1,c_2\}$,
where~$c_{1,2}=(-1,0,\dots,0,\pm1)$.
The geometric interpretation of the $+_d\LC$-context of~$o$ is an
L-shape around~$o$ in the~$x_1 x_2$ plane, which intersects $d{-}2$ lines
in the~$x_1 = -1$ plane at the point~$(-1,0,\dots,0)$.

The set of twigs~$\LC^{(d)}$ (where $\LC^{(2)}=\LC$) is comprised of
all~$2^{2d-2}$ occupancy options for the cells neighboring~$o$ (which
are not in its $+_d\LC$-context):
In dimensions~$x_3,\dots,x_d$, the construction simply covers all~$2^{2(d-2)}$
occupancy options for the two neighbors of~$o$.
In the $x_1 x_2$ plane, the occupancies of the neighbors
of~$o$ are captured by~$\LC$, as in Fig.~\ref{fig:Ltwigs},
and the only ``problematic'' case
is when the cube~$(1,0,\dots,0)$ is white and all other neighbors
of~$o$ are not (as is the case with the twigs $L_2, L_3$ in
Figure~\ref{fig:Ltwigs}, and $T_{13},T_{14}$ in Figure~\ref{fig:3d-twigs}).
It is, thus, necessary to encode the status of the cell~$(1,-1,0,\dots,0)$,
since, by construction, it is contained in the $+_d\LC$ context of the
cell~$(1,0,\dots,0)$.
This results in $2^{2(d-2)} \cdot 2^2 + 1 = 2^{2(d-1)} + 1$ twigs, and compares favorably with the generalization of Eden's construction,
which contains about two times more ($2^{2d-1}$) twigs.

In order to prove that our construction works better, all we need to show
is that for any white cell~$u$ in every twig in~$\LC^{(d)}$,
there is a set of $4+2(d-2)=2d$ cells around~$u$, which can be completely
ignored by the construction when visiting~$u$.  Those cells
can form its $+_d\LC$-context.
Note that except the second white cell in the problematics case mentioned
above, all white cells are neighbors of~$o$.
In case a new neighbor of~$o$, namely,~$(0,0,\dots,0,\pm1)$, is open,
the $\LC$-shape in its $+_d\LC$-context is formed by~$c_1$
(or~$c_2$), $(-1,0,\dots,0)$, $o$, and~$(1,0,\dots,0)$;
the rest of the cells in its $+_d\LC$-context are
$(0,\pm 1, 0,\dots,0)$, $\dots$, $(0,\dots,0,\pm 1,0)$
(that is, all coordinates are 0, except one coordinate in the
range~$2,\dots,(d-1)$, which is~$\pm 1$).
Note that the statuses of these cells are known by construction.
For the other possible white neighbors~$n=(0,\dots,0,\pm 1, 0,\dots,0)$ of~$o$,
$+_d\LC = +_{d-1}\LC \cup (0,\dots,0,\pm 1)$ since, by construction, $o$ is
where the~`$\LC$' and the~`$+$' in the~$+_{d-1}\LC$-context of~$n$ intersect.
Thus, the union~$+_{d-1}\LC \cup (0,\dots,0,\pm 1)$ forms its
$+_d\LC$-context;
The statuses of the cells in its $+_{d-1}\LC$-context
and of~$(0,\dots,0,\pm 1)$ are known by induction and construction,
respectively.
Finally, we need to address the second white cell $p=(1,-1,0,\dots,0)$ in the
problematic twig mentioned above. This cell will \emph{always} be visited after
the first open cell~$q=(0,1,0\dots,0)$ of the twig is visited and assigned a
twig.
Once~$q$ is assigned a twig, the statuses of all its neighbors will be encoded.
Then, it is easy to check that the cells~$(1,0,\dots,0)$, $o$, $(0,1,0\dots,0)$,
and $(0,2,0\dots,0)$ form the $\LC$-shape in the neighborhood of~$p$, and
together with the remaining $2(d{-}1)$ neighbors of~$q$, they form the
$+_d\LC$-context of~$p$.
The statuses of these cells are already encoded by construction.

Finally, we need to compute the weight
function~$W^{(d)}(x,y)=\sum_{t \in \LC^{(d)}}w(t)$.
Since~$o$ has~$2(d{-}1)$ neighbors that are not in its $+_d\LC$-context,
there are~$\binom{2(d-1)}{i}$ twigs in~$\LC^{(d)}$ with exactly~$i$ white cells
and one black cell~($o$), and the weight of each such twig is~$x^i y$.
Recall the problematic case mentioned above, which results in an
additional twig with three cells---1 black and 2 white---whose weight
is therefore~$x^2 y$.  Hence,
\[
   W^{(d)}(x,y) = \sum_{t \in \LC^{(d)}}w(t) =
      \sum_{i=0}^{2(d-1)}\left[ \binom{2(d-1)}{i} x^i y \right]+ x^2 y =
       y((x+1)^{2(d-1)}+x^2).
\]
Substituting~$W^{(d)}(x,y)$ in the generating function from
Equation~(\ref{eq:sum-of-weights}), we obtain
\[
   g^{(d)}(x,y) = l_d(n,m)x^my^n = \sum_{n=0}^{\infty} x y^n \left((1+x)^{2(d-1)} + x^2\right)^n =
      \frac{x}{1-y\left((1+x)^{2(d-1)} + x^2\right)},
\]

Similarly to polyominoes, polycubes of size~$n$ are mapped uniquely to
sequences of elements of~$\LC^{(d)}$ having weight $x^n y^n$.

\subsection{Analysis of the Generating Functions}
\label{sec:agf}
%a generating function which dominates that of $d$-dimensional polycubes; namely,
% We have shown that $A_d(n) \le l_d(n,n)$.
It can be easily observed that~$l_d(n,n)$, the coefficient of~$x^ny^n$
in~$g^{(d)}(x,y)$, is the coefficient of~$x^{n-1}$
in~$\left((1+x)^{2(d-1)} {+} x^2\right)^n$.
We now show how to compute~$l_d(n,n)$.
% ------------ start comment -------------
\begin{comment}
%--- do we need this text at all? ---
For any fixed value of~$d$, the reciprocal of the radius of convergence of
the diagonal of the function~$g(x,y)$ is an improved upper bound
on~$\lambda_d$.
\end{comment}
% ------------ end comment -------------
Let $h^{(d)}(x) = \left((1+x)^{2(d-1)} {+} x^2\right)^n$.
We start with the simple cases of $d=2,3$, and then generalize the
calculation to any value of~$d$.

\subsubsection{$\bm{d=2}$}
\label{sec:d2}

In two dimensions,
$h^{(2)}(x)=\left((1+x)^2 + x^2\right)^n=(1+2x+2x^2)^n$.
By the Multinomial Theorem, we have that

\[
   (1+2x+2x^2)^n=\sum_{i_1,i_2} \left[ \binom{n}{n-i_1-i_2,i_1,i_2}(2x)^{i_1}(2x^2)^{i_2}\right].
\]
Since we want to compute the coefficient of $x^{n-1}$, we require
that~$i_1{+}2i_2=n{-}1$, i.e.,~$i_1=n{-}2i_2{-}1$.
Thus,
\[
\begin{aligned}
   \displaystyle
   & l_2(n,n) =
      \sum_{i_2}\left[\binom{n}{i_2+1,n-2i_2-1,i_2}2^{n-i_2-1}\right] =
      \frac{2^n}{2} \sum_{i_2}\left[\binom{n}{i_2+1,n-2i_2-1,i_2}
           \left( \frac{1}{2} \right)^{i_2} \right]= \\
   & \frac{2^n}{\sqrt{2}}\sum_{i_2} \left[ \binom{n}{i_2+1,n-2i_2-1,i_2}
        \left( \frac{1}{\sqrt{2}} \right)^{i_2}
        \left( \frac{1}{\sqrt{2}} \right)^{i_2+1} \right] <_*
     \frac{2^n}{\sqrt{2}} \left( \frac{1}{\sqrt{2}}
        + \frac{1}{\sqrt{2}} + 1 \right)^n =
     \frac{\left( 2 (\sqrt{2}+1) \right)^n}{\sqrt{2}}.
\end{aligned}
\]
(The relation ``$<_*$'' is because the summation in its left-hand side
contains only a subset of the terms whose sum is equal to the exponential
term on the right-hand side.)
Hence, $\lambda_2 \leq 2 (\sqrt{2}+1) \approx 4.82843$.

\subsubsection{$\bm{d=3}$}

\begin{theorem}
   \label{thm:3}
   $\lambda_3 \leq 9.8073$. %-- this is wrong!: No!
   %$\lambda_3 \leq 10.016$.
\end{theorem}

\begin{proof}
We repeat the calculation in the same manner as above.
\noindent
\[
\begin{aligned}
   \displaystyle
   h^{(3)}(x)=&\left((1+x)^4 + x^2\right)^n=(1+4x+7x^2+4x^3+x^4)^n= \\
   & \sum_{i_1,i_2,i_3,i_4} \left[ \binom{n}{(n-\sum_{j=1}^4i_j),i_1,i_2,i_3,i_4} 4^{i_1}7^{i_2}4^{i_3}
        x^{i_1+2i_2+3i_3+4i_4}\right].
\end{aligned}
\]

Similarly to the 2-dimensional case, we require
that~$i_1{+}2i_2{+}3i_3{+}4i_4=n{-}1$, that
is, \linebreak $i_1{=}n{-}1{-}2i_2{-}3i_3{-}4i_4$.
Substituting~$i_1$ in the right-hand side of the equality above, we obtain
\[
\begin{aligned}
   \displaystyle
   & l_3(n,n) =
      \sum \left[ \binom{n}{i_2+2i_3{+}3i_4{+}1,n{-}1{-}2i_2{-}3i_3{-}4i_4,i_2,i_3,i_4}
         4^{n-1-2i_2-3i_3-4i_4}7^{i_2}4^{i_3} \right].
\end{aligned}
\]
Therefore, by the Multinomial Theorem, we have that
\[
\begin{aligned}
   \displaystyle
   & l_3(n,n) =
      \frac{4^n}{4} \sum_{i_2,i_3,i_4} \left[
         \binom{n}{i_2+2i_3{+}3i_4{+}1,n{-}1{-}2i_2{-}3i_3{-}4i_4,i_2,i_3,i_4}
         \left( \frac{7}{4^2} \right)^{i_2} \left( \frac{4}{4^3} \right)^{i_3}
         \left( \frac{1}{4^4} \right)^{i_4} \right] < \\
   & \qquad
     \frac{4^n}{4} \left( \frac{7}{4^2}+\frac{1}{4^2}+\frac{1}{4^4}+1 +1\right)^n =
      \frac{1}{4} \left( \frac{641}{64} \right)^n.
\end{aligned}
\]
Thus, $\lambda_3 \leq \frac{641}{64} \approx 10.016$, already improving
significantly on the known upper bound of $\lambda_3 \leq 12.2071$ (see
Section~\ref{sec:prev}).
However, we can do better than that.
Let~$b > 0$ be some constant, whose value will be specified later,
and rewrite the multinomial expression above as
\[
\begin{aligned}\displaystyle &
l_3(n,n) =
\frac{4^n}{4}
\sum_{i_2,i_3,i_4} \left[
\binom{n}{i_2+2i_3{+}3i_4{+}1,n{-}1{-}2i_2{-}3i_3{-}4i_4,i_2,i_3,i_4}
\underbrace{
	\left( \frac{7}{\left(b\frac{4}{b}\right)^2} \right)^{i_2}
	\left( \frac{4}{\left(b\frac{4}{b}\right)^3} \right)^{i_3}
	\left( \frac{1}{\left(b\frac{4}{b}\right)^4} \right)^{i_4}
}_{c(b)}\right],
\end{aligned}
\]
and rearrange the three terms in~$c(b)$ as follows.

\[
\begin{aligned}\displaystyle &
c(b)=\left(\frac{1}{b^2}\right)^{i_2}
\left(\frac{7}{\left(\frac{4}{b}\right)^2}\right)^{i_2}
\left(\frac{1}{b^3}\right)^{i_3}
\left(\frac{4}{\left(\frac{4}{b}\right)^3}\right)^{i_3}
\left(\frac{1}{b^4}\right)^{i_4}
\left(\frac{1}{\left(\frac{4}{b}\right)^4}\right)^{i_4} =\\& \quad
\left(\frac{1}{b}\right)^{i_2}
\left(\frac{1}{b}\right)^{i_2}
\left(\frac{7}{\left(\frac{4}{b}\right)^2}\right)^{i_2}
\left(\frac{1}{b^2}\right)^{i_3}
\left(\frac{1}{b}\right)^{i_3}
\left(\frac{4}{\left(\frac{4}{b}\right)^3}\right)^{i_3}
\left(\frac{1}{b^3}\right)^{i_4}
\left(\frac{1}{b}\right)^{i_4}
\left(\frac{1}{\left(\frac{4}{b}\right)^4}\right)^{i_4} =\\& \quad
\left(\frac{1}{b}\right)^{i_2}
\left(\frac{7}{ \frac{16}{b} }\right)^{i_2}
\left(\frac{1}{b}\right)^{2i_3}
\left(\frac{4}{\frac{4^3}{b^2}}\right)^{i_3}
\left(\frac{1}{b}\right)^{3i_4}
\left(\frac{1}{\frac{4^4}{b^3}}\right)^{i_4}=
\left(\frac{1}{b}\right)^{i_2+2i_3+3i_4}
\left(\frac{7}{ \frac{16}{b} }\right)^{i_2}
\left(\frac{4}{\frac{4^3}{b^2}}\right)^{i_3}
\left(\frac{1}{\frac{4^4}{b^3}}\right)^{i_4}.
\end{aligned}
\]
Thus,
\[
\begin{aligned}\displaystyle &
l_3(n,n) =\\& \frac{4^n}{4}
\sum_{i_2,i_3,i_4} \left[
\binom{n}{i_2{+}2i_3{+}3i_4{+}1,n{-}1{-}2i_2{-}3i_3{-}4i_4,i_2,i_3,i_4}
\left(\frac{1}{b}\right)^{i_2+2i_3+3i_4}
\left(\frac{7}{ \frac{16}{b} }\right)^{i_2}
\left(\frac{4}{\frac{4^3}{b^2}}\right)^{i_3}
\left(\frac{1}{\frac{4^4}{b^3}}\right)^{i_4} \right] \\&
<
4^n \left(\frac{1}{b} + 1 + \frac{7b}{16} + \frac{b^2}{4^2} + \frac{b^3}{4^4} \right)^n,
\end{aligned}
\]
where the last relation is again due to the Multinomial Theorem and due to
the partial summation.

The heart of our trick is that the partial summation allows us to
choose the value of~$b$ that minimizes the sum of the chosen summands
(by assigning appropriate weights to the five components).  Define
\[
f(b)=\frac{1}{b} + 1 + \frac{7b}{16} + \frac{b^2}{4^2} + \frac{b^3}{4^4}.
\]
Our goal, then, is to choose~$b$ so as to minimize~$f(b)$.
% Consider
% $f'(b)=-\frac{1}{b^2}+\frac{7}{16}+\frac{b}{8}+\frac{3b^2}{256}$ and
% $f''(b)=\frac{2}{b^3}+\frac{1}{8}+\frac{3b}{128}$.
Elementary calculus shows that~$f(b)$ assumes its minimum
at~$b_0=1.274306378$ and that $f(b_0)=2.451823893$.
Recall that~$l_3(n,n) < 4^n f^n(b)$ for \emph{any}~$b$,
in particular, for~$b=b_0$.  Hence, finally,
\[
   l_3(n,n) < 4^n \cdot 2.451823893^n = 9.807295572^n.
\]
(Had we chosen~$b=1$, we would have obtained
again the bound~$\lambda_3 \leq 10.016$.)
The claim follows.
\end{proof}

\subsubsection{General value of~$\bm{d}$}

\begin{theorem}
   $\lambda_d \leq (2d-2)e + 1/(2d-2)$.
\end{theorem}

\begin{proof}
The proof for a general value of~$d>3$ is similar to that for $d=2,3$.
For simplicity, let us fix~$a=2(d-1)$. We have that
\[
\begin{aligned}\displaystyle &
h^{(d)}(x)=
\left((1+x)^{a} + x^2\right)^n=
\left(1+ax+\left(\binom{a}{2}+1\right)x^2+\sum_{j=3}^{a}\binom{a}{j}x^j\right)^n=\\&
\sum_{i_1,\dots,i_a} \left[ \displaystyle\binom{n}{(n-\sum_{j=1}^{a}i_j),
i_1,\dots,i_{a}}
a^{i_1}
\left( \binom{a}{2} +1\right)^{i_2}
\left( \prod_{j=3}^{a}\binom{a}{j}^{i_j}\right)
x^{i_1+2i_2+\dots+ai_a} \right].
\end{aligned}
\]
Again, we require that~$i_1+2i_2+\dots+ai_{a}=n{-}1$, that
is,~$i_1=n-1-\sum_{j=2}^{a}(j\cdot i_j)$.  Thus,

\[\begin{aligned}
\displaystyle &
l_d(n,n){=}
\sum_{i_2,{\dots},i_a}  \left[
\binom{n}{( \sum_{j{=}2}^{a}(j-1)i_j{+}1), (n{-}1{-}\sum_{j{=}2}^{a}(j{\cdot} i_j))
,i_2,{\dots},i_{a}}
a^{n{-}1{-}\sum_{j{=}2}^{a}(j\cdot i_j)}
\left( \binom{a}{2}{+}1\right)^{i_2}
\left( \prod_{j=3}^{a}\binom{a}{j}^{i_j} \right) \right].
\end{aligned}
\]
Therefore,
\[
\begin{aligned}\displaystyle &
l_d(n,n)=
\frac{a^n}{a}
\sum_{i_2,\dots,i_a} \left[
\binom{n}{( \sum_{j=2}^{a}(j{-}1)i_j{+}1), (n{-}1{-}\sum_{j=2}^{a}(j\cdot i_j))
,i_2,\dots,i_{a}}
\frac{\left( \binom{a}{2} +1\right)^{i_2}}{a^{2i_2}}
 \prod_{j=3}^{a}
 \frac{\binom{a}{j}^{i_j}}{a^{ji_j}} \right] = \\
 &
 \frac{a^n}{a}
 \sum_{i_2,\dots,i_a} \left[
 \binom{n}{( \sum_{j=2}^{a}(j{-}1)i_j{+}1), (n{-}1{-}\sum_{j=2}^{a}(j\cdot i_j))
 	,i_2,\dots,i_{a}}
 \left(\frac{\binom{a}{2} +1}{a^{2}}\right)^{i_2}
 \prod_{j=3}^{a} \left( \frac{\binom{a}{j}}{a^{j}} \right)^{i_j} \right].
 \end{aligned}
 \]
It is well-known that for all values of~$m$ and~$k$, such that~$1 \leq k \leq m$,
we have that~$\binom{m}{k} \le \frac{m^k}{k!}$.
Hence, for $j=3,\dots,a$, we have
that~$\frac{\binom{a}{j}}{a^{j}} \leq \frac{1}{j!}$.
It is also known
that~$e=\sum_{j=0}^{\infty} \frac{1}{j!}$.
Therefore,
\[
\begin{aligned}\displaystyle &
l_d(n,n) \leq
%\frac{a^n}{a}
%\sum \binom{n}{( \sum_{j=2}^{a}(j{-}1)i_j{+}1), n{-}1{-}\sum_{k=2}^{a}ki_k
%	,i_2,\dots,i_{a}}
%\prod_{k=3}^{a-1} \left( \frac{\binom{a}{k}}{a^{k}} \right)^{i_k}
%\left(\frac{\binom{a}{2} +1}{a^{2}}\right)^{i_2}  \\&
\frac{a^n}{a}
\sum_{i_2,\dots,i_a} \left[
\binom{n}{( \sum_{j=2}^{a}(j{-}1)i_j{+}1), (n{-}1{-}\sum_{j=2}^{a}(j\cdot i_j))
	,i_2,\dots,i_{a}}
\left( \frac{1}{2} + \frac{1}{a^2} \right)^{i_2}
\prod_{j=3}^{a} \left( \frac{1}{j!} \right)^{i_j} \right]
\\&
< a^n \left( 1+1+\left(\frac{1}{2}+\frac{1}{a^2}\right)+\sum_{j=3}^{a}\frac{1}{j!} \right)^n
=
a^n\left(\frac{1}{a^2}+\sum_{j=0}^a\frac{1}{j!}\right)^n < (ae+1/a)^n
\end{aligned}
\]
%It is easy to see that $ \frac{1}{2a} + \frac{1}{a^3} \le \frac{1}{2}$.
(The relation ``$<$'' above is again because the summation in its left-hand side
contains only a subset of the terms whose sum is equal to the exponential
term on the right-hand side, and
the factor $1/a$ in its left-hand side.)
Consequently, $\lambda_d \leq (2d-2)e + \frac{1}{2d-2}$.

\end{proof}

This compares well with the conjecture
that~$\lambda_d \sim (2d-3)e$~\cite{BBR10},
and improves upon Eden's upper bound of~$(2d-1)e$ (which can actually be
shown to be~$(2d-1.5)e$; see Section~\ref{sec:prev}).
For example, for $d=4$, we obtain~$\lambda_4 \leq 15.1284$,
whereas the bound provided by the generalized Eden's method is~17.6514.
%Moreover, we can also apply the same process as in
%TODO:move
%Section~\ref{subsec:construct_C_i} for building sets of increasingly
%larger twigs for any fixed value of~$d>2$, and improve the upper bound
%further.

\begin{comment}
$$
\left((1+x)^{2(d-1)} + x^2\right)^n =
\sum_{i=0}^{n} \binom{n}{i}(1+x)^{2(d-1)i}x^{2(n-i)} =
\sum_{i=0}^{n} \binom{n}{i} x^{2(n-i)} \sum_{j=0}^{2(d-1)i} \binom{2(d-1)i}{j} x^j
$$

$2n-2i+2di-2i < n-1 \rightarrow 2n-4i+2di < n-1 \rightarrow 4i-2di > n+1$
for $d>3$, $4i-2di < 0$, thus we obtain
$ i < \frac{n+1}{2(d-2)}$

we need to consider values of $j$ satisfying: $2(n-i)+j=n-1$, namely, $j=2i-n-1$. Thus we rewrite the last expression as

$$
\sum_{i=\lceil \frac{n+1}{2(d-2)}\rceil}^{n} \binom{n}{i} \binom{2(d-1)i}{2i-n-1} x^{n-1}
$$

$$i=n: \binom{2(d-1)n}{n-1}$$
$$i=n-1: n\binom{2(d-1)(n-1)}{n-3}$$
$$i=n-2: \binom{n}{n-2}\binom{2(d-1)(n-2)}{n-5}$$
$$i=n-3: \binom{n}{n-3}\binom{2(d-1)(n-3)}{n-7}$$
\end{comment}

\section{Further Improvements of the Upper Bounds on \bm{$\lambda_2$} and \bm{$\lambda_3$}}

\begin{comment}
\begin{figure}
 	\centering
 	\scalebox{0.5}{\input{single.pdf_t}}
 	\caption{A twig ($\bar{s}$) containing a single open cell}
 	\label{fig:liv-cell}
\end{figure}
\end{comment}

\figbox{
 	\scalebox{0.5}{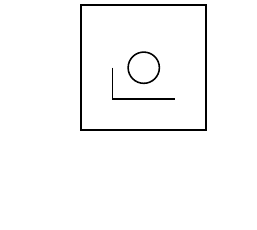}
}{figure}{fig:liv-cell}{A twig with one open cell}
Klarner and Rivest~\cite{KR73} developed their idea further,
noting that it is possible to start with a configuration containing a single open cell (as shown in Figure~\ref{fig:liv-cell}),
and keep adding twigs and updating the configuration,
%noting that the same
% spanning trees of polyominoes may be viewed as sequences of elements taken
% from a set of larger twigs.
 to construct from~$\LC$ %the basic twigs
 increasingly larger
 sets~$\C_1,\C_2,\C_3,\dots$, where the set~$\C_i$ contains all possible twigs with~$i$
 black cells (and possibly some white cells) or less than~$i$ black cells (and no
 white cells).
 In particular, $\C_1=\LC$.
 The process for building all twigs with $i$ black cells is as follows:

%\subsection{Constructing~\boldmath{$\C_i$}}
 %\label{subsec:construct_C_i}

 %Since all polyominoes of size~$i$ can be constructed with $i$
 %The procedure is the following.
 \begin{enumerate}%[label={\emph{Step} \arabic*.},leftmargin=2cm]
 	\item Set $\C_i := \emptyset$, $B := \{\bar{s}\}$
 	(the twig shown in~Figure~\ref{fig:liv-cell}, and
 	$W_i(x,y) := 0$;
 	\item \texttt{If} $B = \emptyset$,
 	\texttt{then} output~$\C_i$ and halt;
 	\item Remove some twig~$T$ from $B$;
 	\item \texttt{If}~$T$ contains no open cells or exactly~$i$ dead cells,
 	\texttt{then} add~$T$ to~$\C_i$, set $W := W+w(T)$, and
 	\texttt{goto} Step~2;
 	\item \texttt{For} $j=1,\dots,5$ \texttt{do} \\
 	\mbox{} \quad Set $T_j := T * L_j$; \\
 	\mbox{} \quad \texttt{If}~$T_j$ meets condition~\cast\ below,
 	\texttt{then} add~$T_j$ to~$B$; \\
 	\texttt{od}
 	\item \texttt{Goto} Step~2.
 \end{enumerate}

 \noindent \textbf{Condition \cast}:
 %\begin{itemize}
\emph{
None of the cells of~$L_i$ (except of the black cell) overlap with any of the cells (black or white) of~$T$ nor with any of the cells of~$T$ marked with~$\mathsf{X}$. %; and
 	%\item None of the forbidden cells of~$L_i$ overlap with any cells of~$T$.
}
% \end{itemize}

Condition \cast\ guarantees that adding a new twig to the configuration will
not cause cells to overlap. %TODO

\begin{observation}
   $A_2(i) \leq |C_i|$
\end{observation}

Indeed, this relation is trivially justified by the facts that
every polyomino of size~$i$ can be built with \emph{some} sequence
of~$i$ twigs, and the algorithm above constructs \emph{all} valid sequences
of~$i$ twigs.

  \figbox[l]{
  	\scalebox{0.20}{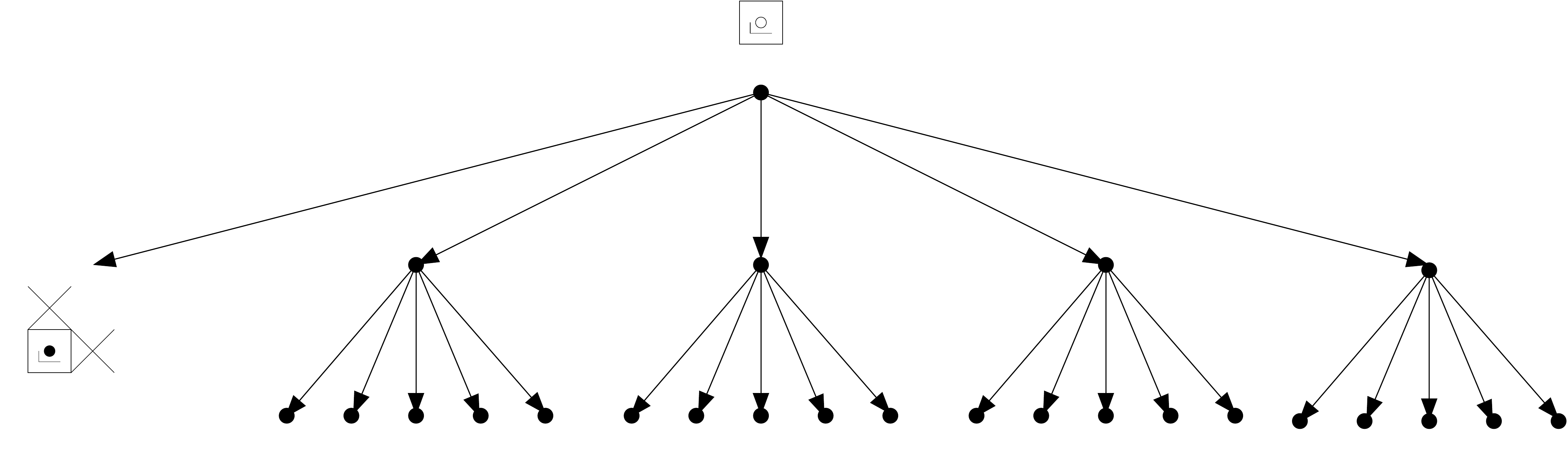}
  }{figure}{fig:branching}
  {The tree modeling the algorithm that generates~$\C_i$.
  	The root~$r$ is a twig with one open cell; its $\LC$-context is
  	shown in Fig.~\ref{fig:Lcontexts}(a).
  	For~$i,j=1,\dots,5$, set~$T_i=L_i= r * L_i$, and~$T_{i,j} = T_i * L_j$.
  	The twig~$T_1$ is a leaf because it has no open cells.
  }
 The algorithm above can be viewed as a breadth-first-search traversal on an
 infinite tree (see Figure~\ref{fig:branching}) rooted at the twig~$\bar{s}$
 (Figure~\ref{fig:liv-cell}).
 All other vertices of the tree are twigs that can be ``grown'' from its
 root by repeatedly applying the operation~`$*$' (defined in Section~\ref{sec:prev}).
 The tree contains an edge directed from a twig~$T_1$ to another twig~$T_2$
 if~$T_2=T_1*L_i$ (for some~$L_i \in \LC$).
 Hence, each vertex of the tree has at most five outgoing edges, and its
 leaves are all twigs which have no open cells.

 The key idea is that, given a polyomino~$P$, it is possible to encode~$P$
 with a sequence of elements of~$\C_i$, for any~$i \geq 1$, and
 any such sequence can be converted into a
 sequence of elements of~$\LC$.

\begin{observation}
	\label{ob:ci}
 The set of converted sequences of elements of~$\C_{i+1}$ is a proper subset of
 the set of converted sequences of elements of~$\C_i$,
 since the former contains less invalid sequences (those that do not
 represent polyominoes) than the latter.
\end{observation}

 Similarly to~$\LC$, every twig
 $T\in C_i$ is assigned a weight~$w(T) := x^a y^b$ (where~$a$ denotes the number of
 cells in~$T$ minus~1, and~$b$ denotes the number of black cells in~$T$),
and, thus, it can be shown that
 every polyomino of size~$n$ gives rise to a unique
 sequence of elements of~$\C_i$ of weight~$x^n y^n$.
% , but not all such
% sequences encode legal polyominoes.
Letting~$W_i(x,y) = \sum_{T \in \C_i} w(T)$,
we can plug $W_i(x,y)$ in the generating function in
Equation~(\ref{eq:sum-of-weights})
and obtain~$\sum_{m,n} c_i(m,n) x^m y^n{=}x/(1{-}W_i(x,y))$.
Again, we are interested in
 the diagonal term~$c_i(n,n)$
 of the series expansion~$\sum_{m,n} c_i(m,n) x^m y^n$.
 %because it
 % represents configurations that contain no living cells.
 %which a superset of the set of all polyominoes.
 Due to Observation~\ref{ob:ci},
 the sets~$\C_1,\C_2,\dots$ yield a sequence of improving (decreasing)
 upper bounds on~$\lambda_2$.
 %This method has the advantage of being amenable to systematic improvement.
 Thus, as~$i$ increases, the upper bound decreases. %$1/\sigma_i$ decreases. %, obtaining a better upper bound
 %on~$\lambda_2$.
 Therefore,
 the goal is to compute an upper bound on~$c_i(n,n)$.
 %the goal is to compute the growth constant of the number of
 % sequences of twigs taken from~$\C_i$.
 The main computational challenge in this approach is to construct
 algorithmically the sets~$\C_i$ (in order to compute $W_i(x,y)$), as $|\C_i|$ is increasing exponentially
 with~$i$, like~$A_2(i)$ does.
 Klarner and Rivest carried their approach to the limit of the resources
 they had available at the time, and computed~$\C_i$ up to $i=10$.
 %The value $1/\sigma_{10} \approx 4.6495$ has remained the best known upper
 %bound on~$\lambda_2$ for more than almost half a century.
 Their computations are summarized in Table~\ref{tab:results_2d}.

\subsection{Two Dimensions}

\label{sec:ub_imp}

\begin{theorem}
	$\lambda_2 \leq 4.5252$. \qed
\end{theorem}

We implemented the algorithm described in the previous section for constructing the sets $\C_i$ in a parallel C++ program, using
\texttt{Maple} (see Appendix~\ref{app:maple}) to derive an upper bound on~$c_i(x,y)$.
Since the size of the set~$\C_i$ is growing exponentially with~$i$,
we did not keep it in memory.
Instead, we accumulated the weights of the twigs as in Step~4 in the
algorithm.
The ``\texttt{for} loop'' in Step~5 can be run in parallel since there
are no dependencies between the twigs~$T_1,\dots,T_5$, as illustrated in
Figure~\ref{fig:branching}.
\begin{comment}
\begin{figure}
	\centering
	\scalebox{0.29}{\input{branching.pdf_t}}
	%\begin{tabular}{L{2.5cm} L{2cm} C{1cm} R{2.2cm} R{2.5cm}}
	%	{ } & $\bm{\vdots}$ & $\bm{\vdots}$ & $\bm{\vdots}$ & $\bm{\vdots}$\\
	%\end{tabular}
	\caption{The tree modeling the algorithm that generates the set $\C_i$.
		The root~$r$ is a twig with one open cell and its $\LC$-context is the one shown in Figure~\ref{fig:Lcontexts}(a).
		For~$i,j=1,\dots,5$, set~$T_i=L_i=L_i * r$, and~$T_{i,j} = L_j * T_i$.
		The twig~$T_1$ is a leaf because it has no open cells.
	}
	\label{fig:branching}
\end{figure}
\end{comment}
We used \texttt{OpenMP} and \texttt{OpenMPI} to run the program in parallel
on
%MIRA: I COMMENTED THIS
%``Tamnun'' (octopus in English)---
a high-performance computer cluster at the Technion.
We used~33 computing nodes,
each having 12 cores, for a total of~396 cores.
% We had a limit of 24 hours for a single run of the program,
The time for computing~$\C_{10}$ was negligible even without parallelizing
the program.
Results were systematically improved by increasing~$i$, the number of dead
cells of the twigs.
However, %the running time grew roughly exponentially with~$i$,
as the size of~$\C_i$ increases roughly by a factor of~4 as~$i$ is
incremented by~1,
constructing~$\C_{i+1}$ %and
%computing~$\sigma_{i+1}$
requires more than four times the computing power
needed to construct~$\C_i$. %and compute~$\sigma_i$.
The improved upper bound $\lambda_2 \leq 4.5252$ was obtained by using
twigs with~21 dead cells.
Computing~$\C_{21}$ took roughly seven hours.
Our results, alongside Klarner and Rivest's results, are summarized in
Table~\ref{tab:results_2d}.  The two sets of results differ
for~$i=6,\dots,10$.
We address these differences in Section~\ref{app:2d}.
The weight functions~$W_1(x,y),\dots,W_{21}(x,y)$ are provided in
Appendix~\ref{sec:GFs}.

For $i \geq 6$, the number of twigs ($|\C_i|$) we found is slightly (but
consistently) larger than the number reported by Klarner and
Rivest~\cite{KR73} (see Table~\ref{tab:results_2d}).
As a result, the value of the upper bound we computed for~$\C_{10}$ is slightly higher than
the value they reported.
Since they provided neither the computer program which generated the
sets~$\C_i$, nor the functions~$W_6(x,y),\dots,W_{10}(x,y)$ which they
obtained, we had no means for comparing our results to theirs.

\subsection{Three Dimensions}

We applied the process described in the previous section
to construct sets~$\C^3_1, \C^3_2, \dots$ of larger 3-dimensional twigs.
Again, we began with a single open cell on the cubical lattice,
and constructed all twigs with~$i$ dead cells or fewer dead cells and no
open cells.
We were able to reach twigs with~$i=9$ dead cells, obtaining a set of
about~$17 \cdot 10^9$ twigs, by which we proved
that~$\lambda_3 \leq 9.3835$.
Computing~$\C^3_9$ took~3 hours on the same cluster mentioned in
Section~\ref{sec:ub_imp}.
Our results (reported in Table~\ref{tab:OURresults3D}),
and~$W_9(x,y)=\sum_{\ell \in \C^3_8} w(\ell)$ are provided in
Appendix~\ref{app:3d}.

\subsection{Code}

Our code is available at \url{https://github.com/mshalah/polyominoes_polycubes_upperbounds}.

\newpage

\input{nub10.bbl}
% ----------------------------------------------------------------------------
% Appendices
% ==========

\newpage

\appendix

%$W_1(x,y)=2x^2y +2xy +y$\\[0.05in]

%$W_2(x,y)= 4x^4y^2 +8x^3y^2 +6x^2y^2 +2xy^2 +y$\\[0.05in]

%$W_3(x,y)=8x^6y^3 +24x^5y^3 +32x^4y^3 +20x^3y^3 +6x^2y^3 +2xy^2 +y$\\[0.05in]

%\parbox{5in}{
%	$W_4(x,y)=14x^8y^4 +58x^7y^4 +113x^6y^4 -124x^5y^4 +71x^4y^4 +20x^3y^4 +6x^2y^3 +2xy^2 +y$
%}\\[0.2in]

\section{Comparison of Results}

\label{app:2d}

\begin{table}[h]
	
	\centering
	
	\begin{tabular}{r|rr|rr|c}
		& \multicolumn{2}{c|}{$|\C_i|$} & \multicolumn{2}{c|}{$1/\sigma_i$} &
		Time (Hours) \\
		\multicolumn{1}{c|}{$i$} & \multicolumn{1}{c}{Ref.~\cite{KR73}} &
		\multicolumn{1}{c|}{Ours} & \multicolumn{1}{c}{Ref.~\cite{KR73}} &
		\multicolumn{1}{c|}{Ours} & Ours \\
		\hline \hline
		1 & 5 & 5 & 4.828428 & 4.828427124~~ \\
		2 & 21 & 21 & 4.828428 & 4.828427124~~ \\
		3 & 93 & 93 & 4.828428 & 4.828427124~~ \\
		4 & 409 & 409 & 4.796156 & 4.796155640~~ \\
		5 & 1,803 & 1,803 & 4.765534 & 4.765532996~~ \\
		6 & 7,929 & \textbf{7,937} & 4.738062 & \textbf{4.738743624} \\
		7 &  34,928 & \textbf{35,084} & 4.714292 & \textbf{4.716641912} \\
		8 & 151,897 & \textbf{153,458} & 4.690920 & \textbf{4.695386599} \\
		9 & 656,363 & \textbf{668,128} & 4.669409 & \textbf{4.676042980} \\
		10 & 2,821,227 & \textbf{2,899,941} & 4.649551 & \textbf{4.658412767} \\
		11 & & \textbf{12,557,503} & & \textbf{4.642235017} \\
		12 & & \textbf{54,137,703} & & \textbf{4.627069746} \\
		13 & & \textbf{232,203,877} & & \textbf{4.612780890} \\
		14 & & \textbf{991,607,177} & & \textbf{4.599355259} \\
		15 & & \textbf{4,218,349,778} & & \textbf{4.586741250} \\
		16 & & \textbf{17,881,987,659} & & \textbf{4.574877902} \\
		17 & & \textbf{75,568,307,191} & & \textbf{4.563716381} & \\
		18 & & \textbf{318,489,941,731} & & \textbf{4.553209881} & 0:04 \\
		19 & & \textbf{1,339,093,701,964} & & \textbf{4.543308340} & 0:20 \\
		20 & & \textbf{5,617,897,764,831} & & \textbf{4.533962650} & 1:30 \\
		21 & & \textbf{23,521,568,438,976} & & \textbf{4.525128839} & 7:00
		% 22 & \textbf{23521568438976} & \textbf{4.525128839}\\
		% long long limit 9,223,372,036,854,775,807
	\end{tabular}
	\caption{Left: Results obtained by Klarner and Rivest~\cite[Table~1]{KR73};
		Right: Our results for~$d{=}2$.
	}
	\label{tab:results_2d}
\end{table}

%Moreover, and although Klarner and Rivest did not state the following
%explicitly, they probably did not use the second part of
%condition~\cast\ in their program.
In fact, Klarner and Rivest claim to have used the following version of condition \cast: \\

\noindent \textbf{Condition \cast}:
\begin{itemize}
	\item None of the cells of~$L_i$ (except its root) overlap with any of
	the cells or forbidden cells of~$T$; and
	\item None of the forbidden cells of~$L_i$ overlap with any cells of~$T$.
\end{itemize}

However,
although they did not state the following
explicitly, they probably did not use the second part of condition~\cast\ in their program.
We motivate this claim by the two arguments provided below,
emphasizing that indeed using the second part of this condition \emph{as-is}
is incorrect, as we explain in the second argument.
%Therefore, we also disregarded the second part of condition~\cast,
%and allowed forbidden cells of~$L_i$ to overlap with cells of~$T$ in
%Step~5 of the construction algorithm.
Still, our results agree with those of Klarner and Rivest's only up
to~$i=5$, and we were unable to trace further the causes for the
differences for~$i \geq 6$.

\subsection*{Argument 1}

Consider the set~$\C_4$, which contains all twigs that have exactly~4 black
cells, or fewer black cells and no white cells.
Let us enumerate the twigs that have one, two, or three black cells and no
open cells.  We can easily observe the following.
\begin{enumerate}
	\item There is only \emph{one} twig ($L_1 \in \LC$, shown in
	Figure~\ref{fig:Ltwigs}) with one dead cell and no open cells.
	\item There are only \emph{two} twigs (see
	Figure~\ref{fig:two-dead-cells})
	\begin{figure}
		\centering
		\begin{tabular}{ccc}
			\scalebox{0.5}{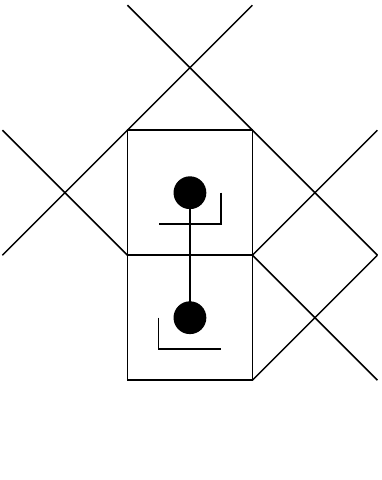} & ~~ &
			\scalebox{0.5}{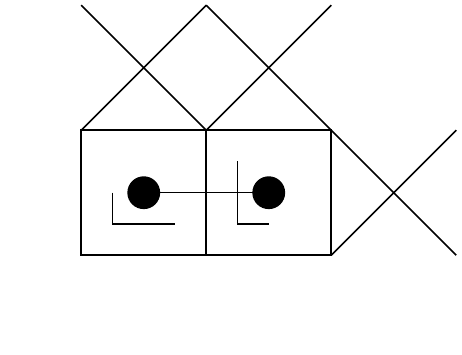} \\
			$L_2 \ast L_1$ & & $ L_4 \ast L_1$
		\end{tabular}
		\caption{The only two twigs with two dead cells and no open
			cells}
		\label{fig:two-dead-cells}
	\end{figure}
	with two dead cells and no open cells.
	\item There are only \emph{six} twigs with three dead cells and no open
	cells.  The sequences corresponding to these six twigs are
	$L_2 \ast L_2 \ast L_1$,
	$L_2 \ast L_4 \ast L_1$,
	$L_3 \ast L_1 \ast L_1$,
	$L_4 \ast L_2 \ast L_1$,
	$L_4 \ast L_4 \ast L_1$, and
	$L_5 \ast L_1 \ast L_1$.
\end{enumerate}

We now show that there are~400 twigs with~4 dead cells.
% and zero open cells.
%, resulting in 409 in the set $\C_4$.
Each such twig~$T$ corresponds to a
sequence~$(\alpha, \beta, \gamma, \delta)$ of four elements of~$\LC$
(Figure~\ref{fig:Ltwigs}), such
that~$T= \alpha \ast \beta \ast \gamma \ast \delta$.
%(Note that~$T$ contains its forming twigs in reverse order.)
Trivially, in total, there are~$|\LC|^4=5^4=625$ sequences of four elements of~$\LC$.
However, some of these sequences are invalid, namely, they do not
represent valid twigs.
For example, the only valid sequence that starts with~$L_1$ is of length~1
since~$L_1$ has no open cells.
It is easy to verify that the following sequences are exactly those that
are invalid since their prefixes correspond to configurations with less
than four dead cells and no open cells:
%TODO define descendants
%has no descendants.
% till i=5 there is no overlap of dead cells
$S_1=(L_1,\beta,\gamma,\delta)$,
$S_2=(L_2,L_1,\gamma,\delta)$,
$S_3=(L_4,L_1,\gamma,\delta)$,
$S_4=(L_2,L_4,L_1,\delta)$,
$S_5=(L_4,L_2,L_1,\delta)$,
$S_6=(L_2,L_2,L_1,\delta)$,
$S_7=(L_4,L_4,L_1,\delta)$,
$S_8=(L_3,L_1,L_1,\delta)$, and
$S_9=(L_5,L_1,L_1,\delta)$.

Clearly, we have that~$|S_1|=5^3=125$, $|S_2|=|S_3|=5^2=25$,
and~$|S_4|=|S_5|=|S_6|=|S_7|=|S_8|=|S_9|=5$.
There remain exactly twenty invalid sequences.
Refer to Figure~\ref{fig:Ltwigs}.
\begin{figure}
	\begin{tabular}{>{\RaggedRight}p{3cm}>{\RaggedRight}p{3cm}>{\RaggedRight}p{4cm}}
		{} & {} & \scalebox{0.5}{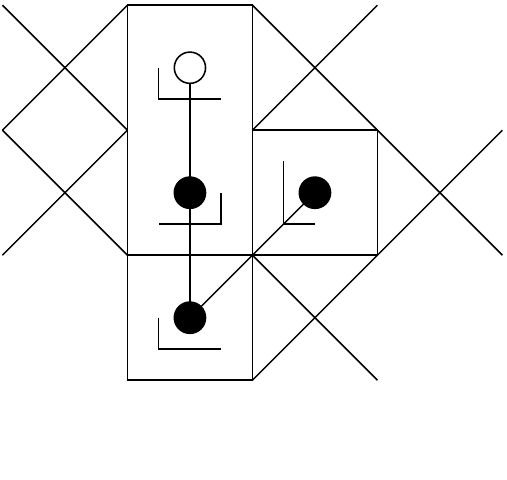} \\
		{} & {} &\vspace{-8mm} (a) $L_3 \ast L_2 \ast L_1$ \\[0.2in]
		
		{} & {} & \scalebox{0.5}{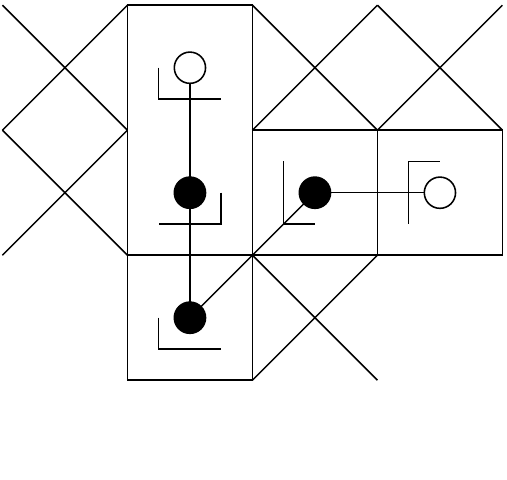} \\
		{} & {} & \vspace{-8mm} (b) $L_3\ast L_2 \ast L_2$ \\[0.2in]
		
		{\scalebox{0.5}{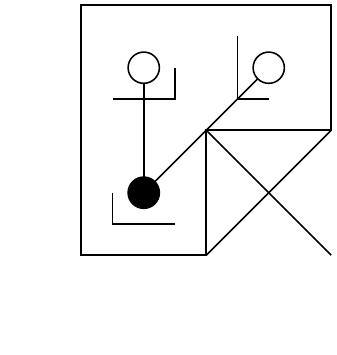}} & {\scalebox{0.5}{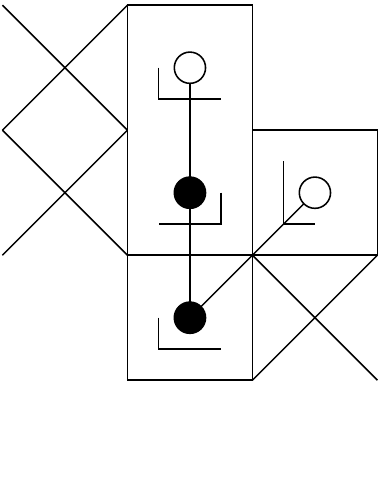}} & \scalebox{0.5}{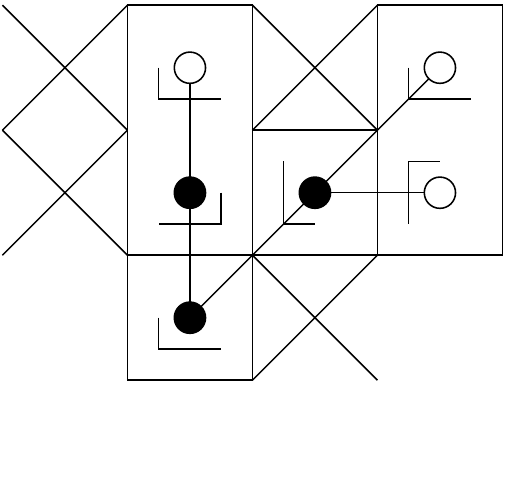} \\
		\vspace{-8mm} $L_3$ & \vspace{-8mm} $L_3\ast L_2$ & \vspace{-8mm} (c) $L_3\ast L_2 \ast L_3$ \\[0.2in]
		
		{} & {} & \scalebox{0.5}{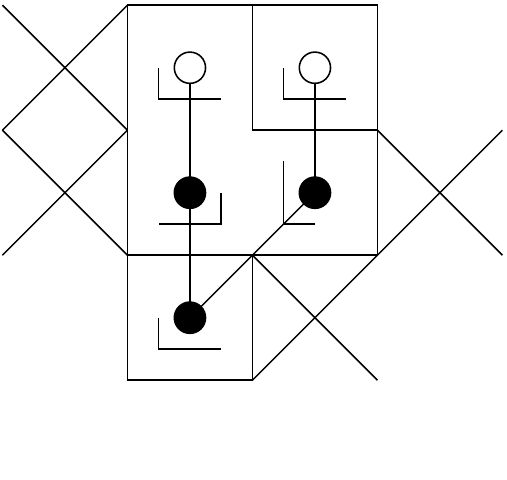} \\
		{} & {} &\vspace{-8mm} (d) $ =L_3 \ast L_2 \ast L_4$ \\[0.2in]
		
		{} & {} & \scalebox{0.5}{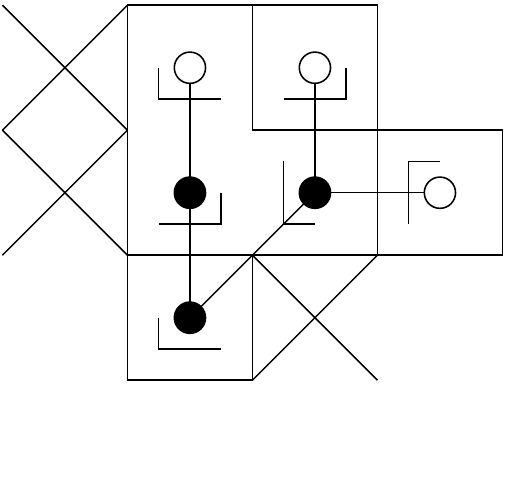} \\
		{} & {} & \vspace{-8mm} (e) $L_3\ast L_2\ast L_5$ \\
	\end{tabular}
	\caption{Concatenation of $L_2$ and $L_3$.}
	\label{fig:l2l3}
\end{figure}
\begin{figure}
	\begin{tabular}{>{\RaggedRight}p{3cm}>{\RaggedRight}p{3cm}>{\RaggedRight}p{3cm}}
		{} & {} & \scalebox{0.5}{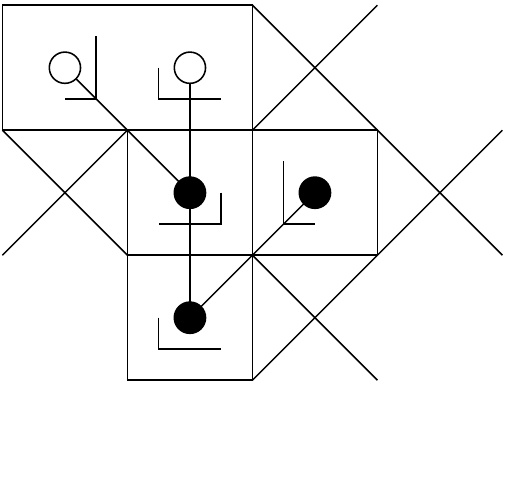} \\
		{} & {} &\vspace{-8mm} (f) $L_3\ast L_3\ast L_1$ \\[0.2in]
		
		{} & {} & \scalebox{0.5}{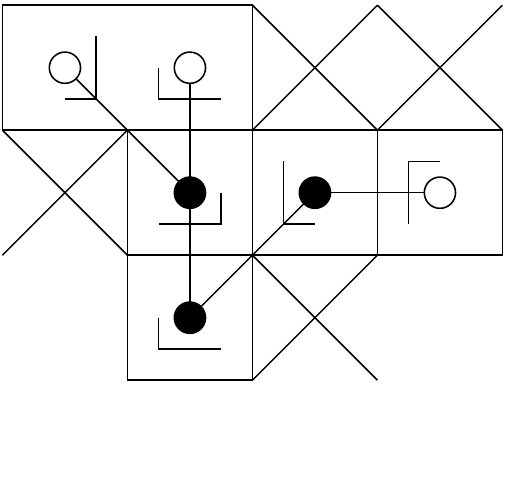} \\
		{} & {} &\vspace{-8mm} (g) $L_3 \ast L_3\ast L_2$ \\[0.2in]
		
		{\scalebox{0.5}{\input{figures/L3.pdf_t}}} & {\scalebox{0.5}{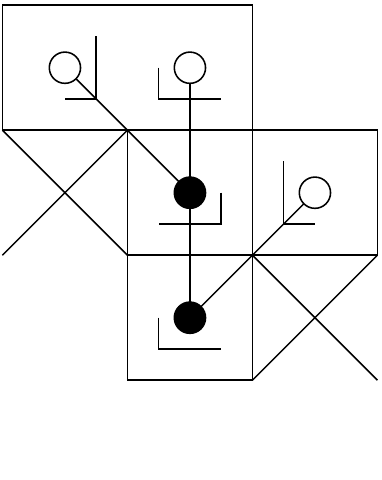}} & \scalebox{0.5}{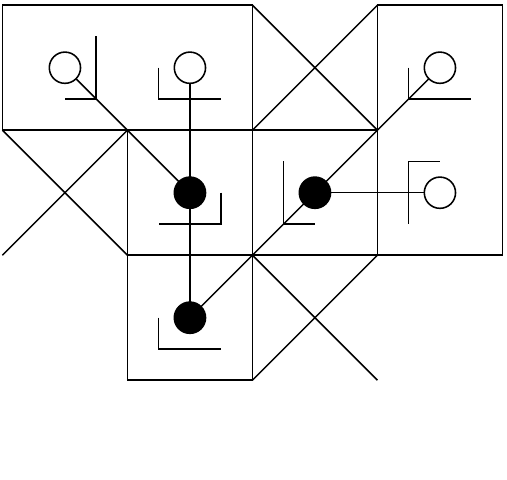} \\
		\vspace{-8mm} $L_3$ & \vspace{-8mm} $L_3\ast L_3$ & \vspace{-8mm} (h) $L_3\ast L_3 \ast L_3$ \\[0.2in]
		
		{} & {} & \scalebox{0.5}{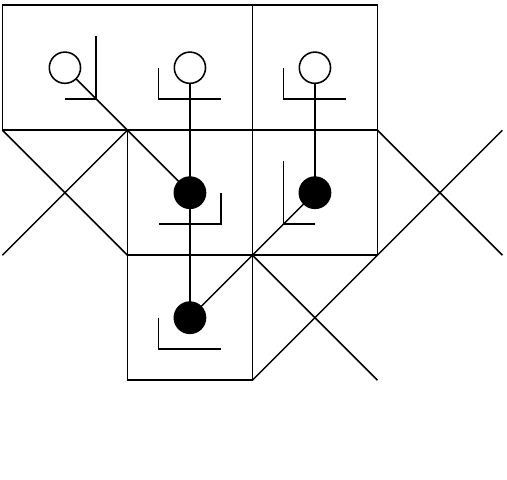} \\
		{} & {} &\vspace{-8mm} (i) $L_3\ast L_3 \ast L_4$ \\[0.2in]

		{} & {} & \scalebox{0.5}{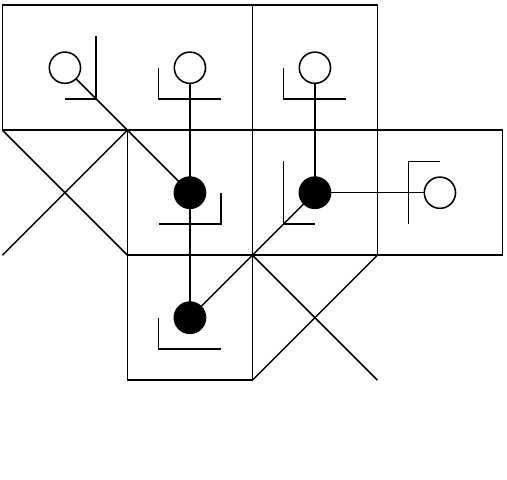} \\
		{} & {} & \vspace{-8mm} (j) $L_3\ast L_3 \ast L_5$ \\
	\end{tabular}
	
	\caption{Concatenation of $L_3$ and $L_3$.}
	\label{fig:l3l3}
\end{figure}
The twigs~$L_4$ and~$L_5$ cannot be concatenated to any of the
configurations (a--j) (Figures~\ref{fig:l2l3} and~\ref{fig:l3l3}) since this would violate the first part of
condition~\cast, that is, a cell of~$L_4$ or~$L_5$ would overlap with
an occupied cell of the configuration or a cell marked as forbidden.
This results in $2{\cdot}10=20$ more invalid sequences.
Thus, we obtain that the number of twigs with four dead cells is
\[
|\LC|^4 - \sum_{i=1}^{9} |S_i| - 2 \cdot 10 =
625 - 125 - 2 \cdot 25 - 6 \cdot 5 - 20 = 400.
\]
Hence, adding items (1--3) above, we obtain that
$|\C_4| = 400 {+} 1 {+} 2 {+} 6 = 409$, which is the number provided
for~$|\C_4|$ by Klarner and Rivest as well (see Table~\ref{tab:results_2d}).
On the other hand, restricting the construction of~$\C_4$ further with
the second part of condition~\cast\ would imply that, for example,
concatenating~$L_1$ to configuration~(d) is invalid, which would
decrease the size of~$\C_4$.
Therefore, we conclude that Klarner and Rivest most probably did not use the
second part of condition~\cast.

\subsection*{Argument 2}

Consider the pentomino~$P$ shown in Figure~\ref{fig:pentomino}.
\begin{figure}
	\centering
	\includegraphics[scale=0.5]{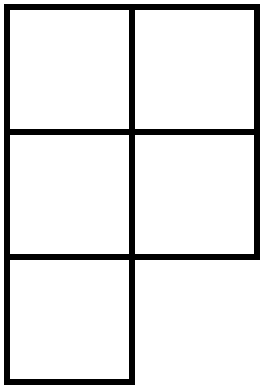}
	\caption{A pentomino}
	\label{fig:pentomino}
\end{figure}
The spanning tree of~$P$ corresponds to the
sequence~$L_3 {\ast} L_2 {\ast} L_4 {\ast} L_1 {\ast} L_1$ (which is
equivalent to concatenating~$L_1$ twice to configuration~(d) in
Figure~\ref{fig:l2l3}).
In terms of elements of~$\C_4$, this is a sequence of length two:
its first element~$T$ is the twig corresponding to adding~$L_1$ to
configuration~(d), and its second element is~$L_1$.
% [[[I SUGGEST TO DRAW THIS IN FIGURE~\ref{fig:pentomino}]]]
% MIRA: I THINK THIS IS NOT TRULY NEEDED, BUT LET'S SEE
By definition, the two twigs~$T$ and~$L_1$ belong to~$\C_4$.
However, if we were to use the second part of condition~\cast,
$T$~would be discarded as an element of~$\C_4$.
In such a situation, $\C_4$ would not be a \emph{complete} set of building
blocks for polyominoes, and~$P$ would have no corresponding sequence of
weight~$x^5 y^5$ of elements of~$\C_4$.
Therefore, the method would have failed to provide an upper bound
on~$\lambda_2$ if the second part of condition~\cast~had been used,
as some polyominoes (such as~$P$) would be overlooked.

\section{The Functions \boldmath{$W_i(x,y)$}}

\label{sec:GFs}

Here are the weight functions~$W_1(x,y),\dots,W_{21}(x,y)$: \\

\noindent
$W_1(x,y)=2x^2y +2xy +y$ \\[-0.12in]

\noindent
$W_2(x,y)= 4x^4y^2 +8x^3y^2 +6x^2y^2 +2xy^2 +y$ \\[-0.12in]

\noindent
$W_3(x,y)=8x^6y^3 +24x^5y^3 +32x^4y^3 +20x^3y^3 +6x^2y^3 +2xy^2 +y$ \\[-0.12in]

\noindent
\parbox{5.5in}{
	$W_4(x,y)=14x^8y^4 +58x^7y^4 +113x^6y^4 -124x^5y^4 +71x^4y^4 +20x^3y^4 +6x^2y^3 +2xy^2 +y$
} \\[-0.06in]

\noindent
\parbox{5.5in}{
	$W_5(x,y)=
	24x^{10}y^5 +124x^9y^5 +317x^8y^5 +494x^7y^5 +483x^6y^5 +261x^5y^5 +71x^4y^5 +20x^3y^4 +6x^2y^3 +2xy^2 +y$
} \\[0.07in]

\noindent
\parbox{5.5in}{
	$W_6(x,y)=
	36x^{12}y^6 +240x^{11}y^6 +772x^{10}y^6 +1550x^9y^6 +2099x^8y^6 +1895x^7y^6 +984x^6y^6  +261x^5y^6 +71x^4y^5 +20x^3y^4 +6x^2y^3 +2xy^2 +y$
} \\[0.07in]

\noindent
\parbox{5.5in}{
	$W_7(x,y)=
	64x^{14}y^7 +468x^{13}y^7 +1750x^{12}y^7 +4221x^{11}y^7 +7177x^{10}y^7 +8795x^9y^7 +7489x^8y^7 +3775x^7y^7 +984x^6y^7 +261x^5y^6 +71x^4y^5 +20x^3y^4 +6x^2y^3 +2xy^2 +y$
} \\[0.07in]

\noindent
\parbox{5.5in}{
	$W_8(x,y)=
	88x^{16}y^8 +780x^{15}y^8 +3487x^{14}y^8 +10135x^{13}y^8 +20921x^{12}y^8 +32015x^{11}y^8 +36517x^{10}y^8 +29738x^9y^8 +14657x^8y^8 +3775x^7y^8 +984x^6y^7 +261x^5y^6 +71x^4y^5 +20x^3y^4 +6x^2y^3 +2xy^2 +y$
} \\[0.07in]

\noindent
\parbox{5.5in}{
	$W_9(x,y)=
	96x^{18}y^9 +1092x^{17}*y^9 +6138x^{16}y^9 +21679x^{15}y^9 +53840x^{14}y^9 +99208x^{13}y^9 +139805x^{12}y^9 +150644x^{11}y^9 +118455x^{10}*y^9 +57394x^9y^9 +14657x^8y^9 +3775x^7y^8 +984x^6y^7 +261x^5y^6 +71x^4y^5 +20x^3y^4 +6x^2y^3 +2xy^2 +y$
} \\[0.07in]

\noindent
\parbox{5.5in}{
	$W_{10}(x,y)=
	64x^{20}y^{10} +1288x^{19}y^{10} +9620x^{18}y^{10} +41940x^{17}y^{10} +124236x^{16}y^{10} +271585x^{15}y^{10}  +455916x^{14}y^{10} +600672x^{13}y^{10} +618318x^{12}y^{10} +472966x^{11}y^{10} +226165x^{10}y^{10} +57394x^9y^{10} +14657x^8y^9 +3775x^7y^8 +984x^6y^7 +261x^5y^6 +71x^4y^5 +20x^3y^4 +6x^2y^3 +2xy^2 +y$
} \\[0.07in]

\noindent
\parbox{5.5in}{
	$W_{11}(x,y)=
	32x^{22}y^{11} + 1560x^{21}y^{11} +15116x^{20}y^{11} + 77222x^{19}y^{11} +265528x^{18}y^{11}
	+671900x^{17}y^{11} + 1315757x^{16}y^{11} +
	2043184x^{15}y^{11} + 2547938x^{14}y^{11} + 2528282x^{13}y^{11} + 1892135x^{12}y^{11} + 895513x^{11}y^{11} + 226165x^{10}y^{11} + 57394x^9y^{10} + 14657x^8y^9 + 3775x^7y^8 + 984x^6y^7 + 261x^5y^6 + 71x^4y^5 + 20x^3y^4 + 6x^2y^3 + 2xy^2 + y$
} \\[0.07in]

\noindent
\parbox{5.5in}{
	$W_{12}(x,y)=
	32x^{24}y^{12} +2448x^{23}y^{12} +24984x^{22}y^{12} +140612x^{21}y^{12} +537148x^{20}y^{12}
	+1535243x^{19}y^{12}+3428784x^{18}y^{12} +6148920x^{17}y^{12} + 8968766x^{16}y^{12} +10700784x^{15}y^{12}
	+10309921x^{14}y^{12} +7582080x^{13}y^{12} +3559132x^{12}y^{12} + 895513x^{11}y^{12} +226165x^{10}y^{11} + 57394x^9y^{10}
	+14657x^8y^9 +3775x^7y^8 +984x^6y^7 +261x^5y^6 +71x^4y^5 +20x^3y^4 +6x^2y^3 +2xy^2 +y$
} \\[0.07in]

\noindent
\parbox{5.5in}{
	$W_{13}(x,y)=
	64x^{26}y^{13} +3376x^{25}y^{13} +39052x^{24}y^{13} +242230x^{23}y^{13} +1029746x^{22}y^{13} +3276965x^{21}y^{13}  +8225862x^{20}y^{13} +16714930x^{19}y^{13} +27959240x^{18}y^{13} +38764654x^{17}y^{13} +44612842x^{16}y^{13} +41963681x^{15}y^{13} +30425691x^{14}y^{13} +14187563x^{13}y^{13} +3559132x^{12}y^{13} +895513x^{11}y^{12} +226165x^{10}y^{11} +57394x^9y^{10} +14657x^8y^9 +3775x^7y^8 +984x^6y^7 +261x^5y^6 +71x^4y^5 +20x^3y^4 +6x^2y^3 +2xy^2 +y$
} \\[0.03in]

\noindent
\parbox{5.5in}{
	$W_{14}(x,y)=3872x^{27}y^{14} +53860x^{26}y^{14} +388828x^{25}y^{14} +1856137x^24y^{14} +6593524x^{23}y^{14} +18410515x^{22}y^{14} +41847658x^{21}y^{14} +78846479x^{20}y^{14} +124566489x^{19}y^{14} +165600553x^{18}y^{14} +184977014x^{17}y^{14} +170581831x^{16}y^{14}  +122243680x^{15}y^{14}+56691193x^{14}y^{14} +14187563x^{13}y^{14}  +3559132x^{12}y^{13} +895513x^{11}y^{12} +226165x^{10}y^{11}+57394x^9y^{10} +14657x^8y^9 +3775x^7y^8 +984x^6y^7 +261x^5y^6 +71x^4y^5 +20x^3y^4+6x^2y^3 +2xy^2 +y$					
} \\[0.02in]

\noindent
\parbox{5.5in}{
	$W_{15}(x,y)=
	4400x^{29}y^{15} +74688x^{28}y^{15} +598128x^{27}y^{15} +3182157x^26y^{15} +12522050x^{25}y^{15} +38772694x^{24}y^{15} +97650143x^{23}y^{15} +204840498x^{22}y^{15} +362613604x^{21}y^{15} +546205155x^{20}y^{15} +701024617x^{19}y^{15} +763765263x^{18}y^{15} +692808602x^{17}y^{15} +491675078x^{16}y^{15} +226975964x^{15}y^{15} +56691193x^{14}y^{15} +14187563x^{13}y^{14} +3559132x^{12}y^{13} +895513x^{11}y^{12} +226165x^{10}y^{11} +57394x^9y^{10} +14657x^8y^9 +3775x^7y^8 +984x^6y^7 +261x^5y^6 +71x^4y^5 +20x^3y^4 +6x^2y^3 +2xy^2 +y$
} \\[0.02in]

\noindent
\parbox{5.5in}{
	$W_{16}(x,y)=6912x^{31}y^{16} +106640x^{30}y^{16} +902716x^{29}y^{16} +5194974x^{28}y^{16} +22502316x^{27}y^{16} +76836395x^{26}y^{16} +213862804x^{25}y^{16} +495599251x^{24}y^{16} +972881530x^{23}y^{16} +1634495588x^{22}y^{16} +2365001740x^{21}y^{16} +2946546711x^{20}y^{16} +3143569000x^{19}y^{16} +2812205702x^{18}y^{16} +1979423214x^{17}y^{16} +910239465x^{16}y^{16} +226975964x^{15}y^{16} +56691193x^{14}y^{15} +14187563x^{13}y^{14} +3559132x^{12}y^{13} +895513x^{11}y^{12} +226165x^{10}y^{11} +57394x^9y^{10} +14657x^8y^9 +3775x^7y^8 +984x^6y^7 +261x^5y^6 +71x^4y^5 +20x^3y^4 +6x^2y^3 +2xy^2 +y$
} \\[0.02in]

\noindent
\parbox{5.5in}{
	$W_{17}(x,y)=7488x^{33}y^{17} +136064x^{32}y^{17} +1260248x^{31}y^{17} +8008422x^{30}y^{17} +38155640x^{29}y^{17} +143730425x^{28}y^{17} +440598289x^{27}y^{17} +1124991218x^{26}y^{17} +2431711692x^{25}y^{17} +4512423641x^{24}y^{17} +7250819239x^{23}y^{17} +10138105194x^{22}y^{17} +12316094597x^{21}y^{17} +12907416312x^{20}y^{17} +11411126315x^{19}y^{17} +7975518589x^{18}y^{17} +3655351652x^{17}y^{17} +910239465x^{16}y^{17} +226975964x^{15}y^{16} +56691193x^{14}y^{15} +14187563x^{13}y^{14} +3559132x^{12}y^{13} +895513x^{11}y^{12} +226165x^{10}y^{11} +57394x^9y^{10} +14657x^8y^9 +3775x^7y^8 +984x^6y^7 +261x^5y^6 +71x^4y^5 +20x^3y^4 +6x^2y^3 +2xy^2 +y$
} \\[0.02in]

\noindent
\parbox{5.5in}{
	$W_{18}(x,y)=6656x^{35}y^{18}
	+144624x^{34}y^{18}
	+1596016x^{33}y^{18}
	+11499558x^{32}y^{18}
	+61231243x^{31}y^{18}
	+254612673x^{30}y^{18}
	+859000063x^{29}y^{18}
	+2406743605x^{28}y^{18}
	+5707780042x^{27}y^{18}
	+11616872329x^{26}y^{18}
	+20533746813x^{25}y^{18}
	+31753025591x^{24}y^{18}
	+43111962291x^{23}y^{18}
	+51254247441x^{22}y^{18}
	+52900498537x^{21}y^{18}
	+46293847405x^{20}y^{18}
	+32158564611x^{19}y^{18}
	+14696358415x^{18}y^{18}
	+3655351652x^{17}y^{18}
	+910239465x^{16}y^{17}
	+226975964x^{15}y^{16}
	+56691193x^{14}y^{15}
	+14187563x^{13}y^{14}
	+3559132x^{12}y^{13}
	+895513x^{11}y^{12}
	+226165x^{10}y^{11}
	+57394x^9y^{10}
	+14657x^8y^9
	+3775x^7y^8
	+984x^6y^7
	+261x^5y^6
	+71x^4y^5
	+20x^3y^4
	+6x^2y^3
	+2xy^2
	+y
	$
} \\[0.02in]

\noindent
\parbox{5.5in}{
	$W_{19}(x,y)=4160x^{37}y^{19}
	+132864x^{36}y^{19}
	+1806784x^{35}y^{19}
	+15379120x^{34}y^{19}
	+92764008x^{33}y^{19}
	+429323463x^{32}y^{19}
	+1592168897x^{31}y^{19}
	+4883567215x^{30}y^{19}
	+12646494275x^{29}y^{19}
	+28110031644x^{28}y^{19}
	+54277687090x^{27}y^{19}
	+91972916089x^{26}y^{19}
	+137593335616x^{25}y^{19}
	+182153289931x^{24}y^{19}
	+212563199986x^{23}y^{19}
	+216507646418x^{22}y^{19}
	+187791339852x^{21}y^{19}
	+129752205674x^{20}y^{19}
	+59145846645x^{19}y^{19}
	+14696358415x^{18}y^{19}
	+3655351652x^{17}y^{18}
	+910239465x^{16}y^{17}
	+226975964x^{15}y^{16}
	+56691193x^{14}y^{15}
	+14187563x^{13}y^{14}
	+3559132x^{12}y^{13}
	+895513x^{11}y^{12}
	+226165x^{10}y^{11}
	+57394x^9y^{10}
	+14657x^8y^9
	+3775x^7y^8
	+984x^6y^7
	+261x^5y^6
	+71x^4y^5
	+20x^3y^4
	+6x^2y^3
	+2xy^2
	+y
	$	
} \\[0.03in]

\noindent
\parbox{5.5in}{
	$W_{20}(x,y)=2176x^{39}y^{20}
	+105984x^{38}y^{20}
	+1912912x^{37}y^{20}
	+19332440x^{36}y^{20}
	+133729026x^{35}y^{20}
	+690201840x^{34}y^{20}
	+2818499583x^{33}y^{20}
	+9436504653x^{32}y^{20}
	+26597729598x^{31}y^{20}
	+64253320144x^{30}y^{20}
	+134933108961x^{29}y^{20}
	+248868733815x^{28}y^{20}
	+406568926658x^{27}y^{20}
	+591104290770x^{26}y^{20}
	+765646007641x^{25}y^{20}
	+879150587033x^{24}y^{20}
	+885155570880x^{23}y^{20}
	+761751496919x^{22}y^{20}
	+523818188901x^{21}y^{20}
	+238239106019x^{20}y^{20}
	+59145846645x^{19}y^{20}
	+14696358415x^{18}y^{19}
	+3655351652x^{17}y^{18}
	+910239465x^{16}y^{17}
	+226975964x^{15}y^{16}
	+56691193x^{14}y^{15}
	+14187563x^{13}y^{14}
	+3559132x^{12}y^{13}
	+895513x^{11}y^{12}
	+226165x^{10}y^{11}
	+57394x^9y^{10}
	+14657x^8y^9
	+3775x^7y^8
	+984x^6y^7
	+261x^5y^6
	+71x^4y^5
	+20x^3y^4
	+6x^2y^3
	+2xy^2
	+y
	$
} \\[0.15in]

\noindent
\parbox{5.5in}{
	$W_{21}(x,y)=
	1152x^{41}y^{21}
	+96256x^{40}y^{21}
	+2056288x^{39}y^{21}
	+24172920x^{38}y^{21}
	+187337816x^{37}y^{21}
	+1070722580x^{36}y^{21}
	+4785628001x^{35}y^{21}
	+17443616434x^{34}y^{21}
	+53301072439x^{33}y^{21}
	+139441849789x^{32}y^{21}
	+316994896043x^{31}y^{21}
	+633630643451x^{30}y^{21}
	+1123186767205x^{29}y^{21}
	+1777867306642x^{28}y^{21}
	+2521701005541x^{27}y^{21}
	+3204911295924x^{26}y^{21}
	+3628247213386x^{25}y^{21}
	+3615763450791x^{24}y^{21}
	+3089960078272x^{23}y^{21}
	+2115755445262x^{22}y^{21}
	+960344267887x^{21}y^{21}
	+238239106019x^{20}y^{21}
	+59145846645x^{19}y^{20}
	+14696358415x^{18}y^{19}
	+3655351652x^{17}y^{18}
	+910239465x^{16}y^{17}
	+226975964x^{15}y^{16}
	+56691193x^{14}y^{15}
	+14187563x^{13}y^{14}
	+3559132x^{12}y^{13}
	+895513x^{11}y^{12}
	+226165x^{10}y^{11}
	+57394x^9y^{10}
	+14657x^8y^9
	+3775x^7y^8
	+984x^6y^7
	+261x^5y^6
	+71x^4y^5
	+20x^3y^4
	+6x^2y^3
	+2xy^2
	+y
	$
}

\section{Three Dimensions}

\begin{table}
	\centering
	\begin{tabular}{crc}
		$i$ & \multicolumn{1}{c}{$|\C^3_i|$} & $1/\sigma_i$\\
		\hline \hline
		1 & 17 & 9.807295572 \\
		2 & 273 & 9.807295567 \\
		3 & 3,745 & 9.701430690 \\
		%3 & 3,471 & 9.553899640 \\
		4 & 51113 & 9.631827042 \\
		% 4 & 37,802 & 9.322364685 \\
		5 & 693,725 & 9.573610717 \\
		% 5 & 413,173 & 9.166209443 \\
		6 & 9,047,959 & 9.517471577 \\
		% 6 & 4,549,717 & 9.050849037 \\
		7 & 114,736,608 & 9.467046484 \\
		% 7 & 48,913,632 & 8.956292810 \\
		8 & 1,428,690,351 & 9.422618063 \\
		% 8 & 514,879,590 & 8.878133389
		9 & 17,538,443,750 & 9.383460515 \\
		%10 &
	\end{tabular}
	\caption{Our results in~3 dimensions}
	\label{tab:OURresults3D}
\end{table}

\label{app:3d}

The following is the weight function for~$i=8$, from which we computed the upper
bound~$\lambda_3 \leq 9.3835$. \\

\noindent
\parbox{5.5in}{
$W_8(x,y)=
y +4xy^2 +23x^2y^3 +150x^3y^4 +1051x^4y^5 +7661x^5y^6 +57337x^6y^7 +437050x^7y^8 +3376485x^8y^9 +26352274x^9y^9 +108757201x^{10}y^9 +306714778x^{11}y^9 +674917794x^{12}y^9 +1222175063x^{13}y^9 +1866911075x^{14}y^9 +2434995919x^{15}y^9 +2728046412x^{16}y^9 +2631637304x^{17}y^9 +2185885771x^{18}y^9 +1560584567x^{19}y^9 +954538066x^{20}y^9 +497886496x^{21}y^9 +220105634x^{22}y^9 +81810253x^{23}y^9 +25294655x^{24}y^9 +6411687x^{25}y^9 +1305352x^{26}y^9 +207134x^{27}y^9 +24462x^{28}y^9 +1992x^{29}y^9 +97x^{30}y^9 +2x^{31}y^9);
$
}

\section{Maple Code}
\label{app:maple}
Let \texttt{f(x,y)} be a rational two-variable generating function.
Klarner and Rivest~\cite[\S3]{KR73} showed how to obtain the the radius of
convergence of the diagonal of~\texttt{f(x,y)}.
This requires a change of variable in order to apply the residue theorem.
The diagonal function~$f_D(z) = \sum_n l(n,n)z^n$ could then be written as
a sum of residues.
The following is our \texttt{Maple} implementation of this method. \\

%f:=(x,y)->x/(1-y-x^4*y-4*x^3*y-4*x*y-7*x^2*y);
\noindent \texttt{g:=(x,y)->f(x,y)/x;} \\
\texttt{par := g(s,z/s);} \\
\texttt{d := denom(par);} \\
\noindent \texttt{with(Physics):} \\
\texttt{c := Coefficients(d,s,leading);} \\
\texttt{div := Coefficients(c,z,leading);} \\
\texttt{d := d / div;} \\
\texttt{dis := discrim(d,s);} \\
\texttt{sols := fsolve(dis=0,z);} \\
\texttt{maxroot := max(sols);} \\
\texttt{ub := evalf(1/maxroot);}

%~ \\ ~
%[[[DO YOU WANT TO PROVIDE THE GF OF THE DIAGONAL FUNCTION?]]] \\
%MIRA: LET ME THINK ABOUT THAT

\end{document}

%% file: figbox.tex
\def\captionstyle{}
\def\boxcaptionstyle{\raggedright}

\def\maxfigfraction{.6}%        --- max frac of column for marg figure

\newdimen\figboxmargin%         --- margin between text and picture
\newdimen\figboxhang%           --- amount to hang picture into margin
\def\DVIscaling{1}%             --- for dvi2ps
\def\globalscaling{1}%          --- for scaling everything
\def\figuredirectory{./figures}%--- where to find the .ps files 
\let\boxer=\llboxer%            --- default origin in lower left for non-
\def\missingfigure#1{\hbox{Missing figure #1.ps}}

\catcode `\@=11

\newbox\figurebox

\def\figbox{%   [place] box style label caption
\@ifnextchar[{\figboxaux}{\figboxaux[htb]}}

\long\def\figboxaux[#1#2]#3#4#5#6{%     --- want left margin figure?
\writepict{{#3}{#4}{#5}{#6}}%           --- save figure if journalpicts
\setbox\figurebox\hbox{#3}%
\if l#1\tryleftbox{#4}{#5}{#6}%
\else%                                  --- want right margin figure?
\if r#1\tryrightbox{#4}{#5}{#6}%
\else%                                  --- want one or two column figure?
\if *#1\checktwocoloptions#2]{\box\figurebox}{#4*}{#5}{#6}%
\else\tryonecol[#1#2]{#4}{#5}{#6}%
\fi
\fi
\fi\ignorespaces}

\long\def\tryleftbox#1#2#3{%            --- does left margin figure fit?
\ifdim\wd\figurebox>\maxfigfraction\columnwidth \tryonecol[htb]{#1}{#2}{#3}%
\else\leftbox{\captionbox{\box\figurebox}{#1}{#2}{#3}}\fi}

\long\def\tryrightbox#1#2#3{%   --- does right margin figure fit?
\ifdim\wd\figurebox>\maxfigfraction\columnwidth \tryonecol[htb]{#1}{#2}{#3}%
\else\rightbox{\captionbox{\box\figurebox}{#1}{#2}{#3}}\fi}

\def\checktwocoloptions{%       --- make sure twocol has htb or p
\@ifnextchar]{\floatbox[htb}{\floatbox[}}

\long\def\tryonecol[#1]#2#3#4{% --- does one column figure fit?
\ifdim\wd\figurebox>\columnwidth \floatbox[#1]{\box\figurebox}{#2*}{#3}{#4}%
\else\floatbox[#1]{\box\figurebox}{#2}{#3}{#4}\fi}

\long\def\floatbox[#1]#2#3#4#5{%
\begin{#3}[#1]
\hbox to \hsize{\hfil#2\hfil}
\captionandlabel{#3}{#4}{#5}
\end{#3}
}

\long\def\captionbox#1#2#3#4{% box style label {[optshortcap] caption}
\setbox\figurebox\hbox{#1}%
\parbox[t]{\wd\figurebox}{%
\bigskip\box\figurebox
\let\captionstyle=\boxcaptionstyle
\captionandlabel{#2}{#3}{#4}
\bigskip
}}

\def\captionandlabel#1#2#3{%    --- {style} {label} {caption}
\def\testit{#3}%
\ifx\testit\empty\else%         --- if non empty caption
\writecapt{{#1}{#2}{#3}}%       --- save for journal figs
\captypeunstarred#1*.%          --- remove optional stars
\getcaption#3\endc@ption%       --- parse the caption
\def\testit{#2}%                --- check for empty label
\ifx\testit\empty\else\label{#2}\fi
\fi}

\def\captypeunstarred#1*#2.{%   --- remove two-column star from style
\def\@captype{#1}}

\def\getcaption{\@ifnextchar[{\getcaptwo}{\getcapone}}
\long\def\getcapone#1\endc@ption{\caption[#1]%
{\def\baselinestretch{1}\Large\normalsize\captionstyle\ignorespaces #1}}
\long\def\getcaptwo[#1]#2\endc@ption{\caption[#1]%
{\def\baselinestretch{1}\Large\normalsize\captionstyle\ignorespaces #2}}

\newdimen\figboxht
\newcount\figboxn

\newcount\figboxlines%          --- picture height in number of lines
\newdimen\figboxwid%            --- picture width

\newif\ifisleftbox

\long\def\leftbox#1{%
\setbox\figurebox\hbox{#1}\global\isleftboxtrue
\startmarginbox
\vadjust{\smash{\rlap{\hskip\hsize\hskip\figboxhang
\llap{\raise.7\baselineskip\box\figurebox\hskip\rightskip}}}}%
\endmarginbox%
}

\long\def\rightbox#1{%
\setbox\figurebox\hbox{#1}\global\isleftboxfalse
\startmarginbox
\smash{\llap{\raise.7\baselineskip\box\figurebox\hskip\figboxmargin}}%
\endmarginbox%
}

\def\startmarginbox{%
\ifvmode\passpict\let\endmarginbox=\indent
\else\message{WARNING: marginbox in not in vmode}\hfilneg\ \passpict
\let\endmarginbox=\relax\fi
\figboxht=\dp\figurebox%         actual height
\advance\figboxht by 1.3\baselineskip% round up
\vskip.95\figboxht\penalty-300\vskip-.95\figboxht% filbreak trick
\divide\figboxht by\baselineskip%      how many lines?
\global\figboxlines=\figboxht
\global\figboxwid=\wd\figurebox
\global\advance\figboxwid by \figboxmargin
\global\advance\figboxwid by -\figboxhang
\setmypar\noindent}
\def\addlines#1{\global\advance\figboxlines by #1\myparshape}
\def\zerolines{\origpar\global\figboxlines=0\myparshape}

\def\passpict{\par\ifnum\figboxlines>1\vskip\figboxlines\baselineskip
\zerolines\fi}

\def\emptybox#1#2{\hbox to #1{\vbox to #2{\vss}\hss}}

\global\let\origpar=\@@par%     --- Latex's saved version of Tex's \par
\global\let\dopar=\origpar
\global\def\@@par{\dopar}
\@setpar{\dopar}

\def\setmypar{\global\let\dopar=\mypar
\global\prevgraf=0\myparshape}

\def\mypar{\origpar\global\advance\figboxlines by -\prevgraf%
\global\prevgraf=0\myparshape}

\def\myparshape{\relax%
\ifnum\figboxlines>1\theparshape \else
\global\hangindent=0pt\global\hangafter=1%  these sometimes get restored...
\global\let\dopar=\origpar\fi}

\def\theparshape{%
\ifisleftbox\global\hangindent=-\figboxwid 
\else\global\hangindent=\figboxwid \fi
\global\hangafter=-\figboxlines \global\advance\hangafter by 1%
}

\def\definefnum#1{%             --- ugly latex hacking to get the
\def\fnum@figure{Figure \ref{#1}}%
\def\fnum@table{Table \ref{#1}}%
\def\fnum@code{Algorithm \ref{#1}}%
}

\def\writepict#1{}
\def\writecapt#1{}

\def\journalpicts#1{%           global scale factor
\newwrite\pictfile
\newwrite\captfile
\openout\pictfile\jobname.pic
\openout\captfile\jobname.cap
\gdef\writepict##1{\unexpandedwrite\pictfile{\doit##1}}%
\gdef\writecapt##1{\unexpandedwrite\captfile{\doit##1}}%
\global\let\ENDdocument=\enddocument
\gdef\enddocument{\DOjournalpicts{#1}\ENDdocument}
}

\def\DOjournalpicts#1{{%
\def\writepict##1{}\closeout\pictfile
\def\writecapt##1{}\closeout\captfile
\@fileswfalse%                  --- turn off toc, lof, and so on.
\onecolumn
\def\globalscaling{#1}
\def\doit##1##2##3##4{%         --- print a picture on blank page
\figboxaux[t]{\hss##1\hss}{##2}{}{}%
\vspace*{1in}
\definefnum{##3}%               --- use Figure \ref{label} or Table \ref{label}
\captionandlabel{##2}{}{##4}
\clearpage}%
\input\jobname.pic
\def\doit##1##2##3{%            --- print caption
\definefnum{##2}%               --- use Figure \ref{label} or Table \ref{label}
\captionandlabel{##1}{}{##3}}%
\raggedright\let\captionstyle=\raggedright
\def\@makecaption##1##2{##1: ##2\par}%  --- Be sure captions aren't centered
\input\jobname.cap
}}%                             --- to be followed by \enddocument

\def\llboxer#1{\vbox to \figboxht{\vfil\hbox to \figboxwid{#1\hfill}}}
\def\lcboxer#1{\vbox to \figboxht{\vfil\hbox to \figboxwid{\hfill#1\hfill}}}
\def\oldboxer#1{\vbox to \figboxht{\vfil
                      \hbox to \figboxwid{\hfill\llap{#1\hskip4.25in}\hfill}}}
\def\ulboxer#1{\vbox to \figboxht{\hbox to \figboxwid{#1\hfill}\vfil}}
\def\ccboxer#1{\vbox to \figboxht{\vfil
                        \hbox to \figboxwid{\hfill#1\hfill}\vfil}}

{\catcode`\p=12\catcode`\t=12
\gdef\removedimen#1pt{#1}}

\def\defscaled#1#2{#2=\DVIscaling#2%
\xdef#1{\expandafter\removedimen\the#2}}

\def\DVIspace{ }
\newdimen\hscalefactor
\newdimen\vscalefactor

\def\scale#1{\horizscale{#1}\vertscale{#1}}
\def\horizscale#1{\hscalefactor=#1\hscalefactor\figboxht=#1\figboxht}
\def\vertscale#1{\vscalefactor=#1\vscalefactor\figboxwid=#1\figboxwid}

\def\boxps{%
\@ifnextchar[{\boxpsaux}{\boxpsaux[\relax]}}

\def\boxpsaux[#1]#2#3#4#5{%
{\figboxwid#4\figboxht#5\hscalefactor=1pt\vscalefactor=1pt%
\scale{#3}%
\scale{\globalscaling}%
#1%
\defscaled\DVIhscale\hscalefactor
\defscaled\DVIvscale\vscalefactor
\boxer{\includegraphics{\figpsfilename\DVIspace}}}%
}

\newread\Epsffilein
\newif\ifEpsffileok
\newif\ifEpsfbbfound

\newdimen\pspoints
\divide\pspoints by 72

\def\boxeps{%
\@ifnextchar[{\boxepsaux}{\boxepsaux[\relax]}}
\def\boxepsaux[#1]#2#3{%
\gdef\figpsfilename{\figuredirectory/#2.ps}
\openin\Epsffilein=\figuredirectory/#2.ps
\ifeof\Epsffilein 
\gdef\figpsfilename{\figuredirectory/#2.eps}
\openin\Epsffilein=\figuredirectory/#2.eps \fi
\ifeof\Epsffilein\message{I couldn't open \figuredirectory/#2.ps or \figpsfilename}%
\missingfigure{#2}
\else
   {\Epsffileoktrue\Epsfbbfoundfalse
    \catcode`\%=11 \catcode`\\=11
    \catcode`\{=11 \catcode`\}=11
    \catcode`\$=11 \catcode`\^=11
    \catcode`\&=11 \catcode`\#=11
    \catcode`\~=11 \catcode`\_=11
    \loop
       \read\Epsffilein to \Epsffileline
       \ifeof\Epsffilein\Epsffileokfalse\else
          \expandafter\Epsfaux\Epsffileline . .\\%
       \fi
   \ifEpsffileok\repeat
   \ifEpsfbbfound
        \figboxht=\Epsfury\pspoints
        \advance\figboxht by-\Epsflly\pspoints
        \figboxwid=\Epsfurx\pspoints
        \advance\figboxwid by-\Epsfllx\pspoints
   \else
        \message{No bounding box comment in \figpsfilename }%
        \figboxwid=2in\figboxht=1in%
   \fi%
   \immediate\closein\Epsffilein
   \hscalefactor=1pt\vscalefactor=1pt%
   \scale{#3}%
   \scale{\globalscaling}%
   #1%
   \defscaled\DVIhscale\hscalefactor
   \defscaled\DVIvscale\vscalefactor
   \hscalefactor=-\Epsfllx\hscalefactor
   \hscalefactor=1.00375\hscalefactor%  --- 72.27/72 = TeXpoints/PSpoints
   \defscaled\DVIhoffset\hscalefactor
   \vscalefactor=-\Epsflly\vscalefactor
   \vscalefactor=1.00375\vscalefactor%  --- 72.27/72 = TeXpoints/PSpoints
   \defscaled\DVIvoffset\vscalefactor
   \llboxer{\includegraphics{\figpsfilename\DVIspace}} }%
\fi
}%
{\catcode`\%=11 \global\let\Epsfpar=\par
\global\let\Epsfpercent=%\global\def\Epsfbblit{%BoundingBox:}}%
\long\def\Epsfaux#1#2 #3\\{\relax\ifx#1\Epsfpercent
   \def\testit{#2}\ifx\testit\Epsfbblit
      \Epsfsize #3 . . . .\\%
      \global\Epsffileokfalse
      \global\Epsfbbfoundtrue
   \fi\else\ifx#1\Epsfpar\else\global\Epsffileokfalse\fi\fi}%
\def\Epsfsize#1 #2 #3 #4 #5\\{\global\def\Epsfllx{#1}\global\def\Epsflly{#2}%
   \global\def\Epsfurx{#3}\global\def\Epsfury{#4}}%

\catcode`\@=12                  % done with with internals

\def\pic#1;#2;#3;#4\par{\picsc#1;#2;#3;1;#4\par}% file; wd; ht; caption 

\def\picsc#1;#2;#3;#4;#5\par{% file; wd; ht; scale; caption
\figbox[htb]{\boxeps{#1}{#4}%{#2in}{#3in}%
}{figure}{#1}{%
#5}}

\def\mpic#1;#2;#3;#4\par{\mpicsc#1;#2;#3;1;#4\par}% file; wd; ht; caption 

\def\mpicsc#1;#2;#3;#4;#5\par{% file; wd; ht; scale; caption
\figbox[l]{\boxeps{#1}{#4}%{#2in}{#3in}%
}{figure}{#1}{%
#5}}

%% file: figures/twigs.pdf_t
\begin{picture}(0,0)%
\includegraphics[scale=0.21]{figures/twigs.pdf}%
\end{picture}%
\setlength{\unitlength}{700sp}%
\begingroup\makeatletter\ifx\SetFigFont\undefined%
\gdef\SetFigFont#1#2#3#4#5{%
  \reset@font\fontsize{#1}{#2pt}%
  \fontfamily{#3}\fontseries{#4}\fontshape{#5}%
  \selectfont}%
\fi\endgroup%
\begin{picture}(22824,8642)(589,2409)
\put(6400,8000){\makebox(0,0)[lb]{\smash{{\SetFigFont{15}{18}{\familydefault}{\mddefault}{\updefault}{\color[rgb]{0,0,0}$e_1$}%
}}}}
\put(2300,3000){\makebox(0,0)[lb]{\smash{{\SetFigFont{15}{18}{\familydefault}{\mddefault}{\updefault}{\color[rgb]{0,0,0}$e_2$}%
}}}}
\put(6300,3000){\makebox(0,0)[lb]{\smash{{\SetFigFont{15}{18}{\familydefault}{\mddefault}{\updefault}{\color[rgb]{0,0,0}$e_3$}%
}}}}
\put(11000,3000){\makebox(0,0)[lb]{\smash{{\SetFigFont{15}{18}{\familydefault}{\mddefault}{\updefault}{\color[rgb]{0,0,0}$e_4$}%
}}}}
\put(13000,8000){\makebox(0,0)[lb]{\smash{{\SetFigFont{15}{18}{\familydefault}{\mddefault}{\updefault}{\color[rgb]{0,0,0}$e_5$}%
}}}}
\put(16500,8000){\makebox(0,0)[lb]{\smash{{\SetFigFont{15}{18}{\familydefault}{\mddefault}{\updefault}{\color[rgb]{0,0,0}$e_6$}%
}}}}
\put(19500,8000){\makebox(0,0)[lb]{\smash{{\SetFigFont{15}{18}{\familydefault}{\mddefault}{\updefault}{\color[rgb]{0,0,0}$e_7$}%
}}}}
\put(16500,3000){\makebox(0,0)[lb]{\smash{{\SetFigFont{15}{18}{\familydefault}{\mddefault}{\updefault}{\color[rgb]{0,0,0}$e_8$}%
}}}}
\end{picture}%

%% file: figures/LC1.pdf_t
\begin{picture}(0,0)%
\includegraphics{figures/LC1}%
\end{picture}%
\setlength{\unitlength}{3947sp}%
\begingroup\makeatletter\ifx\SetFigFont\undefined%
\gdef\SetFigFont#1#2#3#4#5{%
  \reset@font\fontsize{#1}{#2pt}%
  \fontfamily{#3}\fontseries{#4}\fontshape{#5}%
  \selectfont}%
\fi\endgroup%
\begin{picture}(1824,1824)(-2561,-1573)
\put(-1724,-736){\makebox(0,0)[lb]{\smash{{\SetFigFont{20}{16.8}{\familydefault}{\mddefault}{\updefault}{\color[rgb]{0,0,0}$u$}%
}}}}
\put(-1124,-1336){\makebox(0,0)[lb]{\smash{{\SetFigFont{20}{24.0}{\familydefault}{\mddefault}{\updefault}{\color[rgb]{0,0,0}$\ast$}%
}}}}
\put(-1724,-1336){\makebox(0,0)[lb]{\smash{{\SetFigFont{20}{24.0}{\familydefault}{\mddefault}{\updefault}{\color[rgb]{0,0,0}$\ast$}%
}}}}
\put(-2324,-1336){\makebox(0,0)[lb]{\smash{{\SetFigFont{20}{24.0}{\familydefault}{\mddefault}{\updefault}{\color[rgb]{0,0,0}$\ast$}%
}}}}
\put(-2324,-736){\makebox(0,0)[lb]{\smash{{\SetFigFont{20}{24.0}{\familydefault}{\mddefault}{\updefault}{\color[rgb]{0,0,0}$\ast$}%
}}}}
\put(-1124,-736){\makebox(0,0)[lb]{\smash{{\SetFigFont{20}{16.8}{\familydefault}{\mddefault}{\updefault}{\color[rgb]{0,0,0}$b$}%
}}}}
\put(-1724,-136){\makebox(0,0)[lb]{\smash{{\SetFigFont{20}{16.8}{\familydefault}{\mddefault}{\updefault}{\color[rgb]{0,0,0}$a$}%
}}}}
\end{picture}%

%% file: figures/LC2.pdf_t
\begin{picture}(0,0)%
\includegraphics{figures/LC2}%
\end{picture}%
\setlength{\unitlength}{3947sp}%
\begingroup\makeatletter\ifx\SetFigFont\undefined%
\gdef\SetFigFont#1#2#3#4#5{%
  \reset@font\fontsize{#1}{#2pt}%
  \fontfamily{#3}\fontseries{#4}\fontshape{#5}%
  \selectfont}%
\fi\endgroup%
\begin{picture}(1824,1824)(-161,-1573)
\put( 76,-736){\makebox(0,0)[lb]{\smash{{\SetFigFont{20}{16.8}{\familydefault}{\mddefault}{\updefault}{\color[rgb]{0,0,0}$b$}%
}}}}
\put( 76,-1336){\makebox(0,0)[lb]{\smash{{\SetFigFont{20}{24.0}{\familydefault}{\mddefault}{\updefault}{\color[rgb]{0,0,0}$\ast$}%
}}}}
\put(1276,-736){\makebox(0,0)[lb]{\smash{{\SetFigFont{20}{24.0}{\familydefault}{\mddefault}{\updefault}{\color[rgb]{0,0,0}$\ast$}%
}}}}
\put(1276,-1336){\makebox(0,0)[lb]{\smash{{\SetFigFont{20}{24.0}{\familydefault}{\mddefault}{\updefault}{\color[rgb]{0,0,0}$\ast$}%
}}}}
\put(676,-1336){\makebox(0,0)[lb]{\smash{{\SetFigFont{20}{24.0}{\familydefault}{\mddefault}{\updefault}{\color[rgb]{0,0,0}$\ast$}%
}}}}
\put(676,-736){\makebox(0,0)[lb]{\smash{{\SetFigFont{20}{16.8}{\familydefault}{\mddefault}{\updefault}{\color[rgb]{0,0,0}$u$}%
}}}}
\put(676,-136){\makebox(0,0)[lb]{\smash{{\SetFigFont{20}{16.8}{\familydefault}{\mddefault}{\updefault}{\color[rgb]{0,0,0}$a$}%
}}}}
\end{picture}%

%% file: figures/LC3.pdf_t
\begin{picture}(0,0)%
\includegraphics{figures/LC3}%
\end{picture}%
\setlength{\unitlength}{3947sp}%
\begingroup\makeatletter\ifx\SetFigFont\undefined%
\gdef\SetFigFont#1#2#3#4#5{%
  \reset@font\fontsize{#1}{#2pt}%
  \fontfamily{#3}\fontseries{#4}\fontshape{#5}%
  \selectfont}%
\fi\endgroup%
\begin{picture}(1824,1824)(2239,-2173)
\put(3076,-1936){\makebox(0,0)[lb]{\smash{{\SetFigFont{20}{16.8}{\familydefault}{\mddefault}{\updefault}{\color[rgb]{0,0,0}$b$}%
}}}}
\put(2476,-736){\makebox(0,0)[lb]{\smash{{\SetFigFont{20}{24.0}{\familydefault}{\mddefault}{\updefault}{\color[rgb]{0,0,0}$\ast$}%
}}}}
\put(2476,-1936){\makebox(0,0)[lb]{\smash{{\SetFigFont{20}{24.0}{\familydefault}{\mddefault}{\updefault}{\color[rgb]{0,0,0}$\ast$}%
}}}}
\put(2476,-1336){\makebox(0,0)[lb]{\smash{{\SetFigFont{20}{24.0}{\familydefault}{\mddefault}{\updefault}{\color[rgb]{0,0,0}$\ast$}%
}}}}
\put(3076,-736){\makebox(0,0)[lb]{\smash{{\SetFigFont{20}{24.0}{\familydefault}{\mddefault}{\updefault}{\color[rgb]{0,0,0}$\ast$}%
}}}}
\put(3676,-1336){\makebox(0,0)[lb]{\smash{{\SetFigFont{20}{16.8}{\familydefault}{\mddefault}{\updefault}{\color[rgb]{0,0,0}$a$}%
}}}}
\put(3076,-1336){\makebox(0,0)[lb]{\smash{{\SetFigFont{20}{16.8}{\familydefault}{\mddefault}{\updefault}{\color[rgb]{0,0,0}$u$}%
}}}}
\end{picture}%

%% file: figures/LC4.pdf_t
\begin{picture}(0,0)%
\includegraphics{figures/LC4.pdf}%
\end{picture}%
\setlength{\unitlength}{3947sp}%
\begingroup\makeatletter\ifx\SetFigFont\undefined%
\gdef\SetFigFont#1#2#3#4#5{%
  \reset@font\fontsize{#1}{#2pt}%
  \fontfamily{#3}\fontseries{#4}\fontshape{#5}%
  \selectfont}%
\fi\endgroup%
\begin{picture}(1824,1824)(4639,-1573)
\put(5476,-136){\makebox(0,0)[lb]{\smash{{\SetFigFont{20}{16.8}{\familydefault}{\mddefault}{\updefault}{\color[rgb]{0,0,0}$b$}%
}}}}
\put(4876,-136){\makebox(0,0)[lb]{\smash{{\SetFigFont{20}{24.0}{\familydefault}{\mddefault}{\updefault}{\color[rgb]{0,0,0}$\ast$}%
}}}}
\put(4876,-1336){\makebox(0,0)[lb]{\smash{{\SetFigFont{20}{24.0}{\familydefault}{\mddefault}{\updefault}{\color[rgb]{0,0,0}$\ast$}%
}}}}
\put(5476,-1336){\makebox(0,0)[lb]{\smash{{\SetFigFont{20}{24.0}{\familydefault}{\mddefault}{\updefault}{\color[rgb]{0,0,0}$\ast$}%
}}}}
\put(4876,-736){\makebox(0,0)[lb]{\smash{{\SetFigFont{20}{24.0}{\familydefault}{\mddefault}{\updefault}{\color[rgb]{0,0,0}$\ast$}%
}}}}
\put(5476,-736){\makebox(0,0)[lb]{\smash{{\SetFigFont{20}{16.8}{\familydefault}{\mddefault}{\updefault}{\color[rgb]{0,0,0}$u$}%
}}}}
\put(6076,-736){\makebox(0,0)[lb]{\smash{{\SetFigFont{20}{16.8}{\familydefault}{\mddefault}{\updefault}{\color[rgb]{0,0,0}$a$}%
}}}}
\end{picture}%

%% file: figures/LC5.pdf_t
\begin{picture}(0,0)%
\includegraphics{figures/LC5.pdf}%
\end{picture}%
\setlength{\unitlength}{3947sp}%
\begingroup\makeatletter\ifx\SetFigFont\undefined%
\gdef\SetFigFont#1#2#3#4#5{%
  \reset@font\fontsize{#1}{#2pt}%
  \fontfamily{#3}\fontseries{#4}\fontshape{#5}%
  \selectfont}%
\fi\endgroup%
\begin{picture}(1824,1824)(-2561,-3973)
\put(-2324,-2536){\makebox(0,0)[lb]{\smash{{\SetFigFont{20}{24.0}{\familydefault}{\mddefault}{\updefault}{\color[rgb]{0,0,0}$\ast$}%
}}}}
\put(-1724,-3136){\makebox(0,0)[lb]{\smash{{\SetFigFont{20}{16.8}{\familydefault}{\mddefault}{\updefault}{\color[rgb]{0,0,0}$u$}%
}}}}
\put(-1724,-3736){\makebox(0,0)[lb]{\smash{{\SetFigFont{20}{16.8}{\familydefault}{\mddefault}{\updefault}{\color[rgb]{0,0,0}$a$}%
}}}}
\put(-2324,-3136){\makebox(0,0)[lb]{\smash{{\SetFigFont{20}{16.8}{\familydefault}{\mddefault}{\updefault}{\color[rgb]{0,0,0}$b$}%
}}}}
\put(-1124,-3136){\makebox(0,0)[lb]{\smash{{\SetFigFont{20}{24.0}{\familydefault}{\mddefault}{\updefault}{\color[rgb]{0,0,0}$\ast$}%
}}}}
\put(-1799,-2536){\makebox(0,0)[lb]{\smash{{\SetFigFont{20}{24.0}{\familydefault}{\mddefault}{\updefault}{\color[rgb]{0,0,0}$\ast$}%
}}}}
\put(-1124,-2536){\makebox(0,0)[lb]{\smash{{\SetFigFont{20}{24.0}{\familydefault}{\mddefault}{\updefault}{\color[rgb]{0,0,0}$\ast$}%
}}}}
\end{picture}%

%% file: figures/LC6.pdf_t
\begin{picture}(0,0)%
\includegraphics{figures/LC6}%
\end{picture}%
\setlength{\unitlength}{3947sp}%
\begingroup\makeatletter\ifx\SetFigFont\undefined%
\gdef\SetFigFont#1#2#3#4#5{%
  \reset@font\fontsize{#1}{#2pt}%
  \fontfamily{#3}\fontseries{#4}\fontshape{#5}%
  \selectfont}%
\fi\endgroup%
\begin{picture}(1824,1824)(-161,-3973)
\put(1276,-3136){\makebox(0,0)[lb]{\smash{{\SetFigFont{20}{16.8}{\familydefault}{\mddefault}{\updefault}{\color[rgb]{0,0,0}$b$}%
}}}}
\put(676,-2536){\makebox(0,0)[lb]{\smash{{\SetFigFont{20}{24.0}{\familydefault}{\mddefault}{\updefault}{\color[rgb]{0,0,0}$\ast$}%
}}}}
\put(1276,-2536){\makebox(0,0)[lb]{\smash{{\SetFigFont{20}{24.0}{\familydefault}{\mddefault}{\updefault}{\color[rgb]{0,0,0}$\ast$}%
}}}}
\put( 76,-3136){\makebox(0,0)[lb]{\smash{{\SetFigFont{20}{24.0}{\familydefault}{\mddefault}{\updefault}{\color[rgb]{0,0,0}$\ast$}%
}}}}
\put(676,-3136){\makebox(0,0)[lb]{\smash{{\SetFigFont{20}{16.8}{\familydefault}{\mddefault}{\updefault}{\color[rgb]{0,0,0}$u$}%
}}}}
\put(676,-3736){\makebox(0,0)[lb]{\smash{{\SetFigFont{20}{16.8}{\familydefault}{\mddefault}{\updefault}{\color[rgb]{0,0,0}$a$}%
}}}}
\put( 76,-2536){\makebox(0,0)[lb]{\smash{{\SetFigFont{20}{24.0}{\familydefault}{\mddefault}{\updefault}{\color[rgb]{0,0,0}$\ast$}%
}}}}
\end{picture}%

%% file: figures/LC7.pdf_t
\begin{picture}(0,0)%
\includegraphics{figures/LC7}%
\end{picture}%
\setlength{\unitlength}{3947sp}%
\begingroup\makeatletter\ifx\SetFigFont\undefined%
\gdef\SetFigFont#1#2#3#4#5{%
  \reset@font\fontsize{#1}{#2pt}%
  \fontfamily{#3}\fontseries{#4}\fontshape{#5}%
  \selectfont}%
\fi\endgroup%
\begin{picture}(1824,1824)(1639,-4573)
\put(1876,-3736){\makebox(0,0)[lb]{\smash{{\SetFigFont{20}{16.8}{\familydefault}{\mddefault}{\updefault}{\color[rgb]{0,0,0}$a$}%
}}}}
\put(3076,-3736){\makebox(0,0)[lb]{\smash{{\SetFigFont{20}{24.0}{\familydefault}{\mddefault}{\updefault}{\color[rgb]{0,0,0}$\ast$}%
}}}}
\put(3076,-4336){\makebox(0,0)[lb]{\smash{{\SetFigFont{20}{24.0}{\familydefault}{\mddefault}{\updefault}{\color[rgb]{0,0,0}$\ast$}%
}}}}
\put(2476,-4336){\makebox(0,0)[lb]{\smash{{\SetFigFont{20}{24.0}{\familydefault}{\mddefault}{\updefault}{\color[rgb]{0,0,0}$\ast$}%
}}}}
\put(2476,-3736){\makebox(0,0)[lb]{\smash{{\SetFigFont{20}{16.8}{\familydefault}{\mddefault}{\updefault}{\color[rgb]{0,0,0}$u$}%
}}}}
\put(2476,-3136){\makebox(0,0)[lb]{\smash{{\SetFigFont{20}{16.8}{\familydefault}{\mddefault}{\updefault}{\color[rgb]{0,0,0}$b$}%
}}}}
\put(3076,-3136){\makebox(0,0)[lb]{\smash{{\SetFigFont{20}{24.0}{\familydefault}{\mddefault}{\updefault}{\color[rgb]{0,0,0}$\ast$}%
}}}}
\end{picture}%

%% file: figures/LC8.pdf_t
\begin{picture}(0,0)%
\includegraphics{figures/LC8}%
\end{picture}%
\setlength{\unitlength}{3947sp}%
\begingroup\makeatletter\ifx\SetFigFont\undefined%
\gdef\SetFigFont#1#2#3#4#5{%
  \reset@font\fontsize{#1}{#2pt}%
  \fontfamily{#3}\fontseries{#4}\fontshape{#5}%
  \selectfont}%
\fi\endgroup%
\begin{picture}(1824,1824)(3889,-4573)
\put(4126,-3736){\makebox(0,0)[lb]{\smash{{\SetFigFont{20}{16.8}{\familydefault}{\mddefault}{\updefault}{\color[rgb]{0,0,0}$a$}%
}}}}
\put(5326,-4336){\makebox(0,0)[lb]{\smash{{\SetFigFont{20}{24.0}{\familydefault}{\mddefault}{\updefault}{\color[rgb]{0,0,0}$\ast$}%
}}}}
\put(5326,-3736){\makebox(0,0)[lb]{\smash{{\SetFigFont{20}{24.0}{\familydefault}{\mddefault}{\updefault}{\color[rgb]{0,0,0}$\ast$}%
}}}}
\put(4726,-3136){\makebox(0,0)[lb]{\smash{{\SetFigFont{20}{24.0}{\familydefault}{\mddefault}{\updefault}{\color[rgb]{0,0,0}$\ast$}%
}}}}
\put(4726,-3736){\makebox(0,0)[lb]{\smash{{\SetFigFont{20}{16.8}{\familydefault}{\mddefault}{\updefault}{\color[rgb]{0,0,0}$u$}%
}}}}
\put(4726,-4336){\makebox(0,0)[lb]{\smash{{\SetFigFont{20}{16.8}{\familydefault}{\mddefault}{\updefault}{\color[rgb]{0,0,0}$b$}%
}}}}
\put(5326,-3136){\makebox(0,0)[lb]{\smash{{\SetFigFont{20}{24.0}{\familydefault}{\mddefault}{\updefault}{\color[rgb]{0,0,0}$\ast$}%
}}}}
\end{picture}%

%% file: figures/Ltwigs.pdf_t
\begin{picture}(0,0)%
\includegraphics{figures/Ltwigs.pdf}%
\end{picture}%
\setlength{\unitlength}{3947sp}%
\begingroup\makeatletter\ifx\SetFigFont\undefined%
\gdef\SetFigFont#1#2#3#4#5{%
  \reset@font\fontsize{#1}{#2pt}%
  \fontfamily{#3}\fontseries{#4}\fontshape{#5}%
  \selectfont}%
\fi\endgroup%
\begin{picture}(9402,2071)(-2339,-1820)
\put(6301, 14){\makebox(0,0)[lb]{\smash{{\SetFigFont{15}{18.0}{\familydefault}{\mddefault}{\updefault}{\color[rgb]{0,0,0}$2$}%
}}}}
\put(6076,-1336){\makebox(0,0)[lb]{\smash{{\SetFigFont{20}{24.0}{\familydefault}{\mddefault}{\updefault}{\color[rgb]{0,0,0}$\ast$}%
}}}}
\put(5476,-1336){\makebox(0,0)[lb]{\smash{{\SetFigFont{20}{24.0}{\familydefault}{\mddefault}{\updefault}{\color[rgb]{0,0,0}$\ast$}%
}}}}
\put(5476,-736){\makebox(0,0)[lb]{\smash{{\SetFigFont{20}{24.0}{\familydefault}{\mddefault}{\updefault}{\color[rgb]{0,0,0}$\ast$}%
}}}}
\put(6076,-1711){\makebox(0,0)[lb]{\smash{{\SetFigFont{20}{24.0}{\familydefault}{\mddefault}{\updefault}{\color[rgb]{0,0,0}$L_5$}%
}}}}
\put(6901,-586){\makebox(0,0)[lb]{\smash{{\SetFigFont{15}{18.0}{\familydefault}{\mddefault}{\updefault}{\color[rgb]{0,0,0}$1$}%
}}}}
\put(-1124,-1336){\makebox(0,0)[lb]{\smash{{\SetFigFont{20}{24.0}{\familydefault}{\mddefault}{\updefault}{\color[rgb]{0,0,0}$\ast$}%
}}}}
\put(-1724,-1336){\makebox(0,0)[lb]{\smash{{\SetFigFont{20}{24.0}{\familydefault}{\mddefault}{\updefault}{\color[rgb]{0,0,0}$\ast$}%
}}}}
\put(-2324,-1336){\makebox(0,0)[lb]{\smash{{\SetFigFont{20}{24.0}{\familydefault}{\mddefault}{\updefault}{\color[rgb]{0,0,0}$\ast$}%
}}}}
\put(-2324,-736){\makebox(0,0)[lb]{\smash{{\SetFigFont{20}{24.0}{\familydefault}{\mddefault}{\updefault}{\color[rgb]{0,0,0}$\ast$}%
}}}}
\put(-1724,-1711){\makebox(0,0)[lb]{\smash{{\SetFigFont{25}{30.0}{\familydefault}{\mddefault}{\updefault}{\color[rgb]{0,0,0}$L_1$}%
}}}}
\put(826,-1336){\makebox(0,0)[lb]{\smash{{\SetFigFont{20}{24.0}{\familydefault}{\mddefault}{\updefault}{\color[rgb]{0,0,0}$\ast$}%
}}}}
\put(226,-1336){\makebox(0,0)[lb]{\smash{{\SetFigFont{20}{24.0}{\familydefault}{\mddefault}{\updefault}{\color[rgb]{0,0,0}$\ast$}%
}}}}
\put(-374,-1336){\makebox(0,0)[lb]{\smash{{\SetFigFont{20}{24.0}{\familydefault}{\mddefault}{\updefault}{\color[rgb]{0,0,0}$\ast$}%
}}}}
\put(-374,-736){\makebox(0,0)[lb]{\smash{{\SetFigFont{20}{24.0}{\familydefault}{\mddefault}{\updefault}{\color[rgb]{0,0,0}$\ast$}%
}}}}
\put(226,-1711){\makebox(0,0)[lb]{\smash{{\SetFigFont{25}{30.0}{\familydefault}{\mddefault}{\updefault}{\color[rgb]{0,0,0}$L_2$}%
}}}}
\put(2776,-1336){\makebox(0,0)[lb]{\smash{{\SetFigFont{20}{24.0}{\familydefault}{\mddefault}{\updefault}{\color[rgb]{0,0,0}$\ast$}%
}}}}
\put(2176,-1336){\makebox(0,0)[lb]{\smash{{\SetFigFont{20}{24.0}{\familydefault}{\mddefault}{\updefault}{\color[rgb]{0,0,0}$\ast$}%
}}}}
\put(1576,-1336){\makebox(0,0)[lb]{\smash{{\SetFigFont{20}{24.0}{\familydefault}{\mddefault}{\updefault}{\color[rgb]{0,0,0}$\ast$}%
}}}}
\put(1576,-736){\makebox(0,0)[lb]{\smash{{\SetFigFont{20}{24.0}{\familydefault}{\mddefault}{\updefault}{\color[rgb]{0,0,0}$\ast$}%
}}}}
\put(2176,-1711){\makebox(0,0)[lb]{\smash{{\SetFigFont{25}{30.0}{\familydefault}{\mddefault}{\updefault}{\color[rgb]{0,0,0}$L_3$}%
}}}}
\put(2401, 14){\makebox(0,0)[lb]{\smash{{\SetFigFont{15}{18.0}{\familydefault}{\mddefault}{\updefault}{\color[rgb]{0,0,0}$1$}%
}}}}
\put(3001, 14){\makebox(0,0)[lb]{\smash{{\SetFigFont{15}{18.0}{\familydefault}{\mddefault}{\updefault}{\color[rgb]{0,0,0}$2$}%
}}}}
\put(4726,-1336){\makebox(0,0)[lb]{\smash{{\SetFigFont{20}{24.0}{\familydefault}{\mddefault}{\updefault}{\color[rgb]{0,0,0}$\ast$}%
}}}}
\put(4126,-1336){\makebox(0,0)[lb]{\smash{{\SetFigFont{20}{24.0}{\familydefault}{\mddefault}{\updefault}{\color[rgb]{0,0,0}$\ast$}%
}}}}
\put(3526,-1336){\makebox(0,0)[lb]{\smash{{\SetFigFont{20}{24.0}{\familydefault}{\mddefault}{\updefault}{\color[rgb]{0,0,0}$\ast$}%
}}}}
\put(3526,-736){\makebox(0,0)[lb]{\smash{{\SetFigFont{20}{24.0}{\familydefault}{\mddefault}{\updefault}{\color[rgb]{0,0,0}$\ast$}%
}}}}
\put(4126,-1711){\makebox(0,0)[lb]{\smash{{\SetFigFont{25}{30.0}{\familydefault}{\mddefault}{\updefault}{\color[rgb]{0,0,0}$L_4$}%
}}}}
\put(6676,-1336){\makebox(0,0)[lb]{\smash{{\SetFigFont{20}{24.0}{\familydefault}{\mddefault}{\updefault}{\color[rgb]{0,0,0}$\ast$}%
}}}}
\end{picture}%

%% file: figures/Lcontext.pdf_t
\begin{picture}(0,0)%
\includegraphics{figures/Lcontext}%
\end{picture}%
\setlength{\unitlength}{3947sp}%
\begingroup\makeatletter\ifx\SetFigFont\undefined%
\gdef\SetFigFont#1#2#3#4#5{%
  \reset@font\fontsize{#1}{#2pt}%
  \fontfamily{#3}\fontseries{#4}\fontshape{#5}%
  \selectfont}%
\fi\endgroup%
\begin{picture}(10470,7004)(2221,20174)
\put(10351,20789){\makebox(0,0)[lb]{\smash{{\SetFigFont{100}{100.2}{\familydefault}{\mddefault}{\updefault}{\color[rgb]{0,0,0}\raisebox{-4mm}{\scalebox{2}{$x_2$}}}%
}}}}
\put(11251,23264){\makebox(0,0)[lb]{\smash{{\SetFigFont{41}{49.2}{\familydefault}{\mddefault}{\updefault}{\color[rgb]{0,0,0}\scalebox{2}{$x_1$}}%
}}}}
\put(5701,24014){\makebox(0,0)[lb]{\smash{{\SetFigFont{41}{49.2}{\familydefault}{\mddefault}{\updefault}{\color[rgb]{0,0,0}$o$}%
}}}}
\put(12676,21764){\makebox(0,0)[lb]{\smash{{\SetFigFont{41}{49.2}{\familydefault}{\mddefault}{\updefault}{\color[rgb]{0,0,0}\scalebox{2}{$x_3$}}%
}}}}
\put(5576,21089){\makebox(0,0)[lb]{\smash{{\SetFigFont{41}{49.2}{\rmdefault}{\bfdefault}{\updefault}{\color[rgb]{0,0,0}*}%
}}}}
\put(4696,22589){\makebox(0,0)[lb]{\smash{{\SetFigFont{41}{49.2}{\rmdefault}{\bfdefault}{\updefault}{\color[rgb]{0,0,0}*}%
}}}}
\put(3186,21089){\makebox(0,0)[lb]{\smash{{\SetFigFont{41}{49.2}{\rmdefault}{\bfdefault}{\updefault}{\color[rgb]{0,0,0}*}%
}}}}
\put(4676,20189){\makebox(0,0)[lb]{\smash{{\SetFigFont{41}{49.2}{\rmdefault}{\bfdefault}{\updefault}{\color[rgb]{0,0,0}*}%
}}}}
\put(6496,21989){\makebox(0,0)[lb]{\smash{{\SetFigFont{41}{49.2}{\rmdefault}{\bfdefault}{\updefault}{\color[rgb]{0,0,0}*}%
}}}}
\put(7986,21089){\makebox(0,0)[lb]{\smash{{\SetFigFont{41}{49.2}{\rmdefault}{\bfdefault}{\updefault}{\color[rgb]{0,0,0}*}%
}}}}
\end{picture}%

%% file: figures/single.pdf_t
\begin{picture}(0,0)%
\includegraphics{figures/single.pdf}%
\end{picture}%
\setlength{\unitlength}{3947sp}%
\begingroup\makeatletter\ifx\SetFigFont\undefined%
\gdef\SetFigFont#1#2#3#4#5{%
  \reset@font\fontsize{#1}{#2pt}%
  \fontfamily{#3}\fontseries{#4}\fontshape{#5}%
  \selectfont}%
\fi\endgroup%
\begin{picture}(1230,1087)(-2339,-1436)
\put(-1124,-1336){\makebox(0,0)[lb]{\smash{{\SetFigFont{20}{24.0}{\familydefault}{\mddefault}{\updefault}{\color[rgb]{0,0,0}$\ast$}%
}}}}
\put(-1724,-1336){\makebox(0,0)[lb]{\smash{{\SetFigFont{20}{24.0}{\familydefault}{\mddefault}{\updefault}{\color[rgb]{0,0,0}$\ast$}%
}}}}
\put(-2324,-1336){\makebox(0,0)[lb]{\smash{{\SetFigFont{20}{24.0}{\familydefault}{\mddefault}{\updefault}{\color[rgb]{0,0,0}$\ast$}%
}}}}
\put(-2324,-736){\makebox(0,0)[lb]{\smash{{\SetFigFont{20}{24.0}{\familydefault}{\mddefault}{\updefault}{\color[rgb]{0,0,0}$\ast$}%
}}}}
\end{picture}%

%% file: figures/branching.pdf_t
\begin{picture}(0,0)%
\includegraphics{figures/branching.pdf}%
\end{picture}%
\setlength{\unitlength}{3947sp}%
\begingroup\makeatletter\ifx\SetFigFont\undefined%
\gdef\SetFigFont#1#2#3#4#5{%
  \reset@font\fontsize{#1}{#2pt}%
  \fontfamily{#3}\fontseries{#4}\fontshape{#5}%
  \selectfont}%
\fi\endgroup%
\begin{picture}(21804,6598)(-4289,-2672)
\put(-3824,614){\makebox(0,0)[lb]{\smash{{\SetFigFont{29}{34.8}{\familydefault}{\mddefault}{\updefault}{\color[rgb]{0,0,0}$T_1=L_1$}%
}}}}
\put(6826,2939){\makebox(0,0)[lb]{\smash{{\SetFigFont{20}{24.0}{\familydefault}{\mddefault}{\updefault}{\color[rgb]{0,0,0}$\ast$}%
}}}}
\put(6226,2939){\makebox(0,0)[lb]{\smash{{\SetFigFont{20}{24.0}{\familydefault}{\mddefault}{\updefault}{\color[rgb]{0,0,0}$\ast$}%
}}}}
\put(5626,2939){\makebox(0,0)[lb]{\smash{{\SetFigFont{20}{24.0}{\familydefault}{\mddefault}{\updefault}{\color[rgb]{0,0,0}$\ast$}%
}}}}
\put(5626,3539){\makebox(0,0)[lb]{\smash{{\SetFigFont{20}{24.0}{\familydefault}{\mddefault}{\updefault}{\color[rgb]{0,0,0}$\ast$}%
}}}}
\put(-524,-2536){\makebox(0,0)[lb]{\smash{{\SetFigFont{25}{30.0}{\familydefault}{\mddefault}{\updefault}{\color[rgb]{0,0,0}$T_{2,1}$}%
}}}}
\put(376,-2536){\makebox(0,0)[lb]{\smash{{\SetFigFont{25}{30.0}{\familydefault}{\mddefault}{\updefault}{\color[rgb]{0,0,0}$T_{2,2}$}%
}}}}
\put(1276,-2536){\makebox(0,0)[lb]{\smash{{\SetFigFont{25}{30.0}{\familydefault}{\mddefault}{\updefault}{\color[rgb]{0,0,0}$T_{2,3}$}%
}}}}
\put(2251,-2536){\makebox(0,0)[lb]{\smash{{\SetFigFont{25}{30.0}{\familydefault}{\mddefault}{\updefault}{\color[rgb]{0,0,0}$T_{2,4}$}%
}}}}
\put(3151,-2536){\makebox(0,0)[lb]{\smash{{\SetFigFont{25}{30.0}{\familydefault}{\mddefault}{\updefault}{\color[rgb]{0,0,0}$T_{2,5}$}%
}}}}
\put(4201,-2536){\makebox(0,0)[lb]{\smash{{\SetFigFont{25}{30.0}{\familydefault}{\mddefault}{\updefault}{\color[rgb]{0,0,0}$T_{3,1}$}%
}}}}
\put(5101,-2536){\makebox(0,0)[lb]{\smash{{\SetFigFont{25}{30.0}{\familydefault}{\mddefault}{\updefault}{\color[rgb]{0,0,0}$T_{3,2}$}%
}}}}
\put(6001,-2536){\makebox(0,0)[lb]{\smash{{\SetFigFont{25}{30.0}{\familydefault}{\mddefault}{\updefault}{\color[rgb]{0,0,0}$T_{3,3}$}%
}}}}
\put(6901,-2536){\makebox(0,0)[lb]{\smash{{\SetFigFont{25}{30.0}{\familydefault}{\mddefault}{\updefault}{\color[rgb]{0,0,0}$T_{3,4}$}%
}}}}
\put(7951,-2536){\makebox(0,0)[lb]{\smash{{\SetFigFont{25}{30.0}{\familydefault}{\mddefault}{\updefault}{\color[rgb]{0,0,0}$T_{3,5}$}%
}}}}
\put(9901,-2536){\makebox(0,0)[lb]{\smash{{\SetFigFont{25}{30.0}{\familydefault}{\mddefault}{\updefault}{\color[rgb]{0,0,0}$T_{4,2}$}%
}}}}
\put(9001,-2536){\makebox(0,0)[lb]{\smash{{\SetFigFont{25}{30.0}{\familydefault}{\mddefault}{\updefault}{\color[rgb]{0,0,0}$T_{4,1}$}%
}}}}
\put(10801,-2536){\makebox(0,0)[lb]{\smash{{\SetFigFont{25}{30.0}{\familydefault}{\mddefault}{\updefault}{\color[rgb]{0,0,0}$T_{4,3}$}%
}}}}
\put(11776,-2536){\makebox(0,0)[lb]{\smash{{\SetFigFont{25}{30.0}{\familydefault}{\mddefault}{\updefault}{\color[rgb]{0,0,0}$T_{4,4}$}%
}}}}
\put(12751,-2536){\makebox(0,0)[lb]{\smash{{\SetFigFont{25}{30.0}{\familydefault}{\mddefault}{\updefault}{\color[rgb]{0,0,0}$T_{4,5}$}%
}}}}
\put(14476,-2536){\makebox(0,0)[lb]{\smash{{\SetFigFont{25}{30.0}{\familydefault}{\mddefault}{\updefault}{\color[rgb]{0,0,0}$T_{5,2}$}%
}}}}
\put(13576,-2536){\makebox(0,0)[lb]{\smash{{\SetFigFont{25}{30.0}{\familydefault}{\mddefault}{\updefault}{\color[rgb]{0,0,0}$T_{5,1}$}%
}}}}
\put(16276,-2536){\makebox(0,0)[lb]{\smash{{\SetFigFont{25}{30.0}{\familydefault}{\mddefault}{\updefault}{\color[rgb]{0,0,0}$T_{5,4}$}%
}}}}
\put(15376,-2536){\makebox(0,0)[lb]{\smash{{\SetFigFont{25}{30.0}{\familydefault}{\mddefault}{\updefault}{\color[rgb]{0,0,0}$T_{5,3}$}%
}}}}
\put(11026,614){\makebox(0,0)[lb]{\smash{{\SetFigFont{29}{34.8}{\familydefault}{\mddefault}{\updefault}{\color[rgb]{0,0,0}$T_4$}%
}}}}
\put(1351,614){\makebox(0,0)[lb]{\smash{{\SetFigFont{29}{34.8}{\familydefault}{\mddefault}{\updefault}{\color[rgb]{0,0,0}$T_2$}%
}}}}
\put(17251,-2536){\makebox(0,0)[lb]{\smash{{\SetFigFont{25}{30.0}{\familydefault}{\mddefault}{\updefault}{\color[rgb]{0,0,0}$T_{5,5}$}%
}}}}
\put(5551,614){\makebox(0,0)[lb]{\smash{{\SetFigFont{29}{34.8}{\familydefault}{\mddefault}{\updefault}{\color[rgb]{0,0,0}$T_3$}%
}}}}
\put(15451,614){\makebox(0,0)[lb]{\smash{{\SetFigFont{29}{34.8}{\familydefault}{\mddefault}{\updefault}{\color[rgb]{0,0,0}$T_5$}%
}}}}
\put(-3074,-1636){\makebox(0,0)[lb]{\smash{{\SetFigFont{20}{24.0}{\familydefault}{\mddefault}{\updefault}{\color[rgb]{0,0,0}$\ast$}%
}}}}
\put(-3674,-1636){\makebox(0,0)[lb]{\smash{{\SetFigFont{20}{24.0}{\familydefault}{\mddefault}{\updefault}{\color[rgb]{0,0,0}$\ast$}%
}}}}
\put(-4274,-1636){\makebox(0,0)[lb]{\smash{{\SetFigFont{20}{24.0}{\familydefault}{\mddefault}{\updefault}{\color[rgb]{0,0,0}$\ast$}%
}}}}
\put(-4274,-1036){\makebox(0,0)[lb]{\smash{{\SetFigFont{20}{24.0}{\familydefault}{\mddefault}{\updefault}{\color[rgb]{0,0,0}$\ast$}%
}}}}
\end{picture}%

%% file: figures/L1L2.pdf_t
\begin{picture}(0,0)%
\includegraphics{figures/L1L2.pdf}%
\end{picture}%
\setlength{\unitlength}{3947sp}%
\begingroup\makeatletter\ifx\SetFigFont\undefined%
\gdef\SetFigFont#1#2#3#4#5{%
  \reset@font\fontsize{#1}{#2pt}%
  \fontfamily{#3}\fontseries{#4}\fontshape{#5}%
  \selectfont}%
\fi\endgroup%
\begin{picture}(1824,2287)(-11,-3836)
\put(226,-3136){\makebox(0,0)[lb]{\smash{{\SetFigFont{20}{24.0}{\familydefault}{\mddefault}{\updefault}{\color[rgb]{0,0,0}$\ast$}%
}}}}
\put(1426,-3736){\makebox(0,0)[lb]{\smash{{\SetFigFont{20}{24.0}{\familydefault}{\mddefault}{\updefault}{\color[rgb]{0,0,0}$\ast$}%
}}}}
\put(826,-3736){\makebox(0,0)[lb]{\smash{{\SetFigFont{20}{24.0}{\familydefault}{\mddefault}{\updefault}{\color[rgb]{0,0,0}$\ast$}%
}}}}
\put(226,-3736){\makebox(0,0)[lb]{\smash{{\SetFigFont{20}{24.0}{\familydefault}{\mddefault}{\updefault}{\color[rgb]{0,0,0}$\ast$}%
}}}}
\end{picture}%

%% file: figures/L1L4.pdf_t
\begin{picture}(0,0)%
\includegraphics{figures/L1L4.pdf}%
\end{picture}%
\setlength{\unitlength}{3947sp}%
\begingroup\makeatletter\ifx\SetFigFont\undefined%
\gdef\SetFigFont#1#2#3#4#5{%
  \reset@font\fontsize{#1}{#2pt}%
  \fontfamily{#3}\fontseries{#4}\fontshape{#5}%
  \selectfont}%
\fi\endgroup%
\begin{picture}(2202,1687)(1411,-3236)
\put(1426,-2536){\makebox(0,0)[lb]{\smash{{\SetFigFont{20}{24.0}{\familydefault}{\mddefault}{\updefault}{\color[rgb]{0,0,0}$\ast$}%
}}}}
\put(2626,-3136){\makebox(0,0)[lb]{\smash{{\SetFigFont{20}{24.0}{\familydefault}{\mddefault}{\updefault}{\color[rgb]{0,0,0}$\ast$}%
}}}}
\put(2026,-3136){\makebox(0,0)[lb]{\smash{{\SetFigFont{20}{24.0}{\familydefault}{\mddefault}{\updefault}{\color[rgb]{0,0,0}$\ast$}%
}}}}
\put(1426,-3136){\makebox(0,0)[lb]{\smash{{\SetFigFont{20}{24.0}{\familydefault}{\mddefault}{\updefault}{\color[rgb]{0,0,0}$\ast$}%
}}}}
\end{picture}%

%% file: figures/L3L2L1.pdf_t
\begin{picture}(0,0)%
\includegraphics{figures/L3L2L1.pdf}%
\end{picture}%
\setlength{\unitlength}{3947sp}%
\begingroup\makeatletter\ifx\SetFigFont\undefined%
\gdef\SetFigFont#1#2#3#4#5{%
  \reset@font\fontsize{#1}{#2pt}%
  \fontfamily{#3}\fontseries{#4}\fontshape{#5}%
  \selectfont}%
\fi\endgroup%
\begin{picture}(2424,2287)(-611,-12836)
\put(451,-10786){\makebox(0,0)[lb]{\smash{{\SetFigFont{12}{14.4}{\familydefault}{\mddefault}{\updefault}{\color[rgb]{0,0,0}$1$}%
}}}}
\end{picture}%

%% file: figures/L3L2L2.pdf_t
\begin{picture}(0,0)%
\includegraphics{figures/L3L2L2.pdf}%
\end{picture}%
\setlength{\unitlength}{3947sp}%
\begingroup\makeatletter\ifx\SetFigFont\undefined%
\gdef\SetFigFont#1#2#3#4#5{%
  \reset@font\fontsize{#1}{#2pt}%
  \fontfamily{#3}\fontseries{#4}\fontshape{#5}%
  \selectfont}%
\fi\endgroup%
\begin{picture}(2424,2287)(-611,-12836)
\put(451,-10786){\makebox(0,0)[lb]{\smash{{\SetFigFont{12}{14.4}{\familydefault}{\mddefault}{\updefault}{\color[rgb]{0,0,0}$1$}%
}}}}
\put(1651,-11386){\makebox(0,0)[lb]{\smash{{\SetFigFont{12}{14.4}{\familydefault}{\mddefault}{\updefault}{\color[rgb]{0,0,0}$2$}%
}}}}
\end{picture}%

%% file: figures/L3.pdf_t
\begin{picture}(0,0)%
\includegraphics{figures/L3}%
\end{picture}%
\setlength{\unitlength}{3947sp}%
\begingroup\makeatletter\ifx\SetFigFont\undefined%
\gdef\SetFigFont#1#2#3#4#5{%
  \reset@font\fontsize{#1}{#2pt}%
  \fontfamily{#3}\fontseries{#4}\fontshape{#5}%
  \selectfont}%
\fi\endgroup%
\begin{picture}(1602,1687)(-3089,-12536)
\put(-1649,-11086){\makebox(0,0)[lb]{\smash{{\SetFigFont{12}{14.4}{\familydefault}{\mddefault}{\updefault}{\color[rgb]{0,0,0}$2$}%
}}}}
\put(-2249,-11086){\makebox(0,0)[lb]{\smash{{\SetFigFont{12}{14.4}{\familydefault}{\mddefault}{\updefault}{\color[rgb]{0,0,0}$1$}%
}}}}
\end{picture}%

%% file: figures/L3L2.pdf_t
\begin{picture}(0,0)%
\includegraphics{figures/L3L2}%
\end{picture}%
\setlength{\unitlength}{3947sp}%
\begingroup\makeatletter\ifx\SetFigFont\undefined%
\gdef\SetFigFont#1#2#3#4#5{%
  \reset@font\fontsize{#1}{#2pt}%
  \fontfamily{#3}\fontseries{#4}\fontshape{#5}%
  \selectfont}%
\fi\endgroup%
\begin{picture}(1824,2287)(-611,-12836)
\put(451,-10786){\makebox(0,0)[lb]{\smash{{\SetFigFont{12}{14.4}{\familydefault}{\mddefault}{\updefault}{\color[rgb]{0,0,0}$2$}%
}}}}
\put(1051,-11386){\makebox(0,0)[lb]{\smash{{\SetFigFont{12}{14.4}{\familydefault}{\mddefault}{\updefault}{\color[rgb]{0,0,0}$1$}%
}}}}
\end{picture}%

%% file: figures/L3L2L3.pdf_t
\begin{picture}(0,0)%
\includegraphics{figures/L3L2L3.pdf}%
\end{picture}%
\setlength{\unitlength}{3947sp}%
\begingroup\makeatletter\ifx\SetFigFont\undefined%
\gdef\SetFigFont#1#2#3#4#5{%
  \reset@font\fontsize{#1}{#2pt}%
  \fontfamily{#3}\fontseries{#4}\fontshape{#5}%
  \selectfont}%
\fi\endgroup%
\begin{picture}(2424,2287)(-611,-12836)
\put(1651,-10786){\makebox(0,0)[lb]{\smash{{\SetFigFont{12}{14.4}{\familydefault}{\mddefault}{\updefault}{\color[rgb]{0,0,0}$3$}%
}}}}
\put(451,-10786){\makebox(0,0)[lb]{\smash{{\SetFigFont{12}{14.4}{\familydefault}{\mddefault}{\updefault}{\color[rgb]{0,0,0}$1$}%
}}}}
\put(1651,-11311){\makebox(0,0)[lb]{\smash{{\SetFigFont{12}{14.4}{\familydefault}{\mddefault}{\updefault}{\color[rgb]{0,0,0}$2$}%
}}}}
\end{picture}%

%% file: figures/L3L2L4.pdf_t
\begin{picture}(0,0)%
\includegraphics{figures/L3L2L4}%
\end{picture}%
\setlength{\unitlength}{3947sp}%
\begingroup\makeatletter\ifx\SetFigFont\undefined%
\gdef\SetFigFont#1#2#3#4#5{%
  \reset@font\fontsize{#1}{#2pt}%
  \fontfamily{#3}\fontseries{#4}\fontshape{#5}%
  \selectfont}%
\fi\endgroup%
\begin{picture}(2424,2287)(3589,-11186)
\put(5251,-9136){\makebox(0,0)[lb]{\smash{{\SetFigFont{12}{14.4}{\familydefault}{\mddefault}{\updefault}{\color[rgb]{0,0,0}$2$}%
}}}}
\put(4651,-9136){\makebox(0,0)[lb]{\smash{{\SetFigFont{12}{14.4}{\familydefault}{\mddefault}{\updefault}{\color[rgb]{0,0,0}$1$}%
}}}}
\end{picture}%

%% file: figures/L3L2L5.pdf_t
\begin{picture}(0,0)%
\includegraphics{figures/L3L2L5}%
\end{picture}%
\setlength{\unitlength}{3947sp}%
\begingroup\makeatletter\ifx\SetFigFont\undefined%
\gdef\SetFigFont#1#2#3#4#5{%
  \reset@font\fontsize{#1}{#2pt}%
  \fontfamily{#3}\fontseries{#4}\fontshape{#5}%
  \selectfont}%
\fi\endgroup%
\begin{picture}(2424,2287)(3589,-14786)
\put(5851,-13336){\makebox(0,0)[lb]{\smash{{\SetFigFont{12}{14.4}{\familydefault}{\mddefault}{\updefault}{\color[rgb]{0,0,0}$3$}%
}}}}
\put(4651,-12736){\makebox(0,0)[lb]{\smash{{\SetFigFont{12}{14.4}{\familydefault}{\mddefault}{\updefault}{\color[rgb]{0,0,0}$1$}%
}}}}
\put(5251,-12736){\makebox(0,0)[lb]{\smash{{\SetFigFont{12}{14.4}{\familydefault}{\mddefault}{\updefault}{\color[rgb]{0,0,0}$2$}%
}}}}
\end{picture}%

%% file: figures/L3L3L1.pdf_t
\begin{picture}(0,0)%
\includegraphics{figures/L3L3L1}%
\end{picture}%
\setlength{\unitlength}{3947sp}%
\begingroup\makeatletter\ifx\SetFigFont\undefined%
\gdef\SetFigFont#1#2#3#4#5{%
  \reset@font\fontsize{#1}{#2pt}%
  \fontfamily{#3}\fontseries{#4}\fontshape{#5}%
  \selectfont}%
\fi\endgroup%
\begin{picture}(2424,2287)(5089,-5336)
\put(6151,-3286){\makebox(0,0)[lb]{\smash{{\SetFigFont{12}{14.4}{\familydefault}{\mddefault}{\updefault}{\color[rgb]{0,0,0}$1$}%
}}}}
\put(5626,-3286){\makebox(0,0)[lb]{\smash{{\SetFigFont{12}{14.4}{\familydefault}{\mddefault}{\updefault}{\color[rgb]{0,0,0}$2$}%
}}}}
\end{picture}%

%% file: figures/L3L3L2.pdf_t
\begin{picture}(0,0)%
\includegraphics{figures/L3L3L2}%
\end{picture}%
\setlength{\unitlength}{3947sp}%
\begingroup\makeatletter\ifx\SetFigFont\undefined%
\gdef\SetFigFont#1#2#3#4#5{%
  \reset@font\fontsize{#1}{#2pt}%
  \fontfamily{#3}\fontseries{#4}\fontshape{#5}%
  \selectfont}%
\fi\endgroup%
\begin{picture}(2424,2287)(5089,-8936)
\put(7351,-7486){\makebox(0,0)[lb]{\smash{{\SetFigFont{12}{14.4}{\familydefault}{\mddefault}{\updefault}{\color[rgb]{0,0,0}$3$}%
}}}}
\put(5626,-6886){\makebox(0,0)[lb]{\smash{{\SetFigFont{12}{14.4}{\familydefault}{\mddefault}{\updefault}{\color[rgb]{0,0,0}$2$}%
}}}}
\put(6151,-6886){\makebox(0,0)[lb]{\smash{{\SetFigFont{12}{14.4}{\familydefault}{\mddefault}{\updefault}{\color[rgb]{0,0,0}$1$}%
}}}}
\end{picture}%

%% file: figures/L3L3.pdf_t
\begin{picture}(0,0)%
\includegraphics{figures/L3L3.pdf}%
\end{picture}%
\setlength{\unitlength}{3947sp}%
\begingroup\makeatletter\ifx\SetFigFont\undefined%
\gdef\SetFigFont#1#2#3#4#5{%
  \reset@font\fontsize{#1}{#2pt}%
  \fontfamily{#3}\fontseries{#4}\fontshape{#5}%
  \selectfont}%
\fi\endgroup%
\begin{picture}(1824,2287)(889,-12536)
\put(1426,-10486){\makebox(0,0)[lb]{\smash{{\SetFigFont{12}{14.4}{\familydefault}{\mddefault}{\updefault}{\color[rgb]{0,0,0}$3$}%
}}}}
\put(1951,-10486){\makebox(0,0)[lb]{\smash{{\SetFigFont{12}{14.4}{\familydefault}{\mddefault}{\updefault}{\color[rgb]{0,0,0}$2$}%
}}}}
\put(2551,-11086){\makebox(0,0)[lb]{\smash{{\SetFigFont{12}{14.4}{\familydefault}{\mddefault}{\updefault}{\color[rgb]{0,0,0}$1$}%
}}}}
\end{picture}%

%% file: figures/L3L3L3.pdf_t
\begin{picture}(0,0)%
\includegraphics{figures/L3L3L3}%
\end{picture}%
\setlength{\unitlength}{3947sp}%
\begingroup\makeatletter\ifx\SetFigFont\undefined%
\gdef\SetFigFont#1#2#3#4#5{%
  \reset@font\fontsize{#1}{#2pt}%
  \fontfamily{#3}\fontseries{#4}\fontshape{#5}%
  \selectfont}%
\fi\endgroup%
\begin{picture}(2424,2287)(5089,-12536)
\put(7351,-10486){\makebox(0,0)[lb]{\smash{{\SetFigFont{12}{14.4}{\familydefault}{\mddefault}{\updefault}{\color[rgb]{0,0,0}$4$}%
}}}}
\put(5701,-10486){\makebox(0,0)[lb]{\smash{{\SetFigFont{12}{14.4}{\familydefault}{\mddefault}{\updefault}{\color[rgb]{0,0,0}$2$}%
}}}}
\put(6151,-10486){\makebox(0,0)[lb]{\smash{{\SetFigFont{12}{14.4}{\familydefault}{\mddefault}{\updefault}{\color[rgb]{0,0,0}$1$}%
}}}}
\put(7351,-11086){\makebox(0,0)[lb]{\smash{{\SetFigFont{12}{14.4}{\familydefault}{\mddefault}{\updefault}{\color[rgb]{0,0,0}$3$}%
}}}}
\end{picture}%

%% file: figures/L3L3L4.pdf_t
\begin{picture}(0,0)%
\includegraphics{figures/L3L3L4}%
\end{picture}%
\setlength{\unitlength}{3947sp}%
\begingroup\makeatletter\ifx\SetFigFont\undefined%
\gdef\SetFigFont#1#2#3#4#5{%
  \reset@font\fontsize{#1}{#2pt}%
  \fontfamily{#3}\fontseries{#4}\fontshape{#5}%
  \selectfont}%
\fi\endgroup%
\begin{picture}(2424,2287)(5089,-16136)
\put(6751,-14086){\makebox(0,0)[lb]{\smash{{\SetFigFont{12}{14.4}{\familydefault}{\mddefault}{\updefault}{\color[rgb]{0,0,0}$3$}%
}}}}
\put(5701,-14086){\makebox(0,0)[lb]{\smash{{\SetFigFont{12}{14.4}{\familydefault}{\mddefault}{\updefault}{\color[rgb]{0,0,0}$2$}%
}}}}
\put(6151,-14086){\makebox(0,0)[lb]{\smash{{\SetFigFont{12}{14.4}{\familydefault}{\mddefault}{\updefault}{\color[rgb]{0,0,0}$1$}%
}}}}
\end{picture}%

%% file: figures/L3L3L5.pdf_t
\begin{picture}(0,0)%
\includegraphics{figures/L3L3L5}%
\end{picture}%
\setlength{\unitlength}{3947sp}%
\begingroup\makeatletter\ifx\SetFigFont\undefined%
\gdef\SetFigFont#1#2#3#4#5{%
  \reset@font\fontsize{#1}{#2pt}%
  \fontfamily{#3}\fontseries{#4}\fontshape{#5}%
  \selectfont}%
\fi\endgroup%
\begin{picture}(2424,2287)(5089,-19736)
\put(7351,-18286){\makebox(0,0)[lb]{\smash{{\SetFigFont{12}{14.4}{\familydefault}{\mddefault}{\updefault}{\color[rgb]{0,0,0}$4$}%
}}}}
\put(5701,-17686){\makebox(0,0)[lb]{\smash{{\SetFigFont{12}{14.4}{\familydefault}{\mddefault}{\updefault}{\color[rgb]{0,0,0}$2$}%
}}}}
\put(6151,-17686){\makebox(0,0)[lb]{\smash{{\SetFigFont{12}{14.4}{\familydefault}{\mddefault}{\updefault}{\color[rgb]{0,0,0}$1$}%
}}}}
\put(6751,-17686){\makebox(0,0)[lb]{\smash{{\SetFigFont{12}{14.4}{\familydefault}{\mddefault}{\updefault}{\color[rgb]{0,0,0}$3$}%
}}}}
\end{picture}%